\documentclass[11pt, letterpaper, twoside]{article}
\usepackage[margin=1in]{geometry}
\usepackage[T1]{fontenc}
\usepackage[inline]{trackchanges}
\usepackage{ulem}

\usepackage[round]{natbib}

\RequirePackage[T1]{fontenc} 
\RequirePackage[tt=false, type1=true]{libertine} 
\RequirePackage[varqu]{zi4} 
\usepackage{newpxmath}

\usepackage{bbm}
\usepackage{tikz}
\usepackage{amsthm}
\usepackage{amsmath}
\usepackage{physics}
\usepackage{graphicx}
\usepackage{setspace}
\usepackage{enumerate}
\usepackage{mathtools}
\usepackage{caption}
\makeatletter
\def\th@plain{%
  \thm@notefont{}
  \itshape 
}
\def\th@definition{%
  \thm@notefont{}
  \normalfont 
}
\makeatother

\newtheorem{theorem}{\textsc{Theorem}}[section]
\newtheorem{proposition}[theorem]{\textsc{Proposition}} 
\newtheorem{lemma}[theorem]{\textsc{Lemma}}

\newtheorem*{definition*}{\textsc{Definition}}
\newtheorem*{example*}{\textsc{Example}}

\theoremstyle{remark}
\newtheorem*{remark*}{Remark}

\usetikzlibrary{calc}

\newcommand{\E}{\mathbb{E}}

\newcommand{\R}{\mathbb{R}}
\newcommand{\F}{\mathcal{F}}

\newcommand{\Cov}{\mathrm{Cov}}

\renewcommand{\d}{\partial}
\renewcommand{\P}{\mathbb{P}}

\makeatletter
\def\blfootnote{\gdef\@thefnmark{}\@footnotetext}
\makeatother

\allowdisplaybreaks

\title{\bfseries{The Paradox of Just-in-Time Liquidity in Decentralized Exchanges: More Providers Can Lead to Less Liquidity} \blfootnote{The authors gratefully acknowledge the support of  Ripple through the University Blockchain  Research Initiative and of the Stellar Development Foundation.}}
\author{Agostino Capponi\footnote{\scriptsize Columbia University (Department of Industrial Enginering and Operations Research). Email: ac3827@columbia.edu.}, \  Ruizhe Jia\footnote{\scriptsize Columbia University (Department of Industrial Enginering and Operations Research). Email: rj2536@columbia.edu.}, \ and Brian Zi Qi Zhu\footnote{\scriptsize Columbia University (Department of Industrial Enginering and Operations Research). Email: bzz2101@columbia.edu.}}
\date{}

\begin{document}

\maketitle

\begin{abstract}
We study Just-in-time (JIT) liquidity provision in blockchain-based decentralized exchanges. A JIT liquidity provider (LP) monitors pending swap orders in public mempools of blockchains to sandwich orders of their choice with liquidity, depositing right before and withdrawing right after the order. Our game-theoretic model with asymmetrically informed agents reveals that a JIT LP's presence does not always enhance liquidity pool depth, as one might expect. While passive LPs face adverse selection by informed arbitrageurs, a JIT LP's ability to detect pending orders for toxic order flow prior to liquidity provision lets them avoid being adversely selected. JIT LPs thus only provide liquidity to uninformed orders and crowd out passive LPs when order volume is not sufficiently elastic to pool depth, possibly reducing overall market liquidity. We show that using a two-tiered fee structure which transfers a part of a JIT LP's fee revenue to passive LPs or allowing for JIT LPs to compete \`{a} la Cournot are potential solutions to mitigate the negative effects of JIT liquidity.

\end{abstract}

\vfill

\section{Introduction}\label{sec:intro}

Market liquidity is essential in financial markets, as it significantly affects the price impact of trades and transaction costs for investors. High market liquidity thus plays a key role in maintaining the stability and efficiency of trading, while low market liquidity results in decreased trading volumes and diminished gains from trade. Consequently, regulatory bodies like the U.S. Securities and Exchange Commission (SEC) and the Bank for International Settlements (BIS) have raised concerns about current market liquidity and the design of existing market structures.

Regulators are constantly in search of innovative solutions to enhance liquidity provision. Advancements in  blockchain technology have spurred significant developments in Decentralized Finance (DeFi), notably in decentralized exchanges (DEXs) that employ automated market makers (AMMs), positioning them at the forefront of financial innovation. One prominent example is the Bank for International Settlements' Project Mariana, which explores the incorporation of AMMs in foreign exchange, especially in tandem with central bank digital currencies (\cite{BIS}). Unlike conventional market structures, AMMs bring unique benefits, such as automated pricing and risk sharing among liquidity providers (\cite{Malinova_Park_2023}). This innovative market model shows great promise in transforming the landscape of liquidity provision within financial markets.

However, a key obstacle to the broader adoption and effectiveness of this innovation is the profitability challenge faced by liquidity providers. Analysis of Uniswap v3 liquidity positions shows that half of these positions are unprofitable (\cite{Loesch21}). Additionally, game-theoretical models indicate that liquidity providers are subject to infrastructure rent extraction, which weakens their incentive to supply liquidity (\cite{AMMAdoption21}). This problem may lead to insufficient market depth in AMMs, resulting in higher transaction costs and lower trading volumes. Recent data from September 2023 indicates a stagnation in trading volume and market depth in the liquidity pools of AMMs (\cite{2023_Report}). Addressing this issue of liquidity depth is vital for AMMs to effectively lower transaction costs.

In response to these challenges, a novel approach known as {\it Just-in-time (JIT) liquidity} has emerged (see \cite{AdamsUniswap} for details). JIT liquidity leverages the unique order execution mechanism of public blockchains, where transactions are processed in discrete blocks and visible in a public memory pool (mempool) before execution. This transparency enables certain liquidity providers (LPs) to preview pending market orders, respond with matching liquidity, earn fees from executing swaps, and then withdraw their liquidity, all within a single block.\footnote{This mechanism resembles that of a sandwich attack analyzed in \cite{CappMEV}, with the difference that the order is now sandwiched between a liquidity supply and a liquidity withdrawal action.} This ability to pre-screen and strategically select which orders to supply liquidity for significantly mitigates adverse selection risk. This feature, in sharp contrast to the lack of transparency in centralized exchanges where market orders are not visible beforehand, highlights the transformative potential of JIT liquidity.

With this new form of liquidity provision, liquidity providers in the DEX ecosystem can be categorized into two types. \textit{Passive} liquidity providers commit their tokens over multiple blocks, lacking choice over the market orders they face in those blocks, and are thereby exposed to risks from market volatility. \textit{Active} providers use the JIT liquidity mechanism, strategically depositing and withdrawing in response to pending market orders. This strategy grants them a {\it second-mover advantage}, allowing them to avoid toxic order flows and pickoff risks. 

Can Just-in-time liquidity enhance overall market liquidity?  At first glance, it would appear to alleviate costs in liquidity provision, particularly by allowing market makers to judiciously choose their participation based on visible order flow.  Moreover, the addition of active liquidity provision to existing passive provision would seemingly indicate a higher amount of market liquidity in the system. We instead argue that this perspective might be overly simplistic. Our analysis, utilizing a game-theoretical framework, suggests that the influence of high-frequency market makers on aggregate liquidity is intricately connected to how uninformed order flow responds to liquidity depth.  

We show that, for tokens with a high elasticity of uninformed order flow with respect to pool depth, active and passive liquidity providers are {\it strategic complements}. Active providers, while capturing a substantial share of transaction fees, indirectly support passive providers by boosting overall trading volume. Their entry reduces transaction costs, attracts more trader order flow, and benefits passive providers, encouraging their continued participation in DEXs. Consequently, this interaction leads to an increase in overall aggregate market liquidity.

In contrast, when the elasticity of uninformed order flow with respect to pool depth is low,   active and passive liquidity providers become {\it strategic substitutes}. In this case, active providers cannot induce a sufficient increase in overall trading volume that offsets the reduced share of fees of passive providers. This imbalance can lead to a decrease in passive providers' engagement in liquidity provision, thus undermining aggregate market liquidity. 

The potential exit of passive liquidity providers poses a significant risk, notably diminishing visible pre-trade liquidity in DEXs. This reduction in the transparency of liquidity, coupled with the unpredictability of additional contributions by active providers, can lead to heightened price impacts and price uncertainty. Such market conditions likely deter uninformed traders from trading, negatively affecting their participation in DEXs and consequently reducing overall gains from trade. 

To optimize the benefits of this novel liquidity provision mechanism in AMM without undermining market liquidity, a reevaluation of the market design is needed. We propose a two-tiered fee system implemented via smart contracts to balance incentives between active and passive liquidity providers. This system mirrors the rebates and reduced costs found in centralized exchanges for designated market makers. Active market makers, identified by their rapid token redemption, would share a portion of their fees with passive providers, who hold tokens longer. This mechanism aims to mitigate the crowding out effect and compensate passive liquidity providers for the risk of always being in the market, ensuring both types of market makers are adequately incentivized. We solve for the optimal fee structure that maximizes gains from trade and ensures an increase in the aggregate market liquidity after the entry of active LPs.

\subsection{Literature Review}

Our study expands the existing body of literature on DEXs and AMMs. Existing literature has so far focused on analyzing the benefits and costs of passive
liquidity provision on DEXs. \cite{AMMAdoption21} argue that, under the current blockchain order execution mechanism, the arbitrage losses incurred by passive liquidity providers are unavoidable and proposes an optimal design of the DEX pricing functions that mitigate these losses and maximize social welfare. \cite{anotherAMM} show that the liquidity invariance pricing function enables front-running, and increases traders' costs. \cite{anotherAMM2} compare centralized and decentralized exchanges, accounting for information asymmetry frictions. \cite{Aoyagi2021CoexistingEP} {investigates market dynamics and trading patterns stemming from the coexistence of centralized and decentralized exchanges.} \cite{dexcexcompare} compare trading costs and price efficiency of centralized and decentralized exchanges, advocating for the growing competitiveness of AMMs. 
\cite{Hasbrouck_Rivera_Saleh_2022} show that setting higher trading fees at a DEX would reduce the price impact of trades, and consequently increase trading volumes. \cite{Milionis_Moallemi_Roughgarden_Zhang} provide a ``Black-Scholes formula'' for AMMs, which quantifies in closed-form the adverse selection costs faced by liquidity providers.  \cite{Hasbrouck_Rivera_Saleh_2023} and \cite{Milionis_Moallemi_Roughgarden_Zhang} quantify the ``optionality'' relinquished by pool deposits, under the assumption of an external price process that
that restricts AMM liquidity providers to fixed durations, preventing position adjustments or liquidity withdrawals.  {A recent study by \cite{Malinova_Park_2023} discuss the potential application of AMMs in equity trading, suggesting AMMs as alternatives to traditional limit order books. 

Relative to the above surveyed studies, our paper breaks new ground by analyzing a novel model of active liquidity provision, which has not been previously explored. We assess whether JIT liquidity can alleviate frictions inherent in the current AMM designs and enhance overall market liquidity. 

Our research is also broadly related with the market microstructure literature, particularly regarding speed heterogeneity and the impact of high-frequency trading (HFT) in traditional centralized exchanges. \cite{Han_Khapko_Kyle_2014} contribute to this discourse by demonstrating that slower market makers, in comparison to HFT entities, exhibit delayed responses in order cancellations, thereby becoming more susceptible to toxic flow. In their paper, they show that the presence of HFTs leads to wider spreads and reduced liquidity, or HFTs might entirely supplant slower market makers without changing the spread, potentially making full crowding out beneficial in certain scenarios on centralized exchanges. Our paper, however, highlights the potential risks associated with the crowding out of low-frequency market makers, a concern that warrants careful consideration in the evolving landscape of DEXs and AMMs. We also refer to the survey paper by \cite{MenkveldARFE} for a discussion on the adverse selection costs imposed by faster market makers on slower market makers.

Our paper also contributes to the body of literature on JIT liquidity. \cite{AdamsUniswap} describes the mechanism of JIT, and provides statistics on the success and failure rates of JIT transactions, along with information about the number of accounts involved. \cite{AdamsUniswapSlippage} conducts an in-depth study of various types of slippage on Uniswap, identifying that an average million dollar-valued order has a price improvement of 0.6 bps due to JIT liquidity.  
\cite{DemystifyingJIT} conducts an empirical study of over 30,000 JIT liquidity provision instances, analyzing metrics such as the liquidity-to-swap volume ratio, share dilution for passive LPs, and price improvement for traders.

\subsection{Institutional Details}

A JIT liquidity provider would execute the following actions: (i) spot a pending uninformed swap order, (ii) add liquidity to the pool, (iii) letting the swap execute, and (iv) remove the liquidity from the pool. This sequence of transactions would take place in a single block. Typically, to submit a transaction, an LP must broadcast it in the peer-to-peer blockchain network and bid a priority fee. Once the transaction is received by the blockchain validators\footnote{Validators ensure the authenticity of transactions, incorporate them into new blocks, and then add these blocks to the chain, thereby earning the priority fees corresponding to those transactions.}, it becomes a pending order in their mempools, visibile to all users of that blockchain. At discrete times, one validator is chosen to append the next block to the chain. As block space is limited, the validator will execute orders in her mempool in descending order of priority fees. In practice, JIT LPs pay fees to the block validators directly in return for their liquidity sandwich being included, bypassing the transaction fee auction.

For example, suppose there is \$10,000 of liquidity in the pool concentrated in a \$0.01 price range. A liquidity provider sees that a user has submitted a swap of \$5,000 worth of Ethereum tokens for Bitcoin tokens, which incurs a trading fee of \$100. Within the same block, the JIT LPr adds \$90,000 of liquidity in that small range which now represents 90\% of liquidity in the range. This means that the JIT LP would receive 90\% of the \$100 fee, leaving the passive LP with just 10 dollars compared to his 90 dollars (minus gas fees). We refer to \cite{AdamsUniswap} for a  more detailed introduction to JIT liquidity.

\section{Model}\label{sec:model}

There are two types of liquidity providers: passive LPs and a JIT LP. There are three types of traders: an informed trader, an uninformed trader, and an arbitrageur. The agents play a sequential game of liquidity provision and trading on an AMM. The AMM determines the execution price of a swap between two coins: a risky coin $R$ initially priced at $p>0$ and a stable coin $S$ priced at 1 that acts as a numeraire. 

\subsection{Automated Market Maker}

\paragraph{Pricing Function.} Let $r_0$ and $s_0$ denote the initial reserves of risky and stable coins, respectively. The AMM employs the Uniswap v3 pricing function, defined as 
\begin{align*}
    F(r,s;a,b) = (r+b^{-1/2}(r_0s_0)^{1/2})(s+a^{1/2}(r_0s_0)^{1/2}),
\end{align*}
where parameters $a$ and $b$ set the liquidity provision price range $[a,b]\subseteq[0,\infty]$, ensuring that the price $p$ falls within this range. 
Notably, when $a=0$ and $b=\infty$, this function aligns with the constant product model of Uniswap V2. In contrast, Uniswap v3 introduces flexibility for LPs to choose their active liquidity provision range. For our analysis, however, we assume that LPs are constrained to deposit within a predefined range, $[a,b]$, rather than allowing them to select this range endogenously.

Liquidity providers must deposit coins at the fundamental exchange rate, i.e.,\!
\begin{align*}
    \frac{F_r(r_0,s_0;a,b)}{F_s(r_0,s_0;a,b)} = p,
\end{align*}
where $F_r$ and $F_s$ denote, respectively, the partial derivatives of the function $F$ with respect to $r$ and $s$. Note that  $F_r/F_s$ is the rate at which an infinitesimal amount of risky coins is swapped for stable coins (see also \cite{AMMAdoption21}).

\paragraph{Invariance Relationship.} The AMM mandates that the reserve levels of the liquidity pool satisfy the invariant condition
\begin{align*}
    F(r,s;a,b) = r_0s_0,to deposi
\end{align*}
before and after any trade. In addition, the trader must pay a trading fee proportional to the quantity of risky or stable coins that are to be swapped. We denote the proportional fee rate by $f$ and assume that this fee goes directly to the LPs.\footnote{Many studies in the literature, including Lehar and Parlour (2021) and Hasbrouck, Rivera, and Saleh (2022), make similar assumptions. In some exchanges such as Uniswap, the fee is incorporated into the liquidity pool through the creation of new tokens, as captured for instance in \cite{AMMAdoption21}.} We set the gas fee paid by the agents to prioritize their execution to zero. Under this assumption, JIT liquidity providers always enter the market, which is the case of interest for our study.\footnote{If there is a nonzero priority fee, JIT liquidity providers will deposit only if their earnings from fees and the price impact of trades outweigh the predetermined priority fee.}

\subsection{Sequential Game}

\paragraph{Period 1: Passive LPs Decide Whether to Provide Liquidity.}

In Period 1, $N$ identical passive LPs, indexed by $i\in[N]$, arrive. Each passive LP is endowed with $e_P/N$ risky coins and $pe_P/N$ stable coins, and decides on whether or not to contribute to the liquidity pool. If a passive LP decides to contribute, then they deposit their entire endowment to the pool. Specifically, passive LP $i$ deposits $d^{(i)}_P\in\{0,e_P/N\}$ risky coins and $pd^{(i)}_P$ stable coins. The total amount of passive liquidity provided is $d_P\equiv \sum_{i\in[N]}d^{(i)}_P$ risky coins and $\sum_{i\in[N]}pd^{(i)}_P=pd_P$ stable coins.

\paragraph{Period 2: Arrival of Informed or Uninformed Traders.}
In Period 2, one of two mutually exclusive and collectively exhaustive events occurs:

\begin{itemize}
\item An informed trader, possessing large amounts of both coin types and a perfect predictive signal for the future price $p'$ of the risky asset realized in Period 5, arrives with probability $\alpha\in[0,1]$. If his private signal is that the price $p'$ of the risky coin will depreciate, the trader  exchanges  $q_R\in\R_+$ risky coins for stable coins. Otherwise, he would swap $q_S\in\R_+$ stable coins for risky coins. The swap order $(q_R,q_S)\in\R_+^2\backslash\R_{++}^2$ is then submitted to the blockchain's public mempool.

\item An uninformed trader, also endowed with a large amount of both coin types, arrives in period 2 with probability $(1-\alpha)$. Conditional on arrival, he either likes risky coins or stable coins. He trades for liquidity reasons only with his private valuation $p_U$ of the risky coins given by
\begin{align*}
p_U = \begin{dcases}
\zeta_U^{-1}p & \text{with probability $\psi_U$} \\
\zeta_Up & \text{with probability $(1-\psi_U)$},
\end{dcases}
\end{align*}
where $\zeta_U\in(1+f,\infty)$. The uninformed trader also decides amounts $q_R\in\R_+$ or $q_S\in\R_+$ to swap such that $(q_R,q_S)\in\R_+^2\backslash\R_{++}^2$ and submits the swap order to the public mempool. 

\end{itemize}

\paragraph{Period 3: JIT LP Decides Whether to Provide Liquidity.}

In Period 3, a JIT LP, endowed with large amounts of both coins, arrives with probability $\pi\in[0,1]$. Conditional on arrival, the JIT LP views the swap order in the public mempool and decides on the amount of liquidity to provide. We assume that the JIT LP is sophisticated and can determine whether an informed or uninformed trader submitted the swap order. When facing an informed trader, the JIT LP deduces the informed trader's predictive signal based on the type of coin swapped; when facing an uninformed trader, the JIT LP deduces that there will be no off-chain price shock. Specifically, the JIT LP submits an order to deposit $d_{J}\in\R_+$ risky coins and $pd_{J}$ stable coins right before the swap, and an order to withdraw its share of the pool right after the swap.

\paragraph{Period 4: Settlement of Transactions.}

In Period 4, the block containing the swap and JIT liquidity transactions is validated, leading to the following sequence of events:
\begin{itemize}
    \item \textbf{JIT LP Deposit:} The order from the JIT LP to deposit into the liquidity pool is executed. 
    \item \textbf{Swap Execution and Fee Distribution:} The swap order of the trader is executed. The passive LPs collectively receive a pro-rata share of the transaction fee, where the passive LPs share is
    \begin{align*}
        (1-\mathbbm{1}\{\text{JIT LP arrives}\})+\frac{d_P}{d_P+d_J}\cdot\mathbbm{1}\{\text{JIT LP arrives}\}.
    \end{align*}
    The remaining share of fees is earned by the JIT LP.
    \item \textbf{JIT LP Withdrawal:} Post-swap, the JIT LP withdraws their share of the pool.
\end{itemize}

\paragraph{Period 5: Market Shocks and Responses.}

In Period 5, events unfold in a manner perfectly correlated with those in Period 2, leading to one of two mutually exclusive scenarios:

\begin{itemize}
    \item  If an informed trader arrived in Period 2, the price of the risky coin moves to $p'$ where
    \begin{align*}
        p' = \begin{dcases}
            \zeta^{-1}p & \text{with probability $\psi$} \\
            \zeta p & \text{with probability $(1-\psi)$}
        \end{dcases}
    \end{align*}
    where $\zeta\in(1+f,\infty)$.
    \item If an uninformed trader arrived in Period 2, the risky coin's remains at $p$. This creates a \textit{reverse trade arbitrage opportunity}, as described in \cite{AMMAdoption21}. An arbitrageur arrives and trades to move the AMM's spot rate back to $p$, restoring the liquidity pool's reserves to the original levels. 
\end{itemize}

All agents are risk-neutral and valuate their holdings after period 5 with no discount factor. We assume that exogenous parameters ($\alpha$, $\zeta$, $\zeta_U$, $\psi$, $\psi_U$, $f$ and $\pi$) are known by all strategic agents.

\vspace{0.33cm}

\begin{center}
    
\tikzset{every picture/.style={line width=0.75pt}} 

\begin{tikzpicture}[x=0.75pt,y=0.75pt,yscale=-1,xscale=1]


\draw    (73,178) -- (107.5,133.37) ;
\draw [shift={(109.33,131)}, rotate = 127.71] [fill={rgb, 255:red, 0; green, 0; blue, 0 }  ][line width=0.08]  [draw opacity=0] (8.93,-4.29) -- (0,0) -- (8.93,4.29) -- cycle    ;
\draw    (73,178) -- (107.52,223.61) ;
\draw [shift={(109.33,226)}, rotate = 232.88] [fill={rgb, 255:red, 0; green, 0; blue, 0 }  ][line width=0.08]  [draw opacity=0] (8.93,-4.29) -- (0,0) -- (8.93,4.29) -- cycle    ;
\draw    (207,225) -- (240.33,225) ;
\draw [shift={(243.33,225)}, rotate = 180] [fill={rgb, 255:red, 0; green, 0; blue, 0 }  ][line width=0.08]  [draw opacity=0] (8.93,-4.29) -- (0,0) -- (8.93,4.29) -- cycle    ;
\draw    (207,132) -- (240.33,132) ;
\draw [shift={(243.33,132)}, rotate = 180] [fill={rgb, 255:red, 0; green, 0; blue, 0 }  ][line width=0.08]  [draw opacity=0] (8.93,-4.29) -- (0,0) -- (8.93,4.29) -- cycle    ;
\draw  [fill={rgb, 255:red, 74; green, 144; blue, 226 }  ,fill opacity=0.33 ] (345.33,126) .. controls (345.33,116.06) and (353.39,108) .. (363.33,108) -- (417.33,108) .. controls (427.27,108) and (435.33,116.06) .. (435.33,126) -- (435.33,228) .. controls (435.33,237.94) and (427.27,246) .. (417.33,246) -- (363.33,246) .. controls (353.39,246) and (345.33,237.94) .. (345.33,228) -- cycle ;
\draw    (308,226) -- (341.29,226.26) ;
\draw [shift={(344.29,226.29)}, rotate = 180.45] [fill={rgb, 255:red, 0; green, 0; blue, 0 }  ][line width=0.08]  [draw opacity=0] (8.93,-4.29) -- (0,0) -- (8.93,4.29) -- cycle    ;
\draw    (308,133) -- (341.29,133.26) ;
\draw [shift={(344.29,133.29)}, rotate = 180.45] [fill={rgb, 255:red, 0; green, 0; blue, 0 }  ][line width=0.08]  [draw opacity=0] (8.93,-4.29) -- (0,0) -- (8.93,4.29) -- cycle    ;
\draw    (345.33,160) -- (435.33,160) ;
\draw    (345.33,193) -- (435.33,193) ;
\draw    (437,225) -- (470.33,225) ;
\draw [shift={(473.33,225)}, rotate = 180] [fill={rgb, 255:red, 0; green, 0; blue, 0 }  ][line width=0.08]  [draw opacity=0] (8.93,-4.29) -- (0,0) -- (8.93,4.29) -- cycle    ;
\draw    (437,132) -- (470.33,132) ;
\draw [shift={(473.33,132)}, rotate = 180] [fill={rgb, 255:red, 0; green, 0; blue, 0 }  ][line width=0.08]  [draw opacity=0] (8.93,-4.29) -- (0,0) -- (8.93,4.29) -- cycle    ;
\draw [color={rgb, 255:red, 128; green, 128; blue, 128 }  ,draw opacity=1 ] [dash pattern={on 4.5pt off 4.5pt}]  (308,177) -- (341.29,177.26) ;
\draw [shift={(344.29,177.29)}, rotate = 180.45] [fill={rgb, 255:red, 128; green, 128; blue, 128 }  ,fill opacity=1 ][line width=0.08]  [draw opacity=0] (8.93,-4.29) -- (0,0) -- (8.93,4.29) -- cycle    ;
\draw [color={rgb, 255:red, 128; green, 128; blue, 128 }  ,draw opacity=1 ] [dash pattern={on 4.5pt off 4.5pt}]  (160.29,154.86) -- (185.89,174.19) ;
\draw [shift={(188.29,176)}, rotate = 217.06] [fill={rgb, 255:red, 128; green, 128; blue, 128 }  ,fill opacity=1 ][line width=0.08]  [draw opacity=0] (8.93,-4.29) -- (0,0) -- (8.93,4.29) -- cycle    ;
\draw [color={rgb, 255:red, 128; green, 128; blue, 128 }  ,draw opacity=1 ] [dash pattern={on 4.5pt off 4.5pt}]  (160.29,199.86) -- (186,177.95) ;
\draw [shift={(188.29,176)}, rotate = 139.57] [fill={rgb, 255:red, 128; green, 128; blue, 128 }  ,fill opacity=1 ][line width=0.08]  [draw opacity=0] (8.93,-4.29) -- (0,0) -- (8.93,4.29) -- cycle    ;

\draw (25,68) node [anchor=north west][inner sep=0.75pt]   [align=left] {$\displaystyle t=1$};
\draw    (-15,155) -- (72,155) -- (72,200) -- (-15,200) -- cycle  ;
\draw (-15,159) node [anchor=north west][inner sep=0.75pt]   [align=left] {\begin{minipage}[lt]{62.72pt}\setlength\topsep{0pt}
\begin{center}
passive LPs\\[-0.15\baselineskip]arrive
\end{center}
\end{minipage}};

\draw (80,128) node [anchor=north west][inner sep=0.75pt]   [align=left] {$\alpha$};

\draw (64,218) node [anchor=north west][inner sep=0.75pt]   [align=left] {$1-\alpha$};

\draw (216,118) node [anchor=north west][inner sep=0.75pt]   [align=left] {$\pi$};

\draw (216,210) node [anchor=north west][inner sep=0.75pt]   [align=left] {$\pi$};

\draw    (110,108) -- (206,108) -- (206,153) -- (110,153) -- cycle  ;
\draw (114.5,112) node [anchor=north west][inner sep=0.75pt]   [align=left] {\begin{minipage}[lt]{62.81pt}\setlength\topsep{0pt}
\begin{center}
informed\\[-0.15\baselineskip]trader arrives
\end{center}

\end{minipage}};
\draw    (110,201) -- (206,201) -- (206,246) -- (110,246) -- cycle  ;
\draw (114.5,205) node [anchor=north west][inner sep=0.75pt]   [align=left] {\begin{minipage}[lt]{62.81pt}\setlength\topsep{0pt}
\begin{center}
uninformed\\[-0.15\baselineskip]trader arrives
\end{center}

\end{minipage}};
\draw    (244,201) -- (307,201) -- (307,246) -- (244,246) -- cycle  ;
\draw (247,205) node [anchor=north west][inner sep=0.75pt]   [align=left] {\begin{minipage}[lt]{40.72pt}\setlength\topsep{0pt}
\begin{center}
JIT LP\\[-0.15\baselineskip]arrives
\end{center}

\end{minipage}};
\draw (365.33,113) node [anchor=north west][inner sep=0.75pt]   [align=left] {\begin{minipage}[lt]{36.91pt}\setlength\topsep{0pt}
\begin{center}
{\small JIT LP}\\[-0.5\baselineskip]{\small deposit}\\[-0.5\baselineskip]{\small (if any)}\\
\end{center}

\end{minipage}};
\draw (360,198) node [anchor=north west][inner sep=0.75pt]   [align=left] {\begin{minipage}[lt]{43.53pt}\setlength\topsep{0pt}
\begin{center}
{\small JIT LP}\\[-0.5\baselineskip]{\small withdrawal}\\[-0.5\baselineskip]{\small (if any)}
\end{center}

\end{minipage}};
\draw (371,172) node [anchor=north west][inner sep=0.75pt]   [align=left] {\begin{minipage}[lt]{26.53pt}\setlength\topsep{0pt}
\begin{center}
swap
\end{center}

\end{minipage}};
\draw    (244,108) -- (307,108) -- (307,153) -- (244,153) -- cycle  ;
\draw (247,112) node [anchor=north west][inner sep=0.75pt]   [align=left] {\begin{minipage}[lt]{40.72pt}\setlength\topsep{0pt}
\begin{center}
JIT LP\\[-0.15\baselineskip]arrives
\end{center}

\end{minipage}};
\draw    (474,108) -- (529,108) -- (529,153) -- (474,153) -- cycle  ;
\draw (477,112) node [anchor=north west][inner sep=0.75pt]   [align=left] {\begin{minipage}[lt]{35.04pt}\setlength\topsep{0pt}
\begin{center}
price\\[-0.15\baselineskip]shock 
\end{center}

\end{minipage}};
\draw    (474,201) -- (531,201) -- (531,246) -- (474,246) -- cycle  ;
\draw (476.5,208) node [anchor=north west][inner sep=0.75pt]   [align=left] {\begin{minipage}[lt]{36.73pt}\setlength\topsep{0pt}
\begin{center}
reverse\\[-0.15\baselineskip]trade
\end{center}

\end{minipage}};
\draw (139,68) node [anchor=north west][inner sep=0.75pt]   [align=left] {$\displaystyle t=2$};
\draw (257,68) node [anchor=north west][inner sep=0.75pt]   [align=left] {$\displaystyle t=3$};
\draw (372,68) node [anchor=north west][inner sep=0.75pt]   [align=left] {$\displaystyle t=4$};
\draw (483,68) node [anchor=north west][inner sep=0.75pt]   [align=left] {$\displaystyle t=5$};
\draw  [color={rgb, 255:red, 128; green, 128; blue, 128 }  ,draw opacity=1 ][dash pattern={on 4.5pt off 4.5pt}]  (190,165) -- (307,165) -- (307,189) -- (190,189) -- cycle  ;
\draw (195,170) node [anchor=north west][inner sep=0.75pt]   [align=left] {\begin{minipage}[lt]{77.2pt}\setlength\topsep{0pt}
\begin{center}
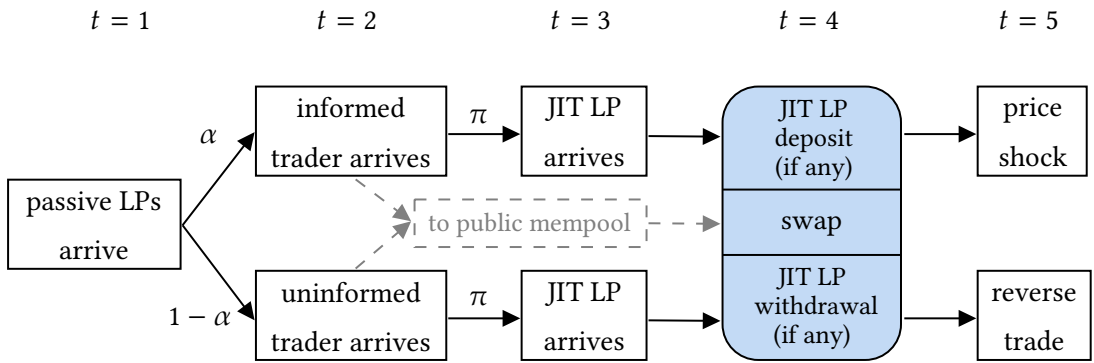

{\small \textcolor[rgb]{0.5,0.5,0.5}{to public mempool}}
\end{center}

\end{minipage}};

\end{tikzpicture} \vspace{0.5cm}
\captionof{figure}{\small Timeline of the sequential game.}
\label{fig:timeline}

\end{center}

\subsection{Strategies and Utilities}

\paragraph{Strategy Space.} Depending on the realization of trader arrivals and price / private value shocks, one of the following four scenarios occurs:
\begin{itemize}
    \item \textit{Informed Sell:} if a price shock occurs with $p'=\zeta^{-1} p$, then the informed trader chooses to sell risky coins for stable coins, due to depreciation in the risky coin's price.
    \item \textit{Informed Buy:} if a price shock occurs with $p'=\zeta p$, then the informed trader chooses to sell stable coins for risky coins, due to appreciation in the risky coin's price.
    \item \textit{Uninformed Sell:} if there is no price shock and the private value is $p_U=\zeta_U^{-1} p$, then the uninformed trader sells risky coins for stable coins, as his private valuation for risky coins is low.
    \item \textit{Uninformed Buy:} if there is no price shock and the private value is $p_U=\zeta_Up$, then the uninformed trader sells stable coins for risky coins, as his private valuation for risky coins is high.
\end{itemize}
The traders and the JIT LP can distinguish between these scenarios, but the passive LPs cannot. Denote the set of scenarios as $\Omega=\{\mathsf{IS}, \mathsf{IB}, \mathsf{US}, \mathsf{UB}\}$ for informed sell, informed buy, uninformed sell, and uninformed buy, respectively. The strategy spaces for the passive LPs, JIT LPs, and traders are, respectively,
\begin{gather*}
    d_P^{(i)}\in[0,{e}_{P}/N] \ \forall \ i\in[N], \\
       d_{J}:[0,{e}_{P}]\times([0,{e}_{P}]\times\Omega\to\R_+^2\backslash\R_{++}^2)\times\Omega\to\R_+ \\
       (q_R,q_S):[0,{e}_{P}]\times\Omega\to\R_+^2\backslash\R_{++}^2. 
\end{gather*}
where $(q_R,q_S)( d_P,\omega)$ for $\omega\in\{\mathsf{IS},\mathsf{IB}\}$ refers to the informed trader's strategy, and for $\omega\in\{\mathsf{US},\mathsf{UB}\}$, it refers to the uninformed trader's strategy. A strategy profile for this game is given by the tuple:
\begin{gather*}
    \sigma\equiv(\sigma_{P},\sigma_T,\sigma_{J})\equiv(\{  d_P^{(i)}\}_{i\in[N]},\{(q_R,q_S)(   d_P,\omega)\}_{\omega\in\Omega},\{d_{J}( d_P,(q_R,q_S)( d_P,\omega);\omega)\}_{\omega\in\Omega}).
\end{gather*}

\paragraph{Utilities.} Explicit forms of these utilities are provided in the appendices. The JIT LP's utility, conditional on arrival, is given by
\begin{align*}
    u_{J}(d_{J};\sigma_{-J},\omega) = \text{JIT LP's Share}\times(\text{Pool Value After Swap}+\text{Fees})-\text{Deposit Value}.
\end{align*}
Both traders' expected utility when choosing a swap order $(q_R,q_S)$ is given by
\begin{align*}
    u_T((q_R,q_S);\sigma_{-T},\omega) &= \pi \times\text{Value Received with JIT Liquidity} \\
    &\hspace{1cm} +(1-\pi )\times\text{Value Received w/o JIT Liquidity}-\text{Value Sent}-\text{Fees},
\end{align*}
where the value sent and received depends on the trader's valuation of the risky coin's price. The passive LPs' total conditionally expected utility given scenario $\omega$ is  given by
\begin{align*}
    u_{P}(d_P;\sigma_{-P},\omega) &= \pi \times \text{Passive LPs' Share}\times(\text{Pool Value After Swap with JIT Liquidity}+\text{Fees}) \\
    &\hspace{1cm}+(1-\pi )\times(\text{Pool Value After Swap w/o JIT Liquidity}+\text{Fees})
    -\text{Deposit Value}.
\end{align*}
Since the passive LPs cannot distinguish between scenarios when deciding the amount of the liquidity to provide, the passive LPs' total expected utility is given by
\begin{align*}
    u_P(d_P;\sigma_{-P})&=\alpha\psi\cdot u_P(d_P;\sigma_{-P},\mathsf{IS})+\alpha(1-\psi)\cdot u_P(d_P;\sigma_{-P},\mathsf{IB}) \\
    &\hspace{1cm}+(1-\alpha)\psi_U\cdot u_P(d_P;\sigma_{-P},\mathsf{US})+(1-\alpha)\psi_U\cdot u_P(d_P;\sigma_{-P},\mathsf{UB})
\end{align*}
Our solution concept for this game is a pure-strategy subgame-perfect Nash equilibrium (SPNE). To facilitate the analysis, we make the following assumption. \\[-0.75\baselineskip]

\noindent \textbf{\textsc{Assumption 1.}} \textit{We assume the following tie-breaking criteria:}
\begin{itemize}
    \item \textit{If two deposit quantities yield the same utility for a liquidity provider, then an passive LP prefers the larger deposit quantities and the JIT LP prefers the smaller deposit quantities.}
    \item \textit{If two swap quantities yield the same utility, then a trader prefers the smaller swap quantity.}
\end{itemize}

\section{The Effects of Just-in-Time Liquidity on Liquidity Pool Depth}\label{sec:benefit}

The search for SPNE is conducted through backward induction. We first fix the deposit amount of passive LPs and then determine the Nash equilibrium in the subgame involving either the informed or uninformed trader and the JIT LP. We address each scenario separately. Following this, we then determine the optimal response of passive LPs given the equilibrium strategies of traders and the JIT LP in their subgame.

\begin{proposition}\label{thm:subgame-eq}
A non-trivial Nash equilibrium in the subgame between the traders and JIT LP exists if and only if $\zeta_U>\underline{\zeta}(f,\pi)$ where
\begin{align*}
    \underline{\zeta}(f,\pi)=\frac{2(1+f)^3}{2+\pi f(3+f)}.
\end{align*}
Moreover, the equilibrium is unique. In the equilibrium outcome:
\begin{itemize}
    \item The value (at Period 2) of the informed trader's buy and sell swap orders is a fixed multiple $\mu_I$ of the amount of passive liquidity provided and is constant in $\pi$.
    \item The value (at Period 2) of the uninformed trader's buy and sell swap orders is a fixed multiple $\mu(\pi)$ of the amount of passive liquidity provided where $\mu(\pi)$ is increasing in $\pi$.
    \item Conditional on arrival, the JIT LP provides liquidity only when facing an uninformed trader. The amount of liquidity provided is a fixed multiple $\nu(\pi)$ of the amount of passive liquidity provided.
\end{itemize}
\end{proposition}

The JIT LP's decision to deposit or not is contingent on the type of trader they face. Facing informed traders, who trade optimally based on their knowledge of future prices, the JIT LP refrains from depositing, avoiding adverse selection. Conversely, when an uninformed trader arrives, the JIT LP is incentivized to deposit to earn fees and capitalize on the price impact of the swap. Consequently, the informed trader, who is aware of the JIT LP's non-deposit in their presence, bases their swap quantity on the existing liquidity provided by passive LPs. Uninformed traders, expecting a deeper liquidity pool on average due to the JIT LP's potential arrival, trade larger amounts. 

In equilibrium, the amount of liquidity provided by the JIT LP is not arbitrarily large, but is rather a fixed multiple of the existing passive liquidity. This is due to a critical balance the JIT LP must strike between two factors: the marginal benefit of acquiring a larger share of the pool's fees by increasing their deposit, and the marginal loss from a lower price impact caused by a deeper liquidity pool. If the JIT LP were to deposit an excessively large amount, the resulting depth of the liquidity pool would significantly diminish the price impact of traders' swaps, leading to a decrease in gains the JIT LP could expect from price movements generated by these trades. When $\zeta_U>\underline{\zeta}(\pi)$ does not hold, then $\zeta_U$ is small and generates a small uninformed trading volume. In this case, the marginal benefit of acquiring a larger share is always greater than the marginal loss from a lower price impact. The JIT LP's optimal deposit amount is infinite given a positive trading volume, so no non-trivial Nash equilibrium exists.

We now turn our attention to characterizing the strategy of passive LPs. The expected utility for a passive LP $i$, denoted as $u_{P}(d^{(i)}_P)$, can be expressed as follows:
\begin{align*}
    u_{P}(d^{(i)}_P) = p(\alpha \mathcal{C} + (1-\alpha)\cdot\mathcal{R}(\pi))d^{(i)}_P,
\end{align*}
where $\mathcal{C}<0$ is the cost of adverse selection per unit of liquidity deposited due to informed trading and $\mathcal{R}(\pi)$ is the expected fee revenue per unit of liquidity deposited earned from uninformed trading. Thus $\alpha \mathcal{C} + (1-\alpha)\cdot\mathcal{R}(\pi)$ is a passive LP's utility per unit of liquidity deposited. Unlike the JIT LP, passive LPs do not benefit from the price impact generated by uninformed traders due to the arbitrageur's reverse trade, cutting off this potential source of revenue for passive LPs.

\begin{proposition}\label{thm:passive-br}
A non-trivial subgame-perfect Nash equilibrium exists if and only if $\zeta_U>\underline{\zeta}(f,\pi)$ where
\begin{align*}
    \underline{\zeta}(f,\pi)=\frac{2(1+f)^3}{2+\pi f(3+f)}.
\end{align*}
Moreover, the equilibrium is unique. Define $\mathcal{U}(\pi)=\alpha \mathcal{C} + (1-\alpha)\cdot\mathcal{R}(\pi)$. In the equilibrium outcome, the amount of passive liquidity provided is $d_P^\star=e_P\cdot\mathbbm{1}\{\mathcal{U}(\pi)\geq0\}$.
\end{proposition}

If the cost of adverse selection is less than the revenue from the transaction fees of uninformed trades, then all passive LPs choose to provide liquidity. Otherwise, if adverse selection losses exceed expected fee revenues, then passive LPs will abstain from depositing. In this situation, zero passive liquidity implies that traders and the JIT LP do not participate as well. We thus have a market breakdown or \textit{liquidity freeze}, where liquidity vanishes and trading ceases. One may wonder why traders do not swap in the absence of passive market liquidity, or why they do not place market orders in anticipation of JIT LPs filling them. The concern here is that JIT liquidity might be minimal, leading to large price impacts for the traders as JIT LPs, who move after the traders, can limit their liquidity supply to earn a large price impact in such cases. This underscores the critical roles of passive LPs as both a safety net and an alternative for traders, protecting them from predatory practices by JIT LPs analogous to the way HFT firms prey on large orders (see \cite{PREY}).

For $\pi>0$, if $\mathcal{R}(0)>\mathcal{R}(\pi)$, then the per-unit utility of passive LPs is lower in the presence of a JIT LP that arrives with probability $\pi$ compared to the absence of a JIT LP. A reduction in utility significantly affects the equilibrium outcome if $\mathcal{U}(0)\geq0$ but $\mathcal{U}(\pi)<0$. In such scenarios, the possibility of a JIT LP arriving induces a liquidity freeze that would not occur in the complete absence of the JIT LP ($\pi=0$). When there is no adverse selection (i.e.\!\, $\alpha = 0$), the passive LPs' total utility is always positive. Conversely, when all order flows are informed (i.e.\!\, $\alpha = 1$), the utility is always negative. Thus, there always exists some interval $[\underline{\alpha}, \overline{\alpha}]\subseteq[0,1]$ where a JIT LP-induced liquidity freeze occurs when $\alpha\in[\underline{\alpha}, \overline{\alpha}]$. We now formalize this notion of utility reduction and establish the conditions under which it occurs.

\begin{definition*}
Let $\pi\in(0,1]$. Suppose that unique equilibria exist when the JIT LP's arrival probability is zero and $\pi$, holding the other exogenous parameters fixed. We say that the JIT LP
\begin{itemize}
    \item \ul{complements} the passive LPs at arrival probability $\pi$ if $\mathcal{R}(0) \leq \mathcal{R}(\pi)$;
    \item \ul{crowds out} the passive LPs at arrival probability $\pi$ if $\mathcal{R}(0) > \mathcal{R}(\pi)$.
\end{itemize}
\end{definition*}

\begin{theorem}\label{thm:threshold}
Let $\pi\in(0,1]$. Then one of the following is true:
\begin{itemize}
    \item The JIT LP complements the passive LPs at arrival probability $\pi$ for all $\zeta_U>\underline{\zeta}(f,\pi)$.
    \item There exists $\zeta^\star(f,\pi)>\underline{\zeta}(f,\pi)$ such that the JIT LP complements the passive LPs at arrival probability $\pi$ if and only if $\zeta_U\geq\zeta^\star(f,\pi)$.
\end{itemize}
Moreover, when $\pi=1$, the $\zeta_U$ threshold is given explicitly by $\zeta^\star(f,1)=(\sqrt{f}+\sqrt{1+f})^2$, assuming that $\underline{\zeta}(f,\pi)<(\sqrt{f}+\sqrt{1+f})^2$.
\end{theorem}

\begin{figure}[t!]
\centering
  \includegraphics[scale=0.67]{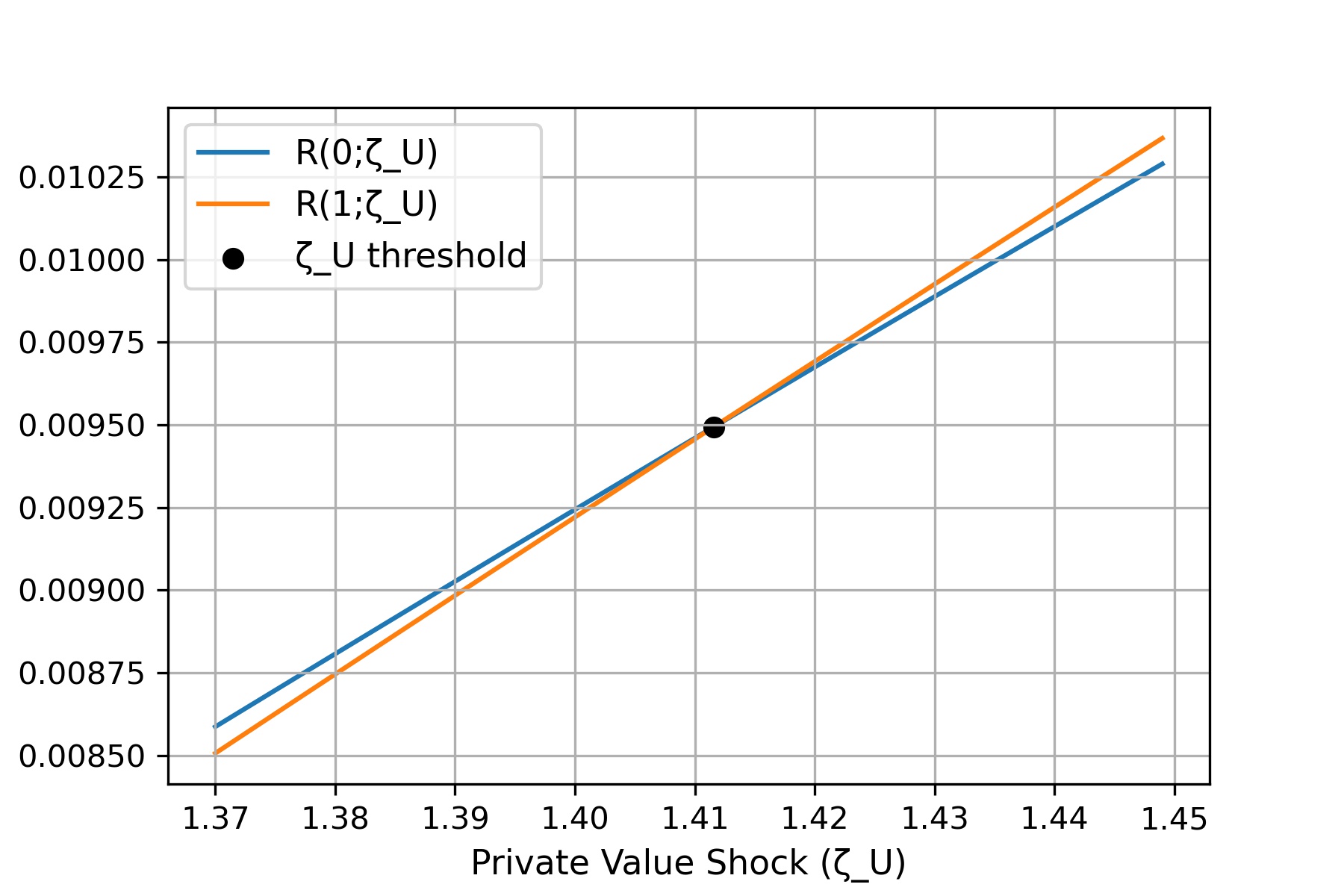} \\[0.25\baselineskip]
    \captionof{figure}{\centering \small $\zeta^\star(1)=(\sqrt{f}+\sqrt{1+f})^2$ is a threshold for crowding out and complementing when the JIT arrives with probability one. Here $f=0.03$ (chosen to illustrate the threshold the best), yielding $\zeta^\star(1)\approx1.4116$.}
    \label{fig:threshold}
\end{figure}

To understand the economic forces underlying the result of Theorem \ref{thm:threshold}, note that we can decompose the fee revenue $\mathcal{R}(\pi)$ as follows:
\begin{gather*}
    \mathcal{R}(\pi) \propto \E\!\left[\frac{\text{Passive Liquidity}}{\text{Passive Liquidity}+\text{JIT Liquidity}}\times(\text{Uninformed Swap Volume}+\text{Reverse Swap Volume})\right] 
\end{gather*}
The fees earned by the passive LPs is linked to their share of the liquidity pool and the total trading volume attracted by the AMM in the event of an uninformed arrival. In the complete absence of a JIT LP, passive LPs have a 100\% share of the pool. When the probability of a JIT LP arriving moves to $\pi\in(0,1]$, there is a chance that the the passive LPs' share of the pool is reduced by the JIT liquidity. However, Proposition \ref{thm:passive-br} states that the uninformed trader's swap volume increases in $\pi$ (since he must submit the order without knowledge of whether the JIT LP's arrival), so the total trading volume attracted by the AMM is increasing in $\pi$. If the increase in expected trading volume is proportionally less than the expected reduction in pool share, then passive LPs' fee revenue decreases with the potential arrival of a JIT LP. Conversely, if the increase in expected trading volume is proportionally greater than the expected reduction in pool share, passive LPs' fee revenue will increase when the JIT LP arrival probability moves from zero to $\pi$.

When uninformed trader's demand for coins is highly responsive to the expected depth of the liquidity pool, the JIT LP and passive LPs become strategic complements. In this scenario, the JIT LP, while capturing a significant portion of transaction fees, indirectly benefits passive LPs by inducing a higher trading volume that compensates for the share of fees it takes away from the passive LPs. Their participation thus reduces transaction costs and attracts more orders, supporting the involvement of passive LPs in the AMM and leading to increased aggregate market liquidity.

Conversely, when the uninformed trader's demand for coins is relatively unresponsive to changes in expected pool depth, the JIT LP and passive LPs are strategic substitutes. Here, the JIT LP accrues a portion of the transaction fees but does not induce a proportional increase in trading volume from the uninformed trader. The extra trading volume fails to compensate for the reduced pool share of passive LPs, lowering passive LPs' equilibrium payoffs and possibly leading to a liquidity freeze.

Theorem \ref{thm:threshold} asserts that the responsiveness of the uninformed trader's demand is intrinsically related to the private value shock size. For any value of $\pi\in(0,1]$, we can find a threshold $\zeta^\star(f,\pi)$ that delineates a crowding out regime for $\zeta_U<\zeta^\star(f,\pi)$ (small $\zeta_U$) and a complementing regime for $\zeta_U>\zeta^\star(f,\pi)$ (large $\zeta_U$). If $\zeta^\star(f,\pi)\leq\underline{\zeta}(f,\pi)$, then the first outcome (always complementing on $\zeta_U\in(\underline{\zeta}(f,\pi),\infty)$) occurs. Otherwise, if $\zeta^\star(f,\pi)>\underline{\zeta}(f,\pi)$, then we are in the second outcome, where the threshold is the range of $\zeta_U$ values that yield equilibria. When $\pi=1$, this threshold has an explicit form. Figure \ref{fig:threshold} illustrates the threshold phenomena and confirms the explicit form for $\pi=1$ as well.

\section{Fee Reallocation for Passive and Just-in-Time Liquidity Providers}\label{sec:ttfs}

In this extension, we propose a two-tiered fee structure to address the challenges posed by JIT liquidity crowding out passive LPs. This fee structure aims to retain passive LPs while ensuring that JIT LPs still supply adequate liquidity, improving aggregate liquidity and gains from trade.

Under this fee structure, JIT LPs are allocated a portion of the transaction fees based on a \textit{transfer rate} parameter $\lambda \in [0,1]$. Specifically, JIT LPs retain a $\lambda$ fraction of their pro-rata share of the fees, with the remaining $(1-\lambda)$ fraction being transferred to passive LPs. The value of $\lambda$ thus influences the fee distribution: lower $\lambda$ values correspond to a greater share of fees being transferred to passive LPs. Note that when $\lambda = 1$, the fee structure corresponds to the baseline model, where JIT LPs retain their entire pro-rata share of fees. The resulting \textit{effective share} of the liquidity pool for passive LPs and the JIT LP is given by
\begin{gather*}
    \text{Passive LPs Effective Share} = \frac{\text{Passive Liquidity}+(1-\lambda)\times\text{JIT Liquidity}}{\text{Passive Liquidity}+\text{JIT Liquidity}}, \\[0.25\baselineskip]
    \text{JIT LP Effective Share} = \frac{\lambda \times\text{JIT Liquidity}}{\text{Passive Liquidity}+\text{JIT Liquidity}}.
\end{gather*}
The effective share is factored into the utilities of the passive and JIT LPs, which are explicitly given in the appendices. Since traders still have to pay the fee in full, their utilities do not change relative to the baseline model, although their equilibrium best-response to the JIT LP might differ. We solve for equilibria under the two-tiered fee structure, which can be characterized by the following theorem.

\begin{proposition}\label{thm:twotier-eq}
Under the two-tiered fee structure with transfer rate $\lambda\in[0,1]$, a non-trivial subgame-perfect Nash equilibrium exists if and only if $\zeta_U > \underline{\zeta}(f,\lambda,\pi)$, where
\begin{align*}
    \underline{\zeta}(f,\lambda,\pi) = \frac{2(1+f)^3}{2+\pi[(2+f)\sqrt{(1+f)(1+\lambda f)}-2]}.
\end{align*}
Moreover, the equilibrium is unique. Define $\mathcal{U}(\pi,\lambda)=\alpha\mathcal{C}+(1-\alpha)\cdot\mathcal{R}(\pi,\lambda)$. In the equilibrium outcome:
\begin{itemize}
    \item The amount of passive liquidity provided
    is $d_P^\star=e_P\cdot\mathbbm{1}\{\mathcal{U}(\lambda,\pi)\geq0\}$.
    \item The value of the informed trader's buy and sell swap orders is $\mu_Id_P^\star$ where $\mu_I$ is the same constant as in Proposition \ref{thm:passive-br}.
    \item The value of the uninformed trader's buy and sell swap orders is $\mu(\lambda,\pi)$ for some function $\mu(\lambda,\pi)$ that is increasing in $\pi$. 
    \item Conditional on the arrival of one JIT LP, the JIT LP provides liquidity only when facing an uninformed trader. The amount of JIT liquidity provided is $d_J^\star=\nu(\lambda,\pi)\cdot d_P^\star$.
\end{itemize}
\end{proposition}

Proposition \ref{thm:twotier-eq} shows that important characteristics of the baseline model's equilibrium are preserved in the two-tiered fee system: passive LPs decide whether or not to deposit based on the per-unit utility, and equilibrium swap and JIT liquidity deposit sizes are linear in the amount of passive liquidity provided. The comparative statics of the uninformed trader's swap value with respect to $\pi$ also carry over. The main difference is that important quantities, such as the lower bound on $\zeta_U$ for the existence of equilibria and the passive LPs' per-unit fee revenue, now depend on $\lambda$ as well as $\pi$. The following proposition characterizes some of the comparative statics with respect to $\lambda$.

\begin{proposition}\label{thm:dampen}
Let $\pi\in(0,1]$. Suppose that a unique non-trivial subgame-perfect Nash equilibrium for all $\lambda \in [0,1]$, holding the other exogenous parameters fixed. When $\lambda$ decreases, i.e.\!\, a larger proportion of fees is transferred from the JIT LP to the passive LPs:
\begin{itemize}
\setlength\itemsep{0em}
    \item The ratio of the JIT LP's deposit size to the amount of passive liquidity provided decreases in equilibrium.
    \item The ratio of the uninformed trader's swap value to the amount of passive liquidity provided decreases.
\end{itemize}
\end{proposition}

For a fixed uninformed swap size, as $\lambda$ is reduced, the JIT LP has a diminished best-response deposit size due to having a smaller effective share of the liquidity pool. Allowing the uninformed trader to revise his order knowing that JIT liquidity will decrease means that they swap less aggressively due to the increased price impact resulting form a shallower pool. However, in equilibrium, the JIT LP also needs to consider the uninformed trader's reduced demand, which might increase the best-response deposit size since it is not necessarily monotonic in the swap size. Thus, at first glance, it is uncertain if the JIT LP's equilibrium deposit size increases or decreases in $\lambda$, but Proposition \ref{thm:dampen} clarifies this by showing that these economic forces act alongside each other in a manner that dampens both the JIT LP's deposit size and uninformed trader's swap size in equilibrium.



\paragraph{Preventing Liquidity Freezes} We consider the effect of the fee transfer on the passive LPs' utility. On a surface level, it would seem that transferring more fees to the passive LPs increases their utility, but the transfer's dampening effect suggests that there are contrasting economic forces at play. We find that the uninformed trader's base level of demand (with no transfers) needs to be sufficiently high in order for a higher transfer rate to boost the passive LPs' utility for the entire range of $\lambda\in[0,1]$.

\begin{theorem}\label{thm:preventfreeze}
Let $\pi\in[0,1]$. Suppose that a unique non-trivial subgame-perfect Nash equilibrium exists for all $\lambda\in$ $[0,1]$, holding the other exogenous parameters fixed. If $f\in[0,\pi)$, then there exists $\hat{\zeta}(f,\pi)$, increasing in $f$, such that if $\zeta_U\geq\hat{\zeta}(f,\pi)$, then the passive LPs' per-unit utility decreases in $\lambda$ for all $\lambda\in[0,1]$.
\end{theorem}

Theorem \ref{thm:preventfreeze} provides us with the existence of a threshold $\hat{\zeta}(f,\pi)$ that serves as a \textit{sufficient} condition under which the passive LPs' per-unit utility decreases in $\lambda$ for entire range of $\lambda\in[0,1]$. If $\zeta_U<\hat{\zeta}(f,\pi)$ or $f\geq\pi$, then the passive LPs' per-unit utility is not guaranteed to be decreasing in $\lambda$ for all $\lambda\in[0,1]$. If $f\in[0,\pi)$ and $\zeta_U\geq\hat{\zeta}(f,\pi)$, then transferring a higher proportion of fees from the JIT LP to the passive LPs always increases the passive LPs' per-unit utility. 

This condition can be used to potentially avert a JIT LP-induced liquidity freeze. As previously mentioned, there exists an interval $[\underline{\alpha},\overline{\alpha}]\subset[0,1]$ on which such liquidity freezes occur. If the condition holds, then there exists a subinterval $[\underline{\beta},\overline{\beta}]\subseteq[\underline{\alpha},\overline{\alpha}]$ where lowering $\lambda$ enough pushes the passive LPs' per-unit utility above zero. Formally, for all $\alpha\in[\underline{\beta},\overline{\beta}]$, there exists $\lambda^\star(\alpha)$ such that 
\begin{align*}
    \mathcal{U}(\lambda,\pi) \geq 0 \ \forall \ \lambda\leq\lambda^\star(\alpha), \\
    \mathcal{U}(\lambda,\pi) < 0 \ \forall \ \lambda>\lambda^\star(\alpha).
\end{align*}
In the range $\alpha\in[\underline{\beta},\overline{\beta}]$, there would be a liquidity freeze without the two-tiered fee structure, but choosing any $\lambda\leq\lambda^\star(\alpha)$ prevents the liquidity freeze. Moreover, the choice of $\lambda$ that is optimal for the passive LPs' per-unit utility is $\lambda=0$, i.e.\! the passive LPs take all of the JIT LP's fee revenue, leaving the price impact as the only source of revenue for the JIT LP.

\paragraph{Maximizing Welfare}

We define welfare as the collective utility of all agents: the passive LPs, JIT LPs, traders, and arbitrageurs. While it may appear that the redistribution of coins within this closed system of agents leads to zero welfare, this is not the case. The uninformed trader's private valuation of the risky coin generates positive welfare via gains from trade, as the uninformed trader derives value from buying one coin and selling the other regardless of the actual price dynamics. However, conditional on the arrival of an informed trader, welfare becomes zero as pool value is simply transferred from the passive LPs to the informed trader. 

To maximize welfare in the market while mitigating the risk of market breakdown, it is important to select an appropriate level of $\lambda$. Multiple incentives need to be balanced: the two-tiered fee structure must ensure that passive LPs are incentivized to participate in the liquidity pool, while also ensuring that JIT LPs are not excessively burdened by the fee transfer and can encourage a sufficiently large swap volume from traders. The welfare-optimal choice of $\lambda$ is characterized by the following results.

\begin{theorem}\label{thm:maxwelfare}
Let $\pi\in[0,1]$. Suppose that a unique non-trivial subgame-perfect Nash equilibrium exists for all $\lambda \in$ $ [0,1]$, holding the other exogenous parameters fixed. 
\begin{itemize}
\setlength\itemsep{0em}
    \item If $\mathcal{U}(\lambda,\pi) < 0$ for all $\lambda \in [0,1]$, then welfare is zero for all $\lambda\in[0,1]$, reflecting a consistent absence of passive LP participation and a consequent market breakdown.
    \item Otherwise, if there exists $\lambda\in[0,1]$ such that $\mathcal{U}(\pi,\lambda)\geq0$, then welfare is maximized at the largest value of $\lambda$ at which passive LPs participate in the market, i.e.\!\, $\lambda^\star=\max\{\lambda\in[0,1]:\mathcal{U}(\pi,\lambda)\geq0\}$.
\end{itemize}
\end{theorem}

\vspace{0.5em}

\begin{figure}[h]
\centering
  \includegraphics[scale=0.67]{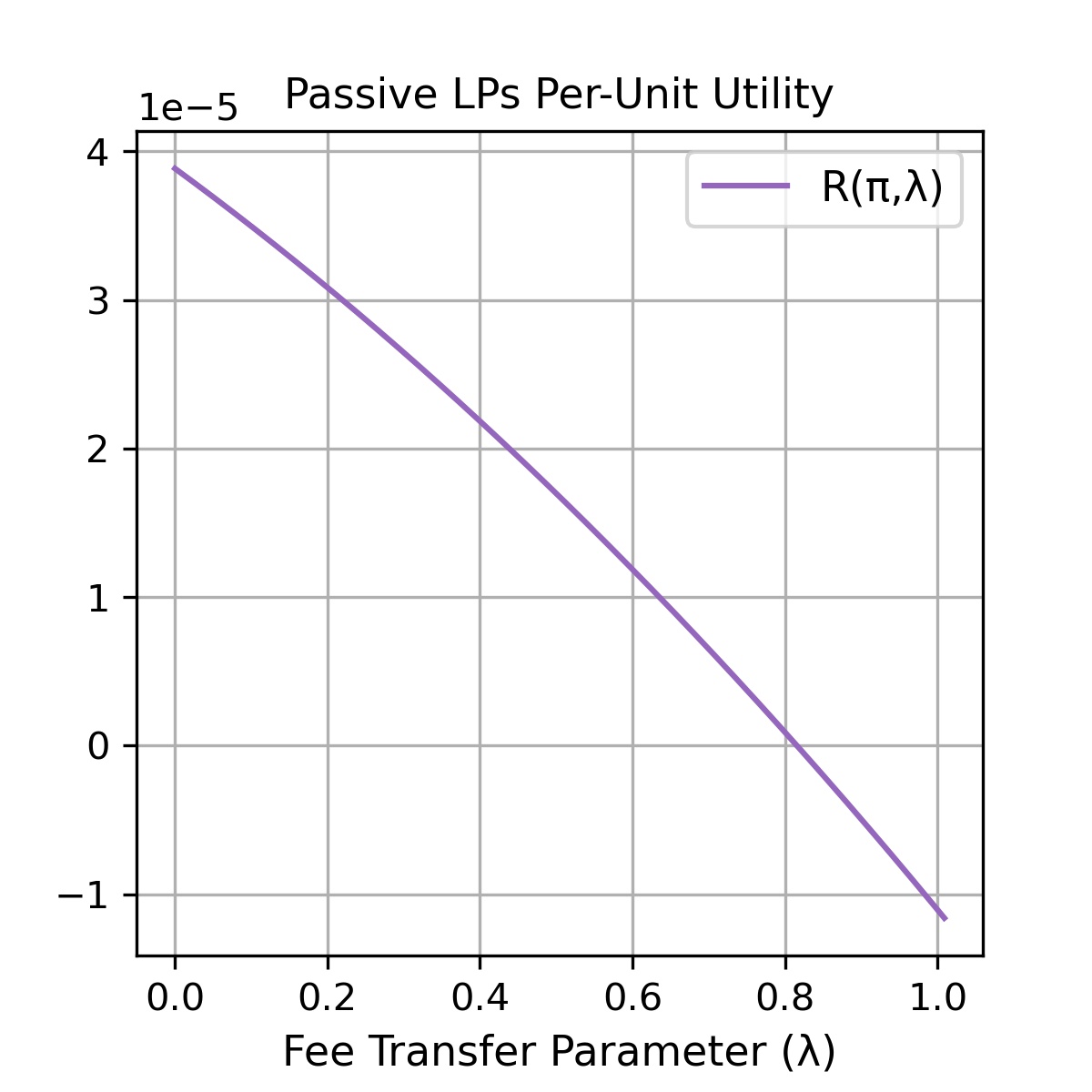} \includegraphics[scale=0.67]{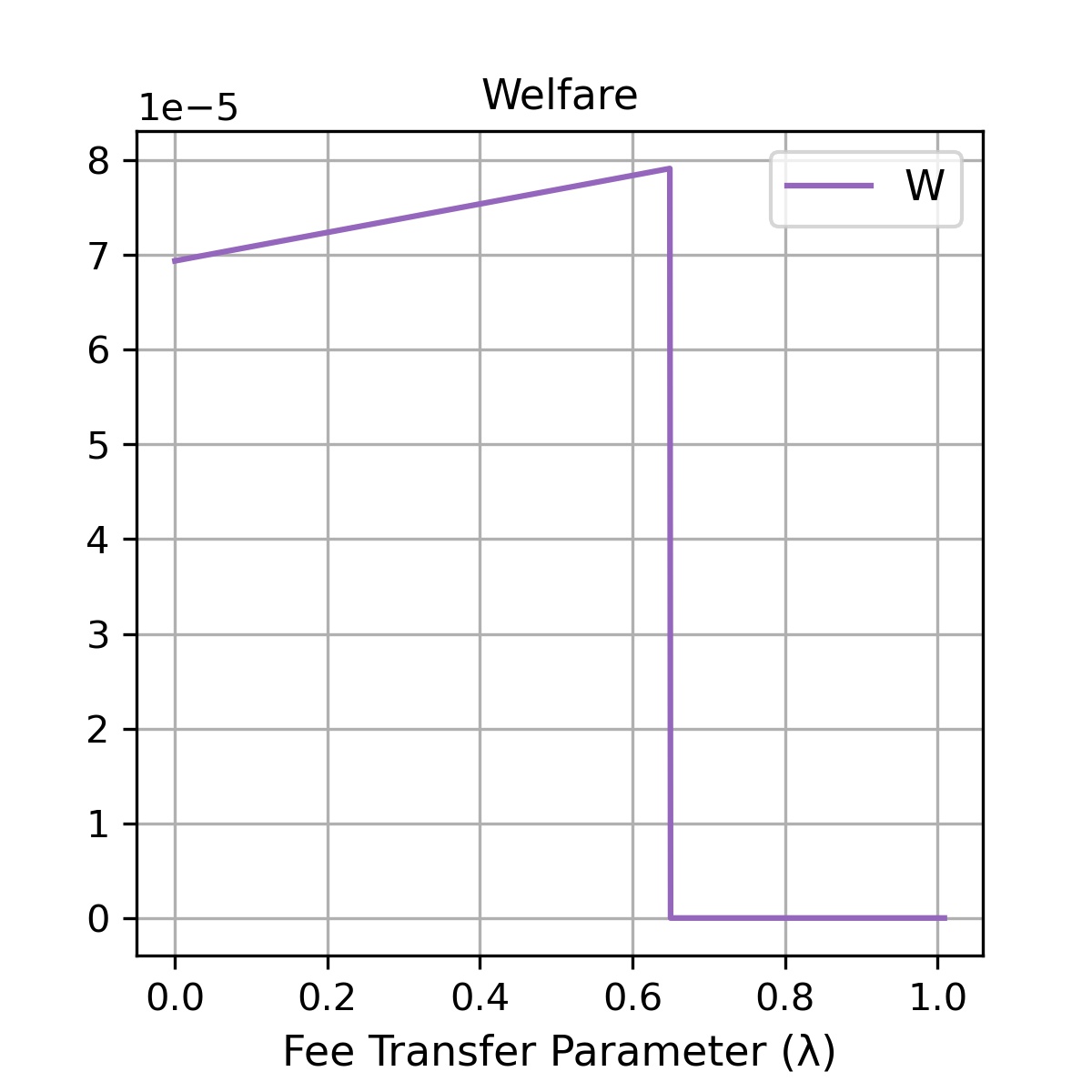} \\[0.25\baselineskip]
    \captionof{figure}{\centering \small Passive LPs' per-unit utility and welfare against $\lambda$. We fixed $\pi=1$, $f=0.003$, $\zeta=1.05$, $\psi=\zeta/(\zeta+1)$, $\zeta_U=1.02$, $\psi_U=\zeta_U/(\zeta_U+1)$, and $\alpha=0.1$.}
    \label{fig:marketdesign}
\end{figure}

The welfare-optimal point is characterized by the highest transfer rate that incentivizes passive LPs to remain in the liquidity pool, yet does not dissuade JIT LPs from rationing their liquidity their too much. If we are in the second case of Theorem \ref{thm:maxwelfare}, then Assumption 1 stipulates that the passive LPs' per-unit utility must be zero when in equilibrium at the welfare-optimal choice of $\lambda$. Figure 4 illustrates this point. As implied by Theorem \ref{thm:dampen}, the choice of $\lambda$ that is optimal for the passive LPs' utility is $\lambda=0$. The parameter values used to generate the figure land us in the second case, where $\lambda^\star\approx0.815$ is the highest fee transfer rate before a JIT LP-induced liquidity freeze occurs. At $\lambda^\star$, the passive LPs' per-unit utility is zero, while the aggregate welfare is maximized.

\paragraph{Practical Implementation}

We discuss the practical implementation of such a two-tiered fee structure in DEXs like Uniswap, with a focus on Uniswap v3's approach using non-fungible tokens (NFTs) to represent individual liquidity positions for LPs.

In Uniswap v3, LPs may use an \textit{NFT position manager} upon depositing liquidity, where LPs receive an NFT containing detailed information about the LP's position, including the specified price range and the amount of liquidity provided. The AMM smart contract tracks the fees owed to each LP's account, represented by these NFTs. Fee calculations consider the trade size, the LP's share of the pool, and the pool's fee specifications. When LPs burn their NFT tokens to redeem their deposits, the corresponding fees, as tracked by the so-called global fee variable in the AMM smart contract, are distributed to them. Currently, the method for calculating fees and the global fee variable are consistent across all LPs and are not influenced by the duration of liquidity provision. To adopt a two-tiered fee structure, we propose the following implementation steps:

\begin{itemize}
    \item \textbf{Commitment Duration Specification}: Upon depositing liquidity, an LP must specify a commitment duration level \(\phi \in \{0,1\}\). Here, \(\phi = 1\) signifies a commitment to provide liquidity for at least one block, defining them as a \textit{long-term LP}. Conversely, \(\phi = 0\) indicates no such commitment, defining them as a \textit{short-term LP}. The commitment mechanism thus classifies passive LPs as long-term and JIT LPs as short-term.
    \item \textbf{Enforcement of Commitment}: LPs are allowed to burn their pool tokens, i.e.\! withdraw, only if their deposit duration exceeds their specified \(\phi\). This mechanism enforces the commitment made at the time of deposit.
    \item \textbf{Differentiated Fee Variables}: The smart contract establishes a transfer rate parameter \(\lambda\in[0,1]\) and tracks two separate global fee variables, one for long-term LPs and another for short-term LPs.
    \item \textbf{Fee Distribution during Swaps}: When a swap occurs, the smart contract computes the short-term LPs' collective share of the pool. The transaction fee is first split pro-rata among the fee variables for short-term and long-term LPs. The smart contract then transfers $(1-\lambda)$ of the transaction fee in the short-term LPs' fee variable to the long-term LPs' fee variable.
    \item \textbf{Burning of Pool Tokens}: When a short/long-term LP burns their pool tokens, the smart contract computes the short/long-term LP's share of the liquidity pool, denoted by $w$, and the short/long-term LP's share of all short/long-term liquidity, denoted by $w_S$/$w_L$. The smart contract then allocates $w$ of the pool and $w_S$/$w_L$ of the short/long-term fee variable to the short/long-term LP. 
\end{itemize}

\section{Cournot Competition between Just-in-Time Liquidity Providers}

In this extension, we consider two competing JIT LPs. In practice, when multiple JIT LPs submit orders to sandwich a swap order with liquidity, the validator selects only one of them to confirm in the block, typically via bundles or gas fee bidding. This reduces to our baseline model with an additional process to determine which JIT LP's transactions to confirm. We instead study the case where the JIT LPs compete \`{a} la Cournot. Under Cournot competition, all of JIT LPs' deposit and withdrawal orders are confirmed. The events of Periods 3 and 4 are modified as follows.

\paragraph{Period 3*: Competing JIT LPs Decide Whether to Provide Liquidity.} In Period 3, two JIT LPs, indexed by $j\in\{1,2\}$ may arrive. Each JIT LP is endowed with $e_J=\bar{\nu}e_P$ risky coins and $pe_J$ stable coins, and arrives with probability $\pi\in[0,1]$ independently of their competitor. Conditional on arrival, a JIT LP views the swap order in the public mempool and decides on the amount of liquidity to provide. In addition to assuming that JIT LPs can determine whether an informed or uninformed trader submitted the swap order, we also suppose that each JIT LP receives a perfect signal about the arrival of their competitor. Conditional on arrival, JIT LP $j$ submits an order to deposit $d_J^{(j)}\in[0,e_J]$ risky coins and $pe_J$ stable coins right before the swap, and an order to withdraw its share of the pool right after the swap.

\paragraph{Period 4*: Settlement of Transactions Under Competition.} In Period 4, the block containing the swap and JIT liquidity transactions is validated, leading to the following sequence of events:
\begin{itemize}
    \item \textbf{JIT LP Deposits:} All orders from arriving JIT LPs to deposit into the liquidity pool are executed.
    \item \textbf{Swap Execution and Fee Distribution:} The swap order of the trader is executed. The passive LPs collectively receive a pro-rata-share of the transaction fee, where the passive LPs share is
    \begin{align*}
        &(1-\mathbbm{1}\{\text{JIT LP $1$ arrives}\})(1-\mathbbm{1}\{\text{JIT LP $2$ arrives}\}) \\[0.25\baselineskip]
        &\hspace{1cm}+\sum_{j\in\{1,2\}}\frac{d_P}{d_P+d_J^{(j)}}\cdot\mathbbm{1}\{\text{JIT LP $j$ arrives}\}\cdot(1-\mathbbm{1}\{\text{JIT LP $-j$ arrives}\}) \\[0.25\baselineskip]
        &\hspace{1cm}+\frac{d_P}{d_P+d_J^{(1)}+d_J^{(2)}}\cdot\mathbbm{1}\{\text{JIT LP $1$ arrives}\}\cdot\mathbbm{1}\{\text{JIT LP $2$ arrives}\}.
    \end{align*}
    The remaining share of fees is earned by the arriving JIT LP(s).
    \item \textbf{JIT LP Withdrawals:} Post-swap, all arriving JIT LPs withdraw their share of the pool. 
\end{itemize}

Our model of competition changes the strategy space of the JIT LPs, as only the JIT LPs can determine the presence of their competitor. Let $\Omega_j=\{\textsf{NA}_{-j},\textsf{A}_{-j}\}$, which correspond to the events of JIT LP $-j$ not arriving and arriving, respectively. The strategy space for JIT LP $j$ is now
\begin{align*}
    d_J:[0,e_P]\times([0,e_P]\times\Omega\to\R_+^2\backslash\R_{++}^2)\times\Omega\times\Omega_j\to[0,e_J].
\end{align*}
The main differences are the extra information received by a JIT LP about their competitor, and the introduction of an endowment for the JIT LP, which upper bounds their liquidity supply. These are crucial to assume since when more than one JIT LP arrives, both will want to deposit as much as possible, making a constraint on the deposit amount necessary to have equilibria. For the ease of computation and explainability, we make the following assumptions. \\[-0.75\baselineskip]

\noindent \textbf{\textsc{Assumption 2.}} \textit{We assume the following:}
\begin{itemize}
    \item \textit{The uninformed private value shock size $\zeta_U$ takes values in $[\underline{\zeta},\overline{\zeta}]$ where $\underline{\zeta}>\underline{\zeta}(f,\pi)$.}
    \item \textit{Given $[\underline{\zeta},\overline{\zeta}]$, we assume that $\bar{\nu}$ is sufficiently large such that when only one JIT LP arrives, their optimal deposit amount does not exceed their endowment.}
\end{itemize}

This assumption states that for any $\zeta_U\in[\underline{\zeta},\overline{\zeta}]$, the JIT LP's liquidity constraint does not bind if they are the only JIT LP in the market, and instead only binds when their competitor arrives. This allows us to isolate the upstream effects of monopoly power and Cournot competition on traders and passive LPs without concern for if the endowment will impact a JIT LP's monopoly power when their competitor is not present. We first consider the subgame between the informed or uninformed trader.

\begin{proposition}\label{thm:subgame-eq-comp}
There exists a unique non-trivial Nash equilibrium in the subgame between the traders and JIT LPs. In the equilibrium outcome:
\begin{itemize}
    \item The value of the informed trader's buy and sell swap orders is a fixed multiple $\mu_I$ of the amount of passive liquidity provided, where $\mu_I$ is the same constant as in Proposition \ref{thm:passive-br}.
    \item The value of the uninformed trader's buy and sell swap order is a fixed multiple $\mu_C(\pi)$ of the amount of passive liquidity provided where $\mu_C(\pi)$ is increasing in $\pi$ for $\bar{\nu}$ sufficiently high.
    \item Conditional on the arrival of a single JIT LP, that JIT LP provides liquidity only when facing an uninformed trader. The amount of liquidity provided is a fixed multiple $\nu_C(\pi)$ of the amount of passive liquidity provided.
    \item Conditional on the arrival of both JIT LPs, each JIT LP provides liquidity only when facing an uninformed trader. Both JIT LPs will deposit their entire endowment.
\end{itemize}
\end{proposition}

A crucial difference in this equilibrium, compared to that of the baseline model with a single JIT LP, is the behavior of JIT LPs conditional on both arriving. Each JIT LP has to vie for a share of the pool not only with the passive LPs, but also with their competitor. In the appendices, we show that as a lemma, a JIT LP's optimal deposit amount must give them a share of the pool that is over one-half. This suggests that the marginal benefit of a larger share always exceeds the marginal loss in price impact profits when the JIT LP's share is less than 50\%. Since the competing JIT LPs cannot both own more than half of the pool, they have an incentive to deposit more liquidity until their endowment is reached. 

A major contrast in characterizing the strategy of passive LPs is that their per-unit fee revenue now depends on the amount of passive liquidity provided itself. The expected utility of passive LP $i$ in the competitive case can be expressed as
\begin{align*}
    u_P(d_P^{(i)}) = p(\alpha\mathcal{C}+(1-\alpha)\cdot\mathcal{R}(\pi,d_P))d_P^{(i)}
\end{align*}
where $\mathcal{C}$ is the per-unit adverse selection loss, unchanged from before, and $\mathcal{R}(\pi,d_P)$ is now the liquidity-dependent expected per-unit fee revenue. There is the possibility for multiple equilibria to arise and for the equilibrium amount of passive liquidity provided to be in the interior of $[0,e_P]$, depending on the values of $\alpha$, $(1-\alpha)$, and $\mathcal{C}$, as the following result suggests.

\begin{proposition}\label{thm:passive-br-comp}
Let $\pi\in[0,1]$ and $d_P(k)=ke_P/N$ be the amount of liquidity provided by $k$ passive liquidity providers, where $k\in[N]$. Define $\mathcal{U}(\pi,d_P(k))=\alpha\mathcal{C}+(1-\alpha)\cdot\mathcal{R}(\pi,d_P(k))$.
\begin{itemize}
    \item For $k<N$, there exist subgame-perfect Nash equilibria where the amount of passive liquidity provided is $d_P^\star=d_P(k)$ if and only if $\mathcal{U}(\pi,d_P(k))\geq0$ and $\mathcal{U}(\pi,d_P(k+1))<0$.
    \item There exist subgame-perfect Nash equilibria where the amount of passive liquidity provided is $d_P^\star=e_P$ if and only if $\mathcal{U}(\pi,e_P)>0$.
\end{itemize}
\end{proposition}

Given an amount of passive liquidity, if the addition of an extra passive LP drives the per-unit utility of each passive LP below zero, then that amount of passive liquidity arises in a SPNE as no contributing passive LP profitably deviates by dropping out (under Assumption 1) and no non-contributing passive LP profitably deviates by opting in. This suggests that in the competitive case, there is now a secondary crowding-out effects where the passive LPs crowd each other out, as the appeal of contributing to the liquidity pool can become less attractive when more passive LPs join. To qualify this notion, we modify our definition of complementing and crowding out to accommodate the dependence of per-unit revenue on the amount of passive liquidity provided.

\begin{definition*}
Let $\pi\in(0,1]$ and $d_P\in[0,e_P]$. Suppose that some equilibria exist when the JIT LP's arrival probability is zero and $\pi$, holding the other exogenous parameters fixed. We say that the competing JIT LPs
\begin{itemize}
    \item \ul{complement} the passive LPs at arrival probability $\pi$ if $\mathcal{R}(0,d_P) \leq \mathcal{R}(\pi,d_P)$;
    \item \ul{crowd out} the passive LPs at arrival probability $\pi$ if $\mathcal{R}(0,d_P) > \mathcal{R}(\pi,d_P)$.
\end{itemize}
\end{definition*}

\begin{theorem}\label{thm:threshold-comp}
Let $\pi\in(0,1]$ and $d_P\in[0,e_P]$. For $\bar{\nu}$ sufficiently large, one of the following is true:
\begin{itemize}
    \item The JIT LPs complement the passive LPs at arrival probability $\pi$ for all $\zeta_U\in[\underline{\zeta},\overline{\zeta}]$.
    \item The JIT LPs crowd out the passive LPs at arrival probability $\pi$ for all $\zeta_U\in[\underline{\zeta},\overline{\zeta}]$.
    \item There exists $\zeta^\star_C(f,\pi,d_P)\in[\underline{\zeta},\overline{\zeta}]$ such that the JIT LP complements the passive LPs at arrival probability $\pi$ if and only if $\zeta_U\geq\zeta^\star_C(f,\pi,d_P)$.
\end{itemize}
Moreover, $\zeta^\star_C(f,\pi,d_P)\leq\zeta^\star(f,\pi)$ for all $\pi\in(0,1)$ where $\zeta^\star(f,\pi)$ is the threshold described in Theorem \ref{thm:threshold} that corresponds to the setting of a monopolist JIT LP.
\end{theorem}

Under Cournot competition, for any fixed $d_P\in[0,e_P]$, there also exists a threshold $\zeta^\star_C(f,\pi,d_P)$ that delineates regimes of $\zeta_U$ where the JIT LPs crowd out the passive LPs ($\zeta_U<\zeta^\star_C(f,\pi,d_P)$) and where the JIT LPs complement the passive LPs ($\zeta_U\geq\zeta^\star_C(f,\pi,d_P)$). The threshold where the JIT LP(s) switch from crowding out to complementing the passive LPs is thus lower under competition than under a monopoly in the JIT liquidity market. Although equilibrium amounts of passive liquidity $d_P^\star$ may also vary in $\zeta_U$, the result holds for all $d_P\in[0,e_P]$, making competition beneficial overall for the threshold. This is illustrated in Figure \ref{fig:threshold}, which shows that we have a region in the $(d_P,\zeta_U)$ plane where we have complementing only under competition.

\begin{figure}[h]
\centering
  \includegraphics[scale=0.67]{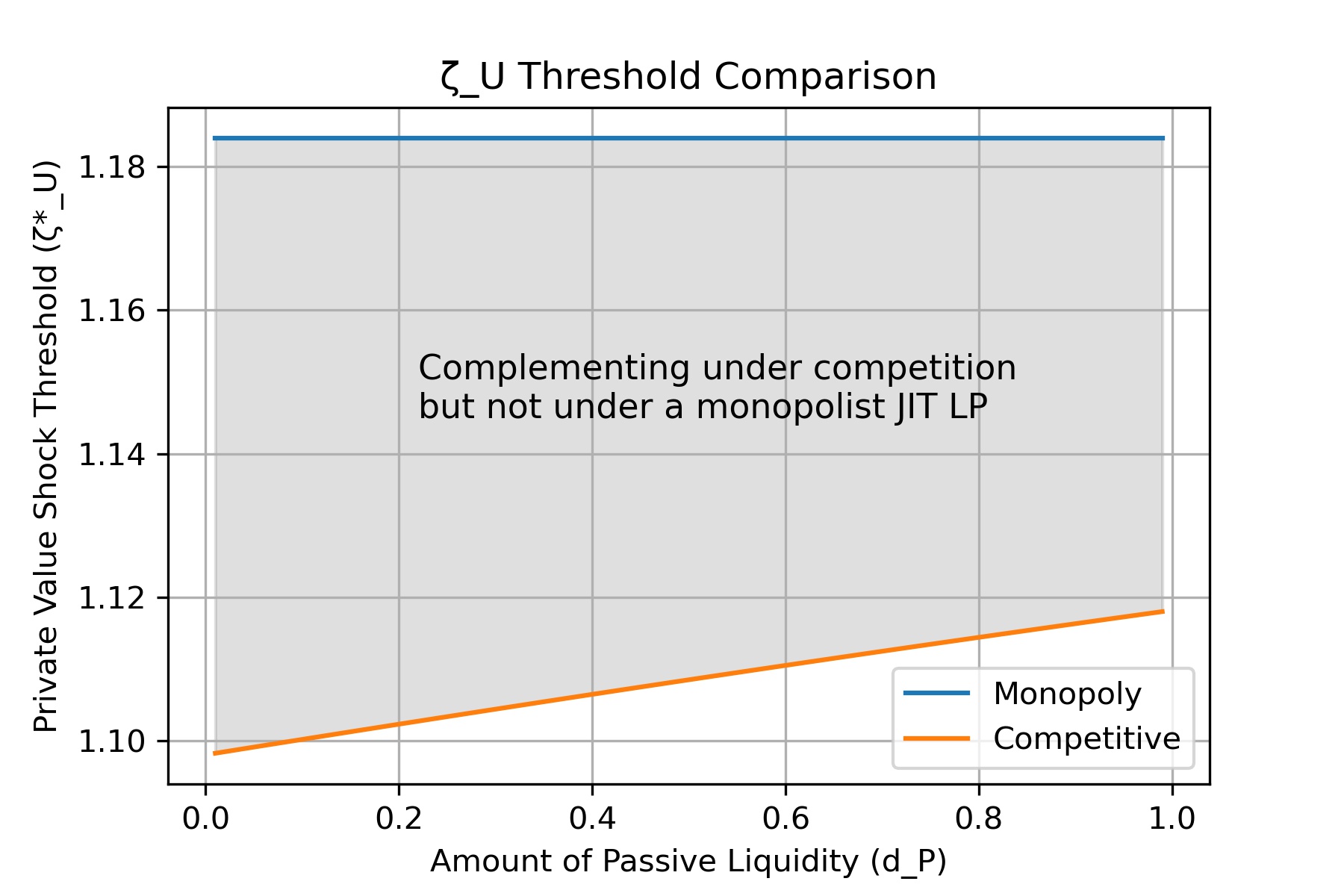} \\[0.25\baselineskip]
    \captionof{figure}{\centering \small The threshold $\zeta_U^\star$ between complementing and crowding out under competing JIT LPs and under a monopolist JIT LP given the amount of passive liquidity provided. We fixed $f=0.01$, $\pi=0.5$, $e_J=3$, and $e_P=1$.}
    \label{fig:threshold}
\end{figure}

The primary benefit of allowing Cournot competition in the JIT market is the increased sensitivity of the uninformed trader's demand as lower value of $\zeta_U$ is needed to induce a proportionally higher trading volume relative to the passive LPs' expected pool share. Cournot competition enables this phenomenon by having the JIT LPs compete with each other for a share of the liquidity pool, resulting in the JIT LPs depositing as much as possible in equilibrium. Notably, pool share competition does not occur when a JIT LP holds monopoly power or when JIT LPs submit via bundles (the current practice): the monopolist JIT LP only shares the pool with passive LPs that have already moved, and current practices result in other types of competition (e.g.\!\, bidding over the gas fee). 



\section{Conclusion}\label{sec:conclusion}

Just-in-time liquidity provision is a novel form of high-frequency market making uniquely enabled by the blockchain technology that underlies decentralized exchanges. As the concern of low pool depth is an ever-present issue for AMMs, it is important to understand when JIT liquidity is beneficial and how to make JIT liquidity beneficial for pool depth. We show that a JIT LP's ability to view pending transactions in the public mempool allows them to completely avoid informed order flows, and profit off of the price impact and fee revenue from uninformed order flows. However, if the prospect of extra liquidity does not induce uninformed traders to increase their trading volume by a sufficiently high amount, passive LPs will earn less expected profits in the presence of a JIT LP.  This can cause liquidity freezes, which can be potentially averted by transferring fee revenue from the JIT LP(s) to passive LPs or allowing for Cournot competition between JIT LPs. The two-tiered fee structure highlights an intricate balance between incentivizing both the JIT LP(s) and passive LPs to deposit, whereas competition seeks to increase the trader's demand. The careful design of AMMs taking into account the effects of JIT liquidity is thus key to enhancing pool depth as well as social welfare.


\newpage

\bibliographystyle{rfs}
\bibliography{defi}

\begin{thebibliography}{19}
\providecommand{\natexlab}[1]{#1}
\expandafter\ifx\csname urlstyle\endcsname\relax
  \providecommand{\doi}[1]{doi:\discretionary{}{}{}#1}\else
  \providecommand{\doi}{doi:\discretionary{}{}{}\begingroup \urlstyle{rm}\Url}\fi

\bibitem[{Adams et~al.(2023)Adams, Chan, Markovich, and Wan}]{AdamsUniswapSlippage}
Adams, A., B.~Y. Chan, S.~Markovich, and X.~Wan. 2023.
\newblock {The Costs of Swapping on the Uniswap Protocol}.
\newblock Papers 2309.13648, arXiv.org.

\bibitem[{Aoyagi and Ito(2021)}]{Aoyagi2021CoexistingEP}
Aoyagi, J., and Y.~Ito. 2021.
\newblock {Coexisting Exchange Platforms: Limit Order Books and Automated Market Makers}.
\newblock Working paper.

\bibitem[{Barbon and Ranaldo(2021)}]{dexcexcompare}
Barbon, A., and A.~Ranaldo. 2021.
\newblock {On the Quality of Cryptocurrency Markets: Centralized Versus Decentralized Exchanges}.
\newblock Working paper.

\bibitem[{BIS(2023)}]{BIS}
BIS. 2023.
\newblock {Project Mariana: cross-border exchange of wholesale CBDCs using automated market-makers}.

\bibitem[{Brunnermeier and Pedersen(2005)}]{PREY}
Brunnermeier, M.~K., and L.~H. Pedersen. 2005.
\newblock {Predatory Trading}.
\newblock \emph{The Journal of Finance} 60(4):1825--63.
\newblock \doi{https://doi.org/10.1111/j.1540-6261.2005.00781.x}.

\bibitem[{Capponi and Jia(2021)}]{AMMAdoption21}
Capponi, A., and R.~Jia. 2021.
\newblock The {{Adoption}} of {{Blockchain-based Decentralized Exchanges}}.
\newblock SSRN Scholarly Paper 3805095.
\newblock \doi{10.2139/ssrn.3805095}.

\bibitem[{Capponi, Jia, and Wang(2023)}]{CappMEV}
Capponi, A., R.~Jia, and Y.~Wang. 2023.
\newblock {Maximal Extractable Value and Allocative Inefficiencies in Public Blockchains}.
\newblock Working paper.

\bibitem[{CoinGecko(2023)}]{2023_Report}
CoinGecko. 2023.
\newblock {2023 Q3 Crypto Industry Report}.

\bibitem[{Han, Khapko, and Kyle(2014)}]{Han_Khapko_Kyle_2014}
Han, J., M.~Khapko, and A.~S. Kyle. 2014.
\newblock {Liquidity with High-Frequency Market Making}.
\newblock SSRN Scholarly Paper 2416396.
\newblock \doi{10.2139/ssrn.2416396}.

\bibitem[{Hasbrouck, Rivera, and Saleh(2022)}]{Hasbrouck_Rivera_Saleh_2022}
Hasbrouck, J., T.~J. Rivera, and F.~Saleh. 2022.
\newblock {The Need for Fees at a DEX: How Increases in Fees Can Increase DEX Trading Volume}.
\newblock SSRN Scholarly Paper 4192925.
\newblock \doi{10.2139/ssrn.4192925}.

\bibitem[{Hasbrouck, Rivera, and Saleh(2023)}]{Hasbrouck_Rivera_Saleh_2023}
---{}---{}---. 2023.
\newblock {An Economic Model of a Decentralized Exchange with Concentrated Liquidity}.
\newblock SSRN Scholarly Paper 4529513.
\newblock \doi{10.2139/ssrn.4529513}.

\bibitem[{Lehar and Parlour(2023)}]{anotherAMM2}
Lehar, A., and C.~Parlour. 2023.
\newblock {Decentralized Exchange: The Uniswap Automated Market Maker}.
\newblock \emph{Forthcoming in Journal of Finance} .

\bibitem[{Loesch et~al.(2021)Loesch, Hindman, Richardson, and Welch}]{Loesch21}
Loesch, S., N.~Hindman, M.~B. Richardson, and N.~Welch. 2021.
\newblock {Impermanent Loss in Uniswap v3}.
\newblock {arXiv}.

\bibitem[{Malinova and Park(2023)}]{Malinova_Park_2023}
Malinova, K., and A.~Park. 2023.
\newblock {Learning from DeFi: Would Automated Market Makers Improve Equity Trading?}
\newblock SSRN Scholarly Paper 4531670.
\newblock \doi{10.2139/ssrn.4531670}.

\bibitem[{Menkveld(2015)}]{MenkveldARFE}
Menkveld, A. 2015.
\newblock {The Economics of High-Frequency Trading: Taking Stock}.
\newblock \emph{Annual Review of Financial Economics} 8(1):1--24.

\bibitem[{Milionis et~al.(2022)Milionis, Moallemi, Roughgarden, and Zhang}]{Milionis_Moallemi_Roughgarden_Zhang}
Milionis, J., C.~C. Moallemi, T.~Roughgarden, and A.~L. Zhang. 2022.
\newblock {Automated Market Making and Loss-Versus-Rebalancing}.
\newblock Working paper.

\bibitem[{Park(2023)}]{anotherAMM}
Park, A. 2023.
\newblock {The Conceptual Flaws of Decentralized Automated Market Making}.
\newblock \emph{Management Science} 69(11):6731--51.

\bibitem[{Wan and Adams(2022)}]{AdamsUniswap}
Wan, X., and A.~Adams. 2022.
\newblock {Just-in-Time Liquidity on the Uniswap Protocol}.
\newblock White paper.

\bibitem[{Xiong et~al.(2023)Xiong, Wang, Knottenbelt, and Huth}]{DemystifyingJIT}
Xiong, X., Z.~Wang, W.~Knottenbelt, and M.~Huth. 2023.
\newblock {Demystifying Just-in-Time (JIT) Liquidity Attacks on Uniswap V3}.
\newblock \emph{Cryptology ePrint Archive} 2023/973.

\end{thebibliography}

\newpage

\appendix

\section{Proofs}

\subsection{Change of Variables}
We make the following changes of variables to facilitate derivations.
\begin{definition*} Define the following:
    \begin{itemize}
        \item Price-adjusted deposit size for passive LPs:  $\tilde{d}_{P}^{(i)} = p^{1/2}d^{(i)}_P$
        \item Price-adjusted endowment for passive LPs: $\Tilde{e}_{P}=p^{1/2}e_P$
        \item Aggregate price-adjusted deposit size for passive LPs: $\tilde{d}_{P}=\sum_{i\in[N]}\tilde{d}_{P}^{(i)}$
        \item Price-adjusted deposit size for JIT LPs: $\tilde{d}_{J} = p^{1/2}d_J$
    \end{itemize}
\end{definition*}

\subsection{Trading Functions}

\begin{definition*}
Let $\delta_S(r,d):\R_+^2\to\R_+$ be the quantity of stable coins that $r$ risky coins can be swapped for at a pool depth of $d$. Similarly, let $\delta_R(s,d):\R_+^2\to\R_+$ be the quantity of risky coins that $s$ stable coins can be swapped for at a pool depth of $d$. It then follows that
\begin{align*}
    \delta_S(r,d) = \frac{p^{1/2}dr}{p^{-1/2}d+r} = \frac{p\tilde{d}r}{\tilde{d}+r}; \\[0.25\baselineskip]
    \delta_R(s,d) = \frac{p^{-1/2}ds}{p^{1/2}d+s} = \frac{\tilde{d}s}{p\tilde{d}+s}.
\end{align*}
\end{definition*}

\subsection{Evolution of Pool Reserves}

The following table shows the evolution of the liquidity pool's reserves when the JIT LP does not and does arrive, respectively:
\begin{center}
\small
    \begin{tabular}{ccc}
        Description & $R$ & $S$  \\
        \hline
        $t=1$: passive LPs deposit & $\tilde{d}_P(1-p^{1/2}b^{-1/2})$ & $\tilde{d}_P(p-p^{1/2}a^{1/2})$ \rule{0pt}{3ex} \\
        $t=2$: trader submits swap order & $\tilde{d}_P(1-p^{1/2}b^{-1/2})$ & $\tilde{d}_P(p-p^{1/2}a^{1/2})$ \rule{0pt}{3ex} \\
        $t=3$: JIT LP does not arrive & $\tilde{d}_P(1-p^{1/2}b^{-1/2})$ & $\tilde{d}_P(p-p^{1/2}a^{1/2})$ \rule{0pt}{3ex} \\
        $t=4$: swap occurs & $\tilde{d}_P(1-p^{1/2}b^{-1/2})+q_R$ & $\tilde{d}_P(p-p^{1/2}a^{1/2})-\delta_S(q_R,\tilde{d}_P)$ \rule{0pt}{3ex} \\
        $t=5$: possible reverse trade & $\tilde{d}_P(1-p^{1/2}b^{-1/2})$ & $\tilde{d}_P(p-p^{1/2}a^{1/2})$ \rule{0pt}{3ex} \\
    \end{tabular} 
\end{center}
\begin{center}
\small
    \begin{tabular}{ccc}
        Description & $R$ & $S$  \\
        \hline
        $t=1$: passive LPs deposit & $\tilde{d}_P(1-p^{1/2}b^{-1/2})$ & $\tilde{d}_P(p-p^{1/2}a^{1/2})$ \rule{0pt}{3ex} \\
        $t=2$: trader submits swap order & $\tilde{d}_P(1-p^{1/2}b^{-1/2})$ & $\tilde{d}_P(p-p^{1/2}a^{1/2})$ \rule{0pt}{3ex} \\
        $t=3$: JIT LP arrives & $\tilde{d}_P(1-p^{1/2}b^{-1/2})$ & $\tilde{d}_P(p-p^{1/2}a^{1/2})$ \rule{0pt}{3ex} \\
        $t=4$: JIT LP deposits & $(\tilde{d}_P+\tilde{d}_J)(1-p^{1/2}b^{-1/2})$ & $(\tilde{d}_P+\tilde{d}_J)(p-p^{1/2}a^{1/2})$ \rule{0pt}{3ex} \\
        $t=4$: swap occurs & $(\tilde{d}_P+\tilde{d}_J)(1-p^{1/2}b^{-1/2})+q_R$ & $(\tilde{d}_P+\tilde{d}_J)(p-p^{1/2}a^{1/2})-\delta_S(q_R,\tilde{d}_P+\tilde{d}_J)$ \rule{0pt}{3ex} \\
        $t=4$: JIT LP withdraws & $\tilde{d}_P(1-p^{1/2}b^{-1/2})+\frac{\tilde{d}_P}{\tilde{d}_P+\tilde{d}_J}\cdot q_R$ & $\tilde{d}_P(p-p^{1/2}a^{1/2})-\frac{\tilde{d}_P}{\tilde{d}_P+\tilde{d}_J}\cdot\delta_S(q_R,\tilde{d}_P+\tilde{d}_J)$ \rule{0pt}{3ex} \\
        $t=5$: possible reverse trade & $\tilde{d}_P(1-p^{1/2}b^{-1/2})$ & $\tilde{d}_P(p-p^{1/2}a^{1/2})$ \rule{0pt}{3ex} \\
    \end{tabular} \\
    \vspace{0.5cm}
\end{center}

\subsection{Explicit Forms of Utilities}

The JIT LP's utility, conditional on arrival, is 
\begin{align*}
    u_J(\tilde{d}_J;\sigma_{-J},\omega) = \begin{dcases}
        \frac{\tilde{d}_J}{\tilde{d}_P+\tilde{d}_J}\left(\frac{p(1+f)}{\zeta}\cdot q_R(\tilde{d}_P;\omega)-\delta_S(q_R(\tilde{d};\omega),\tilde{d}_P+\tilde{d}_J)\right) & \omega=\mathsf{IS} \\[0.25\baselineskip]
        \frac{\tilde{d}_J}{\tilde{d}_P+\tilde{d}_J}((1+f)\cdot q_S(\tilde{d}_P;\omega)-\zeta p\cdot\delta_R(q_S(\tilde{d}_P;\omega),\tilde{d}_P+\tilde{d}_J)) & \omega=\mathsf{IB} \\[0.25\baselineskip]
        \frac{\tilde{d}_J}{\tilde{d}_P+\tilde{d}_J}(p(1+f)\cdot q_R(\tilde{d}_P;\omega)-\delta_S(q_R(\tilde{d};\omega),\tilde{d}_P+\tilde{d}_J)) & \omega=\mathsf{US} \\[0.25\baselineskip]
        \frac{\tilde{d}_J}{\tilde{d}_P+\tilde{d}_J}((1+f)\cdot q_S(\tilde{d}_P;\omega)-p\cdot\delta_R(q_S(\tilde{d}_P;\omega),\tilde{d}_P+\tilde{d}_J)) & \omega=\mathsf{UB}.
    \end{dcases}
\end{align*}
For $\omega\in\{\mathsf{IS},\mathsf{IB}\}$, the informed trader's utility is (suppressing arguments for the JIT LP's strategy)
\begin{align*}
    u_T((q_R,q_S);\sigma_{-T},\omega) = \begin{dcases}
        \pi \cdot\delta_S(q_R,\tilde{d}_P+\tilde{d}_J)+(1-\pi )\cdot\delta_S(q_R,\tilde{d}_P)-p'(\omega)\cdot(1+f)q_R & q_R>0 \\
        p'(\omega)[\pi \cdot\delta_R(q_S,\tilde{d}_P+\tilde{d}_J)+(1-\pi )\cdot\delta_R(q_S,\tilde{d}_P)]-(1+f)q_S & q_S>0 \\
        0 & \text{o.w.}
    \end{dcases}.
\end{align*}
For $\omega\in\{\mathsf{US},\mathsf{UB}\}$, the uninformed trader's utility is (suppressing arguments for the JIT LP's strategy)
\begin{align*}
    u_T((q_R,q_S);\sigma_{-T},\omega) = \begin{dcases}
        \pi \cdot\delta_S(q_R,\tilde{d}_P+\tilde{d}_J)+(1-\pi )\cdot\delta_S(q_R,\tilde{d}_P)-P(\omega)\cdot(1+f)q_R & q_R>0 \\
        P(\omega)[\pi \cdot\delta_R(q_S,\tilde{d}_P+\tilde{d}_J)+(1-\pi )\cdot\delta_R(q_S,\tilde{d}_P)]-(1+f)q_S & q_S>0 \\
        0 & \text{o.w.}
    \end{dcases}.
\end{align*}
The passive LPs' conditional total utility (suppressing arguments for the other agents' strategies) is then
\begin{gather*}
    u_P(\tilde{d}_P;\sigma_{-P},\omega) = \begin{dcases}
        \frac{\pi \tilde{d}_P}{\tilde{d}_P+\tilde{d}_J}\left(\frac{p(1+f)}{\zeta}q_R-\delta_S(q_R,\tilde{d}_P+\tilde{d}_J)\right) & \\
        \hspace{3cm}+\,(1-\pi )\left(\frac{p(1+f)}{\zeta}q_R-\delta_S(q_R,\tilde{d}_P)\right) & \omega=\mathsf{IS} \\[0.25\baselineskip]
        \frac{\pi \tilde{d}_P}{\tilde{d}_P+\tilde{d}_J}((1+f)q_S-\zeta p\cdot\delta_R(q_S,\tilde{d}_P+\tilde{d}_J)) & \\
        \hspace{3cm}+\,(1-\pi )((1+f)q_S-\zeta p\cdot\delta_R(q_S,\tilde{d}_P)) & \omega=\mathsf{IB} \\[0.25\baselineskip]
        \left[\frac{\pi \tilde{d}_P}{\tilde{d}_P+\tilde{d}_J}(pq_R+\delta_S(q_R,\tilde{d}_P+\tilde{d}_J))+(1-\pi )(pq_R+\delta_S(q_R,\tilde{d}_P))\right]f & \omega=\mathsf{US} \\[0.25\baselineskip]
        \left[\frac{\pi \tilde{d}_P}{\tilde{d}_P+\tilde{d}_J}(q_S+p\cdot\delta_R(q_S,\tilde{d}_P+\tilde{d}_J))+(1-\pi )(q_S+p\cdot \delta_R(q_S,\tilde{d}_P))\right]f & \omega=\mathsf{UB}
    \end{dcases}.
\end{gather*}
Since the passive LPs must decide on their strategies before $\omega$ is realized, their expected utility is
\begin{align*}
    u_P(\tilde{d}_P;\sigma_{-P}) &= \alpha\left[\psi\cdot u_P(\tilde{d}_P;\sigma_{-P},\mathsf{IS})+(1-\psi)\cdot u_P(\tilde{d}_P;\sigma_{-P},\mathsf{IB}) \right] \\
    &\hspace{1cm}+(1-\alpha)\left[\psi_U\cdot u_P(\tilde{d}_P;\sigma_{-P},\mathsf{US})+(1-\psi_U)\cdot u_P(\tilde{d}_P;\sigma_{-P},\mathsf{UB}) \right].
\end{align*} 

\subsection{Proofs of Propositions \ref{thm:subgame-eq} and \ref{thm:passive-br}}
It suffices to show the following: let $\tilde{d}_{P}\in[0,\tilde{e}_{P}]$ and  $\underline{\zeta}(\pi )$ be given by
\begin{align*}
    \underline{\zeta}(\pi ) = \frac{2(1+f)^3}{2+\pi  f(3+f)}.
\end{align*}
Suppose that $\zeta_U>\underline{\zeta}(\pi )$. Then there exists a non-trivial Nash equilibrium in the subgame between the traders and JIT LP. Define the following:
\begin{gather*}
    \tilde{\mu}_I = \zeta^{1/2}(1+f)^{-1/2}-1, \\[0.25\baselineskip]
    \tilde{\mu}(\pi ) = \arg\min_{\mu\in\R_+}\,\left|(1-\pi)\cdot\frac{1}{(1+\mu)^2}+\pi\cdot\frac{(2+\mu)\sqrt{(1+f)(1+\mu)}}{2(1+\mu)^2}-\frac{1+f}{\zeta_U}\right|\!, \\[0.25\baselineskip]
    \tilde{\nu}(\pi ) = \frac{f(1+\tilde{\mu}(\pi))+\tilde{\mu}(\pi)\sqrt{(1+f)(1+\tilde{\mu}(\pi))}}{\tilde{\mu}(\pi)-f}.
\end{gather*}
The equilibrium outcome is
\begin{gather*}
    (q_R,q_S)^\star(\tilde{d}_P;\omega) = \begin{dcases}
        (\tilde{\mu}_I\tilde{d}_P,0) & \omega=\mathsf{IS} \\
        (0,\tilde{\mu}_Ip\tilde{d}_P) & \omega=\mathsf{IB} \\
        (\tilde{\mu}(\pi)\cdot \tilde{d}_P,0) & \omega=\mathsf{US} \\
        (0,\tilde{\mu}(\pi)\cdot p\tilde{d}_P) & \omega=\mathsf{UB}
    \end{dcases}, \\[0.25\baselineskip]
    \tilde{d}_J^\star(\tilde{d}_P,(q_R,q_S)^\star(\tilde{d}_P;\omega);\omega) = \begin{dcases}
        0 & \omega\in\{\mathsf{IS},\mathsf{IB}\} \\
        \tilde{\nu}(\pi)\cdot\tilde{d}_P & \omega\in\{\mathsf{US},\mathsf{UB}\}
    \end{dcases}.
\end{gather*}
The multiples $\mu_I$, $\mu(\pi)$, and $\nu(\pi)$ are given by scaling $\tilde{\mu}$, $\tilde{\mu}(\pi)$, and $\tilde{\nu}(\pi)$ by $p^{1/2}$, respectively.

\begin{lemma}
Let $(q_R,q_S)$ be a strategy of the trader. If $(0,1)^\top(q_R,q_S)(\tilde{d}_{P};\mathsf{IS})>0$ for some $\tilde{d}_{P}\in[0,\tilde{e}_{P}]$, then for any JIT LP's strategy $\tilde{d}_{J}$, the profile $((q_R,q_S),\tilde{d}_{J})$ is not a Nash equilibrium in the trader--JIT LP subgame when $\omega=\mathsf{IS}$.
\begin{proof}
Since $p'=\zeta^{-1}p$ when $\omega=\mathsf{IS}$. If $(0,1)^\top(q_R,q_S)(\tilde{d}_{P};\mathsf{IS})>0$, then the trader swaps stable coins for risky coins, so their utility is
\begin{align*}
    u_T((0,q_S);\sigma_{-T},\mathsf{IS}) &= (1-\pi )(\zeta^{-1}p\cdot\delta_R(q_S,\tilde{d}_{P})-(1+f)q_S) \\
    &\hspace{1cm}+\pi (\zeta^{-1}p\cdot\delta_R(q_S,\tilde{d}_{P}+\tilde{d}_{J}(\tilde{d}_{P},(0,q_S);\mathsf{IS}))-(1+f)q_S)
\end{align*}
Conditional on the JIT LP not arriving, the trader's marginal utility is
\begin{align*}
    \frac{\d u_T}{\d q_S} &= \zeta^{-1}p\cdot\frac{\d\delta_R}{\d q_S}(q_S,\tilde{d}_{P})-(1+f) \\[0.25\baselineskip]
    &=\zeta^{-1}p\cdot\left.\frac{F_s}{F_r}\right|_{(\tilde{d}_{P}(1-p^{1/2}b^{-1/2})-\delta_R(q_S,\tilde{d}_{P}),\tilde{d}_{P}(p-p^{1/2}a^{1/2})+q_S)}-(1+f)< \zeta^{-1}p\cdot\frac{1}{p}-(1+f)<0
\end{align*}
since swapping stable coins for risky coins moves the stable-to-risky spot rate lower than $1/p$. A similar argument shows that the trader's marginal utility is also negative condition on the JIT LP arriving. Since there exists some $\tilde{d}_{P}\in[0,\tilde{e}_{P}]$ such that $(0,1)^\top(q_R,q_S)(\tilde{d}_{P};\mathsf{IS})>0$ and the trader's expected marginal utility is always negative, $(q_R,q_S)(\tilde{d}_P;\omega)=(0,0)$ is a profitable deviation.
\end{proof}
\end{lemma}

\begin{proposition}
Let $\tilde{d}_{P}\in[0,\tilde{e}_{P}]$. Then the unique non-trivial Nash equilibrium of the trader--JIT LP subgame when $\omega=\mathsf{IS}$ is given by
\begin{gather*}
    (q_R,q_S)^\star(\tilde{d}_{P};\mathsf{IS}) = ((\zeta^{1/2}(1+f)^{-1/2}-1)\tilde{d}_{P},0), \\
    \tilde{d}_{J}^\star(\tilde{d}_{P},(q_R,q_S);\mathsf{IS}) = 0.
\end{gather*}
\begin{proof}
By Lemma A.1, the trader swaps risky coins for stable coins on an equilibrium path when $\omega=\mathsf{IS}$, so his utility is
\begin{align*}
    u_{J}(\tilde{d}_{J};\sigma_{-J},\mathsf{IS}) = \frac{\tilde{d}_{J}}{\tilde{d}_{P}+\tilde{d}_{J}}(\zeta^{-1}p(1+f)\cdot q_R(\tilde{d}_{P};\mathsf{IS})-\delta_S(q_R(\tilde{d}_{P};\mathsf{IS}),\tilde{d}_{P}+\tilde{d}_{J}))
\end{align*}
When $\pi =1$, note that
\begin{align*}
    \zeta^{-1}p(1+f)\cdot q_R(\tilde{d}_{P};\mathsf{IS})-\delta_S(q_R(\tilde{d}_{P};\mathsf{IS}),\tilde{d}_{P}+\tilde{d}_{J})=-u_T((q_R,0);\sigma_{-T},\mathsf{IS})<0
\end{align*}
on an equilibrium path since if otherwise, i.e.\! $u_T((q_R,0);\sigma_{-T},\mathsf{IS})\leq0$, then by Assumption 1, the trader has a profitable deviation to $(0,0)$. Thus $0$ is the JIT LP's unique best-response to any trader's strategy on an equilibrium path. The trader's utility function is
\begin{align*}
    u_T((q_R,0);\sigma_{-T},\mathsf{IS}) &= (1-\pi )(\delta_S(q_R,\tilde{d}_{P})-\zeta^{-1}p(1+f)q_R) \\
    &\hspace{1cm}+\pi (\delta_S(q_R,\tilde{d}_{P}+\tilde{d}_{J}(\tilde{d}_{P},(q_R,0);\mathsf{IS}))-\zeta^{-1}p(1+f)q_R) \\
    &=\delta_S(q_R,\tilde{d}_{P})-\zeta^{-1}p(1+f)q_R
\end{align*}
since $\tilde{d}_{J}(\tilde{d}_{P},(q_R,0);\mathsf{IS})=0$. Note that $u_T$ is concave in $q_R$. The first-order condition is
\begin{align*}
    \frac{\d u_T}{\d q_R} = \frac{\d\delta_S}{\d q_R}(q_R,\tilde{d}_{P})-\zeta^{-1}p(1+f)=0
\end{align*}
The result follows by solving for $q_R$.
\end{proof}
\end{proposition}

\begin{lemma}
Let $(q_R,q_S)$ be a strategy of the trader. If $(0,1)^\top(q_R,q_S)(\tilde{d}_{P};\mathsf{IB})>0$ for some $\tilde{d}_{P}\in[0,\tilde{e}_{P}]$, then for any JIT LP's strategy $\tilde{d}_{J}$, the profile $((q_R,q_S),\tilde{d}_{J})$ is not a Nash equilibrium in the trader--JIT LP subgame when $\omega=\mathsf{IB}$.
\end{lemma}

\begin{proposition}
Let $\tilde{d}_{P}\in[0,\tilde{e}_{P}]$. Then the unique non-trivial Nash equilibrium of the trader--JIT LP subgame when $\omega=\mathsf{IB}$ is given by
\begin{gather*}
    (q_R,q_S)^\star(\tilde{d}_{P};\mathsf{IS}) = (0,(\zeta^{1/2}(1+f)^{-1/2}-1)p\tilde{d}_{P}), \\
    \tilde{d}_{J}^\star(\tilde{d}_{P},(q_R,q_S);\mathsf{IB}) = 0.
\end{gather*}
\end{proposition}

\noindent The proofs of results A.3 and A.4 are symmetric to those of results A.1 and A.2.

\begin{lemma}
If $\tilde{d}_{J}^\star$ is a stationary point for the JIT LP's utility function given $\tilde{d}_{P}\in\R_+$ and trader strategy $(q_R,q_S)$ when $\omega=\mathsf{US}$, then $\tilde{d}_{J}^\star/c$ is a stationary point for the JIT LP's utility function given $\tilde{d}_{P}/c$ and $(q_R,q_S)/c$ when $\omega=\mathsf{US}$ for any $c>0$.
\begin{proof}
When $\omega=\mathsf{US}$, the JIT LP's utility is
\begin{align*}
    u_{J}(\tilde{d}_{J};\sigma_{-J},\mathsf{US}) = \frac{\tilde{d}_{J}}{\tilde{d}_{P}+\tilde{d}_{J}}(p(1+f)\cdot q_R(\tilde{d}_{P};\mathsf{US})-\delta_S(q_R(\tilde{d}_{P};\mathsf{US}),\tilde{d}_{P}+\tilde{d}_{J})),
\end{align*}
so the first-order condition is
\begin{align*}
    \frac{\tilde{d}_{P}}{(\tilde{d}_{P}+\tilde{d}_{J})^2}(p(1+f)\cdot q_R(\tilde{d}_{P};\mathsf{US})-\delta_S(q_R(\tilde{d}_{P};\mathsf{US}),\tilde{d}_{P}+\tilde{d}_{J})) -\frac{\tilde{d}_{J}}{\tilde{d}_{P}+\tilde{d}_{J}}\cdot\frac{\d\delta_S}{\d q}(q_R(\tilde{d}_{P};\mathsf{US}),\tilde{d}_{P}+\tilde{d}_{J})=0
\end{align*}
By the 1-homogeneity of $\delta_S$ and 0-homogeneity of $\d\delta_S/\d q$,
we can express the first-order condition as
\begin{align*}
    \frac{\tilde{d}_{P}}{\tilde{d}_{P}+\tilde{d}_{J}}\left[p(1+f)\frac{q_R(\tilde{d}_{P};\mathsf{US})}{\tilde{d}_{P}+\tilde{d}_{J}}-\delta_S\left(\frac{q_R(\tilde{d}_{P};\mathsf{US})}{\tilde{d}_{P}+\tilde{d}_{J}},1\right)\right]-\frac{\tilde{d}_{J}}{\tilde{d}_{P}+\tilde{d}_{J}}\cdot\frac{\d\delta_S}{\d q}\left(\frac{q_R(\tilde{d}_{P};\mathsf{US})}{\tilde{d}_{P}+\tilde{d}_{J}},1\right)=0.
\end{align*}
Since $\tilde{d}_{J}^\star$ is a stationary point, the above equation is satisfied when $\tilde{d}_J=\tilde{d}_{J}^\star$. By the 1-homogeneity of $\delta_S$ and 0-homogeneity of $\d\delta_S/\d q$,
\begin{align*}
    \frac{\tilde{d}_{P}/c}{\tilde{d}_{P}/c+\tilde{d}_{J}^\star/c} &= \frac{\tilde{d}_{P}}{\tilde{d}_{P}+\tilde{d}_{J}^\star} \\
    &= \frac{\frac{\tilde{d}_{J}^\star}{\tilde{d}_{P}+\tilde{d}_{J}^\star}\cdot\frac{\d\delta_S}{\d q}\left(\frac{q_R(\tilde{d}_{P};\mathsf{US})}{\tilde{d}_{P}+\tilde{d}_{J}^\star},1\right)}{p(1+f)\frac{q_R(\tilde{d}_{P};\mathsf{US})}{\tilde{d}_{P}+\tilde{d}_{J}^\star}-\delta_S\left(\frac{q_R(\tilde{d}_{P};\mathsf{US})}{\tilde{d}_{P}+\tilde{d}_{J}^\star},1\right)} = \frac{\frac{\tilde{d}_{J}^\star/c}{\tilde{d}_{P}/c+\tilde{d}_{J}^\star/c}\cdot\frac{\d\delta_S}{\d q}\left(\frac{q_R(\tilde{d}_{P};\mathsf{US})/c}{\tilde{d}_{P}/c+\tilde{d}_{J}^\star/c},1\right)}{p(1+f)\frac{q_R(\tilde{d}_{P};\mathsf{US})/c}{\tilde{d}_{P}/c+\tilde{d}_{J}^\star/c}-\delta_S\left(\frac{q_R(\tilde{d}_{P};\mathsf{US})/c}{\tilde{d}_{P}/c+\tilde{d}_{J}^\star/c},1\right)}
\end{align*}
The result follows.
\end{proof}
\end{lemma}

\begin{lemma}
If the JIT LP's utility function has a unique maximum in $\tilde{d}_J$ on $\R_+$ and $q_R^\star$ solves the uninformed trader's problem given $\tilde{d}_{P}$ when $\omega=\mathsf{US}$, then $q_R^\star/c$ solves the trader's problem given $\tilde{d}_{P}/c$ when $\omega=\mathsf{US}$ for any $c>0$.
\end{lemma}
\begin{proof}
The trader's problem given $\tilde{d}_{P}$ is
\begin{align*}
    \max_{q_R\in\R_+} \ & (1-\pi)\cdot\delta_S(q_R,\tilde{d}_P) +\pi\cdot\delta_S(q_R,\tilde{d}_{P}+\tilde{d}_{J}(\tilde{d}_{P},(q_R,0);\mathsf{US}))-\zeta_U^{-1}p(1+f)q_R
\end{align*}
and the trader's problem given $\tilde{d}_{P}/c$ is
\begin{align*}
    \max_{q_R\in\R_+} \ & (1-\pi)\cdot\delta_S(q_R,\tilde{d}_P/c)+\pi\cdot\delta_S(q_R,\tilde{d}_{P}/c+\tilde{d}_{J}(\tilde{d}_{P}/c,(q_R,0);\mathsf{US}))-\zeta_U^{-1}p(1+f)q_R
\end{align*}
Let $\tilde{q}_R=cq_R$. By Lemma A.5, the unique maximum assumption, and the 1-homogeneity of $\delta_S$, this problem is equivalent to
\begin{align*}
    \max_{\tilde{q}_R\in\R_+} \ c^{-1}(1-\pi)\cdot\delta_S(\tilde{q}_R,\tilde{d}_P)+ c^{-1}\pi\cdot\delta_S(\tilde{q}_R,\tilde{d}_{P}+\tilde{d}_{J}(\tilde{d}_{P},(\tilde{q}_R,0);\mathsf{US}))-c^{-1}\zeta_U^{-1}p(1+f)\tilde{q}_R
\end{align*}
The objective is $c^{-1}$ times the original objective, so the result follows.
\end{proof}

\begin{lemma}
Then JIT LP's utility function given $\tilde{d}_{P}\in[0,\tilde{e}_{P}]$ and a trader strategy $(q_R,q_S)$ when $\omega=\mathsf{US}$ has a unique maximum in $\tilde{d}_{J}$ on $\R_+$ if and only if $q_R(\tilde{d}_{P};\mathsf{US})>f\tilde{d}_{P}$.
\begin{proof}
The explicit form of the JIT LP's utility (suppressing arguments for the trader's strategy) is
\begin{gather*}
    u_{J}(\tilde{d}_{J};\sigma_{-J},\mathsf{US}) = \frac{\tilde{d}_{J}}{\tilde{d}_{P}+\tilde{d}_{J}}\left(p(1+f)q_R-\frac{p(\tilde{d}_{P}+\tilde{d}_{J})q_R}{\tilde{d}_{P}+\tilde{d}_{J}+q_R}\right)
\end{gather*}
with partial derivative
\begin{align*}
    \frac{\d u_{J}}{\d \tilde{d}_{J}} = \frac{\tilde{d}_{P}}{(\tilde{d}_{P}+\tilde{d}_{J})^2}\left(p(1+f)q_R-\frac{p(\tilde{d}_{P}+\tilde{d}_{J})q_R}{\tilde{d}_{P}+\tilde{d}_{J}+q_R}\right)-\frac{\tilde{d}_{J}}{\tilde{d}_{P}+\tilde{d}_{J}}\left(\frac{q_R}{\tilde{d}_{P}+\tilde{d}_{J}+q_R}\right)^2.
\end{align*}
Note that the sign of the partial derivative depends on
\begin{align*}
    M_{J}(\tilde{d}_{J})\equiv\frac{(1+f)\tilde{d}_{P}}{(\tilde{d}_{P}+\tilde{d}_{J})^2}-\frac{\tilde{d}_{P}+q_R}{(\tilde{d}_{P}+\tilde{d}_{J}+q_R)^2}.
\end{align*}
If $q_R\leq f\tilde{d}_{P}$, then $M_{J}(\tilde{d}_{J})>0$ for all $\tilde{d}_{J}\in\R_+$, so $u_{J}$ does not have a maximum for $\tilde{d}_{J}\in\R_+$. If $q_R>f\tilde{d}_{P}$, then solving $M_{J}(\tilde{d}_{J})=0$ yields two solutions:
\begin{align*}
    \tilde{d}_{J}^\star = \frac{f\tilde{d}_{P}(\tilde{d}_{P}+q_R)\pm q_R\sqrt{(1+f)\tilde{d}_{P}(\tilde{d}_{P}+q_R)}}{q_R-f\tilde{d}_{P}}
\end{align*}
Since $q_R>f\tilde{d}_{P}$ implies $f\tilde{d}_{P}(\tilde{d}_{P}+q_R)-q_R\sqrt{(1+f)\tilde{d}_{P}(\tilde{d}_{P}+q_R)}<0$, the negative root solution is extraneous. Since $M_{J}$ is continuous in $\tilde{d}_{J}$, $M_{J}(0)>0$, and $M_{J}(\tilde{d}_{J})=0$ has one solution on $\R_+$, it follows that $u_{J}$ has a unique maximum for $\tilde{d}_{J}\in\R_+$.
\end{proof}
\end{lemma}

\begin{proposition}
Let $\tilde{d}_{P}\in[0,\tilde{e}_{P}]$. If
\begin{align*}
    \zeta_U > \underline{\zeta}(f,\pi)=\frac{2(1+f)^3}{2+\pi f(3+f)},
\end{align*}
then the unique non-trivial Nash equilibrium of the trader--JIT LP subgame when $\omega=\mathsf{US}$ is given by
\begin{gather*}
    (q_R,q_S)^\star(\tilde{d}_{P};\mathsf{US}) = (\tilde{\mu}(\pi)\cdot \tilde{d}_P,0) \\[0.25\baselineskip]
    \tilde{d}_{J}^\star(\tilde{d}_{P},(q_R,q_S);\mathsf{US}) = \begin{dcases}
        0 & \tilde{d}_{P}=0, \\
        \frac{f\tilde{d}_{P}(\tilde{d}_{P}+q_R)+q_R\sqrt{(1+f)\tilde{d}_{P}(\tilde{d}_{P}+q_R)}}{q_R-f\tilde{d}_{P}} & \tilde{d}_{P} > 0.
    \end{dcases}
\end{gather*}
Otherwise, if $\zeta_U \leq 2(1+f)^2(2+f)^{-1} $, then there does not exist a non-trivial Nash equilibrium.
\begin{proof}
When $\tilde{d}_{P}=0$, the result follows from Assumption 1. When $\tilde{d}_{P}>0$ and $\zeta_U > 2(1+f)^2(2+f)^{-1}$, the JIT LP's best-response follows from Lemma A.7. Since the trader accounts for this best-response, the trader's utility function when $\omega=\mathsf{US}$ is
\begin{align*}
    u_T((q_R,0);\sigma_{-T},\mathsf{US}) = (1-\pi )\frac{p\tilde{d}_{P}q_R}{\tilde{d}_{P}+q_R}+\pi \cdot\frac{p(\tilde{d}_{P}+\tilde{d}_{J}(\tilde{d}_{P},(q_R,0);\mathsf{US}))q_R}{\tilde{d}_{P}+\tilde{d}_{J}(\tilde{d}_{P},(q_R,0);\mathsf{US})+q_R}-\zeta_U^{-1}p(1+f)q_R.
\end{align*}
Using the expression from Lemma A.7, this has partial derivative
\begin{align*}
    \frac{\d u_T}{\d q_R} &= (1-\pi )p\cdot\frac{\tilde{d}_{P}^2}{(\tilde{d}_{P}+q_R)^2}+\pi  p\cdot\frac{(2\tilde{d}_{P}+q_R)\sqrt{(1+f)\tilde{d}_{P}(\tilde{d}_{P}+q_R)}}{2(\tilde{d}_{P}+q_R)^2}-\frac{p(1+f)}{\zeta_U} \\[0.25\baselineskip]
    &= \left[(1-\pi )\cdot\frac{\tilde{d}_{P}^2}{(\tilde{d}_{P}+q_R)^2}+\pi\cdot\frac{(2\tilde{d}_{P}+q_R)\sqrt{(1+f)\tilde{d}_{P}(\tilde{d}_{P}+q_R)}}{2(\tilde{d}_{P}+q_R)^2}-\frac{1+f}{\zeta_U}\right]p.
\end{align*}
By Lemma A.6, we can normalize $\tilde{d}_{P}$ to 1, so the trader's first-order condition is
\begin{align*}
    M_T(\mu) \equiv (1-\pi)\cdot\frac{1}{(1+\mu)^2}+\pi\cdot\frac{(2+\mu)\sqrt{(1+f)(1+\mu)}}{2(1+\mu)^2} = \frac{1+f}{\zeta_U}
\end{align*}
Since $M_T(0)>1$, $\lim_{\mu\to\infty}M_T(\mu)=0$, and $M_T$ is decreasing in $\mu$, if $\zeta_U>1+f$, then there is a unique solution to the equation $M_T(\mu)=\zeta_U^{-1}(1+f)$ that we defined as $\mu(\pi)$. Note that
\begin{align*}
    M_T(f)\leq\zeta_U^{-1}(1+f)\iff\zeta_U\leq \frac{2(1+f)^3}{2+\pi f(3+f)},
\end{align*}
so $\zeta_U>\underline{\zeta}(f,\pi)\iff \mu>f\iff q_R > f\tilde{d}_{P}$. When $\tilde{d}_{P}>0$ and $q_R\leq f\tilde{d}_{P}$, Lemma A.7 implies that the JIT LP's best-response is to deposit at an infinite quantity, so there does exist a non-trivial Nash equilibrium.
\end{proof}
\end{proposition}

\begin{lemma}
If $\tilde{d}_{J}^\star$ is a stationary point for the JIT LP's utility function given $\tilde{d}_{P}\in\R_+$ and trader strategy $(q_R,q_S)$ when $\omega=\mathsf{UB}$, then $\tilde{d}_{J}^\star/c$ is a stationary point for the JIT LP's utility function given $\tilde{d}_{P}/c$ and $(q_R,q_S)/c$ when $\omega=\mathsf{UB}$ for any $c>0$.
\end{lemma}

\begin{lemma}
If the JIT LP's utility function has a unique maximum for $\tilde{d}_{J}\in\R_+$ and $q_S^\star$ solves the uninformed trader's problem given $\tilde{d}_{P}$ when $\omega=\mathsf{UB}$, then $q_S^\star/c$ solves the trader's problem given $\tilde{d}_{P}/c$ when $\omega=\mathsf{UB}$ for any $c>0$.
\end{lemma}

\begin{lemma}
Then JIT LP's utility function given $\tilde{d}_{P}\in[0,\tilde{e}_{P}]$ and a trader strategy $(q_R,q_S)$ when $\omega=\mathsf{UB}$ has a unique maximum for $\tilde{d}_{J}\in\R_+$ if and only if $q_S(\tilde{d}_{P};\mathsf{UB})>fp\tilde{d}_{P}$.
\end{lemma}

\begin{proposition}
Let $\tilde{d}_{P}\in[0,\tilde{e}_{P}]$. If
\begin{align*}
    \zeta_U > \underline{\zeta}(f,\pi)=\frac{2(1+f)^3}{2+\pi f(3+f)},
\end{align*}
then the unique non-trivial Nash equilibrium of the trader--JIT LP subgame when $\omega=\mathsf{UB}$ is given by
\begin{gather*}
    (q_R,q_S)^\star(\tilde{d}_{P};\mathsf{US}) = (0,\tilde{\mu}(\pi)\cdot p\tilde{d}_P) \\[0.25\baselineskip]
    \tilde{d}_{J}^\star(\tilde{d}_{P},(q_R,q_S);\mathsf{US}) = \begin{dcases}
        0 & \tilde{d}_{P}=0, \\
        \frac{fp\tilde{d}_{P}(p\tilde{d}_{P}+q_S)+q_S\sqrt{(1+f)p\tilde{d}_{P}(p\tilde{d}_{P}+q_S)}}{q_S-fp\tilde{d}_{P}} & \tilde{d}_{P} > 0.
    \end{dcases}
\end{gather*}
Otherwise, if $\zeta_U \leq 2(1+f)^2(2+f)^{-1} $, there does not exist a non-trivial Nash equilibrium.
\end{proposition}

\noindent The proofs of results A.9--A.12 are symmetric to those of results A.5--A.8. It is helpful to use the parameterizations $\tilde{v}_{P}=p\tilde{d}_{P}$ and $\tilde{v}_{J}=p\tilde{d}_{J}$.

It remains to find the passive LPs' best-responses when the remaining agents play the unique non-trivial Nash equilibrium of the trader--JIT LP subgame given $\tilde{d}_{P}$, the amount of passive liquidity provided, denoted $\sigma_{-P}^\star(\tilde{d}_{P})$. Substituting the explicit forms of $\delta_S$, $\delta_R$, and $\sigma_{-P}^\star(\tilde{d}_{P})$ into the expressions for $u_{P}(\tilde{d}_{P};\sigma_{-P}^\star(\tilde{d}_{P}))$ yields
\begin{align*}
    u_{P}(\tilde{d}_{P};\sigma_{-P}^\star) = \mathcal{U}(\pi) = (\alpha\mathcal{C}+(1-\alpha)\cdot\mathcal{R}(\pi))p\tilde{d}_{P}
\end{align*}
where $\mathcal{C}$ and $\mathcal{R}(\pi)$ are given by
\begin{gather*}
    \mathcal{C} = -\left[\psi\left(1-\frac{1+f}{\zeta}\right)^2+(1-\psi\left(\sqrt{\zeta}-\sqrt{1+f}\right)^2\right], \\[0.25\baselineskip]
    \mathcal{R}(\pi) = \left[(1-\pi)\left(\tilde{\mu}(\pi)+\frac{\tilde{\mu}(\pi)}{1+\tilde{\mu}(\pi)}\right)+\pi\cdot\frac{1}{1+\tilde{\nu}(\pi)}\left(\tilde{\mu}(\pi)+\frac{(1+\tilde{\nu}(\pi))\cdot\tilde{\mu}(\pi)}{1+\tilde{\nu}(\pi)+\tilde{\mu}(\pi)}\right)\right]f.
\end{gather*}
Let $\mathbf{\tilde{d}}_{P}\in[0,\tilde{e}_{P}]^N$ be a vector of price-adjusted risky coin deposit amounts for the passive LPs.
\begin{itemize}
    \item Suppose that $\mathcal{U}(\pi)<0$. If $\mathbf{\tilde{d}}_{P}\neq\mathbf{0}$, then there exists $i\in[N]$ such that $\mathbf{\tilde{d}}_{P}^{(i)}>0$, so passive LP $i$ has a profitable deviation to $0$. Thus $\mathbf{\tilde{d}}_{P}$ cannot be a strategy in the SPNE. However, if $\mathbf{\tilde{d}}_{P}=\mathbf{0}$, then note that no passive LP can profitably deviate under $\mathbf{0}$. Thus $\mathbf{0}$ is the unique SPNE strategy for the passive LPs here.
    \item Suppose that $\mathcal{U}(\pi)\geq0$. If $\mathbf{\tilde{d}}_{P}\neq \tilde{e}_{P}/N\cdot\mathbf{1}$, then there exists $i\in[N]$ such that $\mathbf{\tilde{d}}_{P}^{(i)}<\tilde{e}_{P}/N$, so passive LP $i$ has a profitable deviation to $\tilde{e}_{P}/N$. Thus $\mathbf{\tilde{d}}_{P}$ cannot be a strategy in the SPNE. However, if $\mathbf{\tilde{d}}_{P}=\tilde{e}_P/N\cdot\mathbf{1}$, then note that no passive LP can profitably deviate under $\tilde{e}_{P}/N\cdot\mathbf{1}$, it follows that $\tilde{e}_{P}/N\cdot\mathbf{1}$ is the unique SPNE strategy for the passive LPs in this case.
\end{itemize}

\subsection{Proof of Theorem \ref{thm:threshold}}

The passive LPs' per-unit fee revenue scaled by $f^{-1}$ in the complete absence of a JIT LP ($\pi=0$) is the total trading volume attracted by the DEX when an uninformed trader arrives when $\pi=0$, which is given by
\begin{gather*}
    V_0 \equiv \tilde{\mu}(0)+\frac{\tilde{\mu}(0)}{1+\tilde{\mu}(0)} = \left(\frac{\zeta_U}{1+f}\right)^{1/2}-\left(\frac{1+f}{\zeta_U}\right)^{1/2} = \frac{\mathcal{R}(0)}{f}.
\end{gather*}
The passive LPs' per-unit fee revenue scaled by $f^{-1}$ given a JIT LP arrival probability of $\pi$ is the total trading volume attracted by the AMM when an uninformed trader arrives given a JIT LP arrival probability of $\pi$, which can be expressed as a function of the equilibrium uninformed trade size $\mu\equiv\tilde{\mu}(\pi)$:
\begin{align*}
    V(\mu) &\equiv (1-\pi)\left(\mu+\frac{\mu}{1+\mu}\right)+\pi\cdot\frac{1}{1+\nu(\mu)}\left(\mu+\frac{(1+\nu(\mu))\mu}{1+\nu(\mu)+\mu}\right) \\[0.25\baselineskip]
    &= (1-\pi)\left(\mu+\frac{\mu}{1+\mu}\right)+\pi\left[\left(\frac{1+\mu}{1+f}\right)^{1/2}-\left(\frac{1+f}{1+\mu}\right)^{1/2}\right] = \frac{\mathcal{R}(\pi)}{f}.
\end{align*}
Here $\nu(\mu)$ is the JIT LP's normalized deposit size written as a function of the uninformed trader's normalized swap size:
\begin{align*}
    \nu(\mu) = \frac{f(1+\mu)+\mu\sqrt{(1+f)(1+\mu)}}{\mu-f}.
\end{align*}
Recall the following:
\begin{itemize}
    \item The uninformed trader's first-order condition is 
    \begin{align*}
        M_T(\mu;\pi)=(1-\pi)\cdot\frac{1}{(1+\mu)^2}+\pi\cdot\frac{(2+\mu)\sqrt{(1+f)(1+\mu)}}{2(1+\mu)^2} = \frac{1+f}{\zeta_U}
    \end{align*}
    \item By the product and chain rules, if $g,h:\R_+\to\R_+$ such that $g(x)$ is decreasing and $h(x)$ is increasing, then $(g\cdot h)(x)$ is increasing if and only if $\d\log h(x)/\d x \geq -\d\log g(x)/\d x$.
    \item By the definition of covariance, $X$ and $Y$ are positive random variables, then $\E[X/Y]\leq\E[X]/\E[Y]$ if and only if $\Cov[X/Y,Y]\geq0$.
\end{itemize}
We require the following lemma.
\begin{lemma}
Fix $f\in\R_+$ and $\pi\in[0,1]$. Then $M_T(\mu)\cdot(2+V(\mu)^2+V(\mu)\sqrt{4+V(\mu)^2})$ is increasing in $\mu$ on $\mu\in(f,\infty)$.
\end{lemma}
\begin{proof}
It suffices to show that
\begin{align*}
    \frac{\d\log(2+V(\mu)^2+V(\mu)\sqrt{4+V(\mu)^2})}{\d\mu} &\geq -\frac{\d\log M_T(\mu)}{\d\mu}.
\end{align*}
By the chain rule, these simplify to 
\begin{align*}
    \frac{2\cdot V'(\mu)}{\sqrt{4+V(\mu)^2}} &\geq \frac{-M_T'(\mu)}{M_T(\mu)}.
\end{align*}
Let $(\Omega,\F,\P)$ where $\Omega=\{\omega_0,\omega_1\}$, $\F=2^\Omega$, $\P(\omega_0)=1-\pi$, and $\P(\omega_1)=\pi$ be a probability space. Define random variables $\Phi,\tilde{\Phi}:\Omega\to\R$ such that
\begin{gather*}
    \Phi(\omega) = \begin{dcases}
        \frac{1}{(1+\mu)^2} & \omega=\omega_0 \\[0.25\baselineskip]
        \frac{(2+\mu)\sqrt{(1+f)(1+\mu)}}{2(1+\mu)^2} & \omega=\omega_1
    \end{dcases} \\[0.25\baselineskip]
    \tilde{\Phi}(\omega) = \begin{dcases}
        \frac{2}{(1+\mu)^3} & \omega=\omega_0 \\[0.25\baselineskip]
        \frac{(4+\mu)\sqrt{(1+f)(1+\mu)}}{4(1+\mu)^2} & \omega=\omega_1.
    \end{dcases}
\end{gather*}
Then we have
\begin{align*}
    \frac{\tilde{\Phi}(\omega)}{\Phi(\omega)} = \begin{dcases}
        \frac{2}{1+\mu} & \omega=\omega_0 \\[0.25\baselineskip]
        \frac{2}{(1+\mu)(2+\mu)} & \omega=\omega_1.
    \end{dcases}
\end{align*}
Note that $\frac{\tilde{\Phi}(\omega_0)}{\Phi(\omega_0)} \geq \frac{\tilde{\Phi}(\omega_1)}{\Phi(\omega_1)}$ while $\Phi(\omega_0) \leq \Phi(\omega_1)$ on $\mu\in(f,\infty)$, so it follows that 
\begin{align*}
    \Cov\left[\frac{\tilde{\Phi}}{\Phi},\Phi\right]&\leq0 \\[0.25\baselineskip]
    \frac{\E[\tilde{\Phi}]}{\E[\Phi]} &\leq \E\left[\frac{\tilde{\Phi}}{\Phi}\right].
\end{align*}
Define random variables $\Psi,\tilde{\Psi}:\Omega\to\R$ (on the same probability space) such that
\begin{gather*}
    \Psi(\omega) = \begin{dcases}
        \sqrt{4+\left(\mu+\frac{\mu}{1+\mu}\right)^2} = 1+\mu+\frac{1}{1+\mu} & \omega=\omega_0 \\[0.24\baselineskip]
        \sqrt{4+\left(\sqrt{\frac{1+\mu}{1+f}}-\sqrt{\frac{1+f}{1+\mu}}\right)^2} = \sqrt{\frac{1+\mu}{1+f}}+\sqrt{\frac{1+f}{1+\mu}} & \omega=\omega_1 
    \end{dcases} \\[0.24\baselineskip]
    \tilde{\Psi}(\omega) = \begin{dcases}
        2\left(1+\frac{1}{(1+\mu)^2}\right) & \omega=\omega_0 \\[0.24\baselineskip]
        \frac{1}{1+\mu}\left(\sqrt{\frac{1+\mu}{1+f}}+\sqrt{\frac{1+f}{1+\mu}}\right) & \omega=\omega_1
    \end{dcases}.
\end{gather*}
Then we have
\begin{align*}
    \frac{\tilde{\Psi}(\omega)}{\Psi(\omega)} = \begin{dcases}
        \frac{2}{1+\mu} & \omega=\omega_0 \\[0.24\baselineskip]
        \frac{1}{1+\mu} & \omega=\omega_1.
    \end{dcases}
\end{align*}
Note that $\frac{\tilde{\Psi}(\omega_0)}{\Psi(\omega_0)}\geq\frac{\Tilde{\Psi}(\omega_1)}{\tilde{\Psi}(\omega_1)}$ and $\Psi(\omega_0)\geq\Psi(\omega_1)$ on $\mu\in(f,\infty)$, so it follows that
\begin{align*}
    \Cov\left[\frac{\Tilde{\Psi}}{\Psi},\Psi\right] &\geq 0 \\[0.24\baselineskip]
    \frac{\E[\tilde{\Psi}]}{\E[\Psi]} &\geq \E\left[\frac{\tilde{\Psi}}{\Psi}\right].
\end{align*}
Observe that $\E[\tilde{\Phi}/\Phi] \leq \E[\tilde{\Psi}/\Psi]$. Chaining everything together yields
\begin{align*}
    \frac{2\cdot V'(\mu)}{\sqrt{4+V(\mu)^2}} \geq \frac{\E[\tilde{\Psi}]}{\E[\Psi]} &\geq \E\left[\frac{\tilde{\Psi}}{\Psi}\right] \geq \E\left[\frac{\tilde{\Phi}}{\Phi}\right] \geq \frac{\E[\tilde{\Phi}]}{\E[\Phi]} = \frac{-M_T'(\mu)}{M_T(\mu)}
\end{align*}
as desired, noting that $\E[\Psi]\geq\sqrt{4+V(\mu)^2}$ due to Jensen's inequality.
\end{proof}

Fix $\mu\in(f,\infty)$. Let $\zeta_U(\mu)$ be the private value shock size such that the uninformed trader's equilibrium swap size is $\mu$. The first-order condition implies that
\begin{align*}
    \zeta_U(\mu) = \frac{1+f}{M_T(\mu)},
\end{align*}
so $\zeta_U(\mu)$ is well-defined. Since $V_0$ is increasing in $\zeta_U$ and $\lim_{\zeta_U\to\infty}V_0=\infty$, there exists a unique value of $\zeta_U\in(1+f,\infty)$ such that $V(\mu)=V_0$ under $\zeta_U$; let us denote it $\bar{\zeta}_U(\mu)$. Then $V(\mu)\geq V_0$ if and only if $\zeta_U(\mu)\leq\bar{\zeta}_U(\mu)$. The expression for $V_0$ yields
\begin{gather*}
    \bar{\zeta}_U(\mu) = (1+f)\left(1+\frac{V(\mu)^2+V(\mu)\cdot\sqrt{4+V(\mu)^2}}{2}\right).
\end{gather*}
Note that
\begin{align*}
    2 &\leq M_T(\mu)\cdot\left(2+V^2(\mu)+V(\mu)\cdot\sqrt{4+V^2(\mu)}\right) \\[0.25\baselineskip]
    \frac{1}{M_T(\mu)} &\leq \frac{2+V^2(\mu)+V(\mu)\cdot\sqrt{4+V^2(\mu)}}{2} \\[0.25\baselineskip]
    \zeta_U(\mu) &\leq \bar{\zeta}_U(\mu).
\end{align*}
We now have two cases:
\begin{itemize}
    \item If $M_T(f)\cdot\left(2+V^2(f)+V(f)\cdot\sqrt{4+V^2(f)}\right)\geq 2$, then by Lemma A.13, we have $\zeta_U(\mu)\leq\bar{\zeta}_U(\mu)$, corresponding to the first case of the theorem.
    \item Note that $M_T(\mu)=O(\mu^{-1/2})$ and $V(\mu)=\Omega(\mu^{1/2})$, so it follows that
    \begin{align*}
        \lim_{\mu\to\infty}M_T(f)\cdot\left(2+V^2(f)+V(f)\cdot\sqrt{4+V^2(f)}\right)=\infty.
    \end{align*}
    If $M_T(f)\cdot\left(2+V^2(f)+V(f)\cdot\sqrt{4+V^2(f)}\right)< 2$, then by Lemma A.13 and the above asymptotic analysis, there exists a unique $\mu^\star\in(f,\infty)$ such that $\zeta_U(\mu^\star)=\bar{\zeta}_U(\mu^\star)$. 
    \begin{itemize}
        \item If $\mu<\mu^\star$, then $\zeta_U(\mu)$ $>\bar{\zeta}_U(\mu)$, so $V(\mu)<V_0$ under $\zeta_U=\zeta_U(\mu)$: for small trade sizes and thus small shock sizes, we have crowding out.
        \item If $\mu\geq\mu^\star$, then $\zeta_U(\mu)$ $\leq\bar{\zeta}_U(\mu)$, so $V(\mu)\geq V_0$ under $\zeta_U=\zeta_U(\mu)$: for large trade sizes and thus large shock sizes, we have complementing.
    \end{itemize}  
    This corresponds to the second case of the theorem. 
\end{itemize}
We now focus on the case when $\pi=1$ and derive an explicit threshold.

\begin{proposition}
Let $\pi=1$. Suppose that the trader and JIT LP's best-response functions have a unique maximum on $\R_+$. Let $(q_R^\star,q_S^\star)$ be the trader's strategy and $\tilde{d}_{J}^{BR}$ be the JIT LP's best-response function. Then the JIT LP crowds out the passive LPs at arrival probability 1 if and only if for $\omega\in\{\textsf{US},\textsf{UB}\}$, we have
\begin{align*}
    \frac{\d \tilde{d}_{J}^{BR}}{\d(q_R,q_S)}(\tilde{d}_{P},(q_R^\star,q_S^\star)(\tilde{d}_{P};\omega);\omega)<0.
\end{align*}
\begin{proof}
We prove the $\omega=\textsf{US}$ case; the proof for the $\omega=\textsf{UB}$ case is symmetric. The trader's problem is
\begin{align*}
    \max_{q_R\in\R_+} \ (1-\pi)(\delta_S(q_R,\tilde{d}_{P})-\zeta_U^{-1}p(1+f)q_R)+\pi(\delta_S(q_R,\tilde{d}_{P}+\tilde{d}_{J}^{BR}(\tilde{d}_{P},q_R;\textsf{US}))-\zeta_U^{-1}p(1+f)q_R).
\end{align*}
Let $q_R^\star(0)$ and $q_R^\star(1)$ be the optimal solutions for the trader's problem when $\pi=0$ and $\pi=1$, respectively, noting that these exist by assumption and under the class of pricing functions considered. Then $q_R^\star(0)$ and $q_R^\star(1)$ satisfy the first-order conditions (suppressing the arguments of $\tilde{d}_{J}^{BR}$ for brevity):
\begin{gather*}
    \frac{\d\delta_S}{\d r}(q_R^\star(0),\tilde{d}_{P}) = \frac{p(1+f)}{\zeta_U}, \\[0.25\baselineskip]
    \frac{\d\delta_S}{\d r}(q_R^\star(1),\tilde{d}_{P}+\tilde{d}_{J}^{BR})+\frac{\d \tilde{d}_{J}^{BR}}{\d q_R}(\tilde{d}_{P},q_R^\star(1);\textsf{US})\cdot\frac{\d\delta_S}{\d q}(q_R^\star(1),\tilde{d}_{P}+\tilde{d}_{J}^{BR}) = \frac{p(1+f)}{\zeta_U}.
\end{gather*}
If $\frac{\d \tilde{d}_{J}^{BR}}{\d q_R}(\tilde{d}_{P},q_R^\star(1);\textsf{US})<0$, then
\begin{align*}
    \frac{\d\delta_S}{\d r}(q_R^\star(1),\tilde{d}_{P}+\tilde{d}_{J}^{BR})>\frac{p(1+f)}{\zeta_U}
\end{align*}
since $\d\delta_S/\d q>0$. By the 1-homogeneity of $\delta_S$,
\begin{align*}
    \frac{\d\delta_S}{\d r}\left(\frac{q_R^\star(1)}{\tilde{d}_{P}+\tilde{d}_{J}^{BR}},1\right)>\frac{\d\delta_S}{\d r}\left(\frac{q_R^\star(0)}{\tilde{d}_{P}},1\right).
\end{align*}
Since $\delta_S$ is concave in $r$,
\begin{align*}
    \frac{q_R^\star(1)}{\tilde{d}_{P}+\tilde{d}_{J}^{BR}} < \frac{q_R^\star(0)}{\tilde{d}_{P}} &\implies q_R^\star(0) > \frac{\tilde{d}_{P}}{\tilde{d}_{P}+\tilde{d}_{J}^{BR}}\cdot q_R^\star(1) \\[0.25\baselineskip]
    &\implies \delta_S(q_R^\star(0),\tilde{d}_{P}) > \frac{\tilde{d}_{P}}{\tilde{d}_{P}+\tilde{d}_{J}^{BR}}\cdot\delta_S(q_R^\star(1),\tilde{d}_{P}+\tilde{d}_{J}^{BR}).
\end{align*}
Note that
\begin{gather*}
    \mathcal{R}(0)\cdot p\tilde{d}_{P} = (p\cdot q_R^\star(0)+\delta_S(q_R^\star(0),\tilde{d}_{P}))f \\[0.25\baselineskip]
    \mathcal{R}(1) \cdot p\tilde{d}_{P} = \frac{\tilde{d}_{P}}{\tilde{d}_{P}+\tilde{d}_{J}^{BR}}\left(p\cdot q_R^\star(1)+\delta_S(q_R^\star(1),\tilde{d}_{P}+\tilde{d}_{J}^{BR})\right)f.
\end{gather*}
The result follows.
\end{proof}
\end{proposition}
We complete the proof for the $\omega=\textsf{US}$ case; the proof for the $\omega=\textsf{UB}$ case is symmetric. By Lemmas A.5 and A.6, it follows that $\d \tilde{d}_{J}^{BR}/\d q_R<0$ if and only if $\d\nu/\d\mu<0$ where $\nu(\mu)$ is given by
\begin{align*}
    \nu(\mu) = \frac{f(1+\mu)+\mu\sqrt{(1+f)(1+\mu)}}{\mu-f}
\end{align*}
and $\mu$ uniquely satisfies $M_T(\mu)=(1+f)/\zeta_U$. Note that
\begin{align*}
    \frac{\d\nu}{\d\mu} = \frac{(1+f)(\mu^2-2\mu f-2f(1+\sqrt{(1+f)(1+\mu)}))}{2(\mu-f)^2\sqrt{(1+f)(1+\mu)}}
\end{align*}
so the sign of $\d\nu/\d\mu$ depends on $\mu^2-2\mu f-2f(1+\sqrt{(1+f)(1+\mu)})$. Solving the inequality over $\mu>0$ yields a threshold $\bar{\mu}$. The value of $\zeta_U$ that yields a an equilibrium trade size of $\bar{\mu}$ when $\pi=1$ is $\sqrt{\zeta_U}=\sqrt{f}+\sqrt{1+f}$. The result follows by Proposition A.14.

\subsection{Two-Tiered Fee Structure: Explicit Forms of Utilities}

The JIT LP's utility is now
\begin{align*}
    u_J(\tilde{d}_J;\sigma_{-J},\omega) = \begin{dcases}
        \frac{\tilde{d}_J}{\tilde{d}_P+\tilde{d}_J}\left(\frac{p(1+\lambda f)}{\zeta}\cdot q_R(\tilde{d}_P;\omega)-\delta_S(q_R(\tilde{d};\omega),\tilde{d}_P+\tilde{d}_J)\right) & \omega=\mathsf{IS} \\[0.25\baselineskip]
        \frac{\tilde{d}_J}{\tilde{d}_P+\tilde{d}_J}((1+\lambda f)\cdot q_S(\tilde{d}_P;\omega)-\zeta p\cdot\delta_R(q_S(\tilde{d}_P;\omega),\tilde{d}_P+\tilde{d}_J)) & \omega=\mathsf{IB} \\[0.25\baselineskip]
        \frac{\tilde{d}_J}{\tilde{d}_P+\tilde{d}_J}(p(1+\lambda f)\cdot q_R(\tilde{d}_P;\omega)-\delta_S(q_R(\tilde{d};\omega),\tilde{d}_P+\tilde{d}_J)) & \omega=\mathsf{US} \\[0.25\baselineskip]
        \frac{\tilde{d}_J}{\tilde{d}_P+\tilde{d}_J}((1+\lambda f)\cdot q_S(\tilde{d}_P;\omega)-p\cdot\delta_R(q_S(\tilde{d}_P;\omega),\tilde{d}_P+\tilde{d}_J)) & \omega=\mathsf{UB}
    \end{dcases}
\end{align*}
The passive LPs' total utility is now (suppressing arguments for the other agents' strategies)
\begin{gather*}
    u_P(\tilde{d}_P;\sigma_{-P},\omega) \propto \begin{dcases}
        \frac{\pi (\tilde{d}_P+(1-\lambda)\tilde{d}_J)}{\tilde{d}_P+\tilde{d}_J}\left(\frac{p(1+f)}{\zeta}q_R-\delta_S(q_R,\tilde{d}_P+\tilde{d}_J)\right) \\
        \hspace{5cm}+(1-\pi )\left(\frac{p(1+f)}{\zeta}q_R-\delta_S(q_R,\tilde{d}_P)\right) & \omega=\mathsf{IS} \\[0.25\baselineskip]
        \frac{\pi (\tilde{d}_P+(1-\lambda)\tilde{d}_J)}{\tilde{d}_P+\tilde{d}_J}((1+f)q_S-\zeta p\cdot\delta_R(q_S,\tilde{d}_P+\tilde{d}_J)) \\
        \hspace{5cm}+(1-\pi )((1+f)q_S-\zeta p\cdot\delta_R(q_S,\tilde{d}_P)) & \omega=\mathsf{IB} \\[0.25\baselineskip]
        \left[\frac{\pi (\tilde{d}_P+(1-\lambda)\tilde{d}_J)}{\tilde{d}_P+\tilde{d}_J}(pq_R+\delta_S(q_R,\tilde{d}_P+\tilde{d}_J))+(1-\pi )(pq_R+\delta_S(q_R,\tilde{d}_P))\right]f & \omega=\mathsf{US} \\[0.25\baselineskip]
        \left[\frac{\pi (\tilde{d}_P+(1-\lambda)\tilde{d}_J)}{\tilde{d}_P+\tilde{d}_J}(q_S+p\cdot\delta_R(q_S,\tilde{d}_P+\tilde{d}_J))+(1-\pi )(q_S+p\cdot \delta_R(q_S,\tilde{d}_P))\right]f & \omega=\mathsf{UB}
    \end{dcases}
\end{gather*}
The informed and uninformed traders' utilities have the same form as before, but we keep in mind that they are now anticipating the JIT LP's best-response when the JIT LP only keeps $\lambda$ of its pro-rata share of fees, so they anticipate a dampened response.

\subsection{Proof of Proposition \ref{thm:twotier-eq}}

It suffices to show the following: let $\tilde{d}_{P}\in[0,\tilde{e}_{P}]$ and  $\underline{\zeta}(\pi )$ be given by
\begin{align*}
    \underline{\zeta}(\pi ) = \frac{2(1+f)^3}{2+\pi[(2+f)\sqrt{(1+f)(1+\lambda f)}-2]}.
\end{align*}
Suppose that $\zeta_U>\underline{\zeta}(\pi )$. Then there exists a non-trivial Nash equilibrium in the subgame between the traders and JIT LP. Define the following:
\begin{gather*}
    \tilde{\mu}_I = \zeta^{1/2}(1+f)^{-1/2}-1, \\[0.25\baselineskip]
    \tilde{\mu}(\lambda,\pi ) = \arg\min_{\mu\in\R_+}\,\left|(1-\pi)\cdot\frac{1}{(1+\mu)^2}+\pi\cdot\frac{(2+\mu)\sqrt{(1+\lambda f)(1+\mu)}}{2(1+\mu)^2}-\frac{1+f}{\zeta_U}\right|\!, \\[0.25\baselineskip]
    \tilde{\nu}(\lambda,\pi ) = \frac{\lambda f(1+\tilde{\mu}(\pi))+\tilde{\mu}(\pi)\sqrt{(1+\lambda f)(1+\tilde{\mu}(\pi))}}{\tilde{\mu}(\pi)-\lambda f}.
\end{gather*}
The equilibrium outcome is
\begin{gather*}
    (q_R,q_S)^\star(\tilde{d}_P;\omega) = \begin{dcases}
        (\tilde{\mu}_I\tilde{d}_P,0) & \omega=\mathsf{IS} \\
        (0,\tilde{\mu}_Ip\tilde{d}_P) & \omega=\mathsf{IB} \\
        (\tilde{\mu}(\lambda,\pi)\cdot \tilde{d}_P,0) & \omega=\mathsf{US} \\
        (0,\tilde{\mu}(\lambda,\pi)\cdot p\tilde{d}_P) & \omega=\mathsf{UB}
    \end{dcases}, \\[0.25\baselineskip]
    \tilde{d}_J^\star(\tilde{d}_P,(q_R,q_S)^\star(\tilde{d}_P;\omega);\omega) = \begin{dcases}
        0 & \omega\in\{\mathsf{IS},\mathsf{IB}\} \\
        \tilde{\nu}(\lambda,\pi)\cdot\tilde{d}_P & \omega\in\{\mathsf{US},\mathsf{UB}\}
    \end{dcases}.
\end{gather*}
The multiples $\mu_I$, $\mu(\lambda,\pi)$, and $\nu(\lambda,\pi)$ are given by scaling $\tilde{\mu}$, $\tilde{\mu}(\lambda,\pi)$, and $\tilde{\nu}(\lambda,\pi)$ by $p^{1/2}$, respectively. It remains to find the passive LPs' best-responses when the remaining agents play the unique non-trivial Nash equilibrium of the trader--JIT LP subgame given $\tilde{d}_{P}$ provided by the passive LPs, denoted $\sigma_{-P}^\star(\tilde{d}_{P})$. Substituting the explicit forms of $\delta_S$, $\delta_R$, and $\sigma_{-P}^\star(\tilde{d}_{P})$ into the expressions for $u_{P}(\tilde{d}_{P};\sigma_{-P}^\star(\tilde{d}_{P}))$ yields
\begin{align*}
    u_{P}(\tilde{d}_{P};\sigma_{-P}^\star) = \mathcal{U}(\lambda,\pi) = (\alpha\mathcal{C}+(1-\alpha)\cdot\mathcal{R}(\lambda,\pi))p\tilde{d}_{P}
\end{align*}
where $\mathcal{C}$ and $\mathcal{R}(\lambda,\pi)$ are given by
\begin{gather*}
    \mathcal{C} = -\left[\psi\left(1-\frac{1+f}{\zeta}\right)^2+(1-\psi\left(\sqrt{\zeta}-\sqrt{1+f}\right)^2\right], \\[0.25\baselineskip]
    \mathcal{R}(\lambda,\pi) = \left[(1-\pi)\left(\tilde{\mu}(\lambda,\pi)+\frac{\tilde{\mu}(\lambda,\pi)}{1+\tilde{\mu}(\lambda,\pi)}\right)+\pi\cdot\frac{1+(1-\lambda)\cdot\tilde{\nu}(\lambda,\pi)}{1+\tilde{\nu}(\lambda,\pi)}\left(\tilde{\mu}(\lambda,\pi)+\frac{(1+\tilde{\nu}(\lambda,\pi))\cdot\tilde{\mu}(\lambda,\pi)}{1+\tilde{\nu}(\lambda,\pi)+\tilde{\mu}(\lambda,\pi)}\right)\right]f.
\end{gather*}
If $\mathcal{U}(\lambda,\pi)<0$, then each passive LP's best-response is to deposit nothing. Otherwise, if $\mathcal{U}(\lambda,\pi)>0$, then each passive LP's best-response is to deposit their endowment.

We show both of these results with an argument similar to the proofs of Propositions \ref{thm:subgame-eq} and \ref{thm:passive-br}, but using the utilities delineated in the previous subsection.

\subsection{Proof of Proposition \ref{thm:dampen}}

We show the proof for the $\omega=\mathsf{US}$ case; the proof for the $\omega=\mathsf{UB}$ case is symmetric. We will first require the following lemma: 

\begin{lemma}
Fix $\pi\in(0,1]$. Let $\mu (\lambda)\equiv\tilde{\mu}(\lambda,\pi)$ be the uninformed trader's propensity to swap and $\nu(\lambda)\equiv\tilde{\nu}(\lambda,\pi)$ be the JIT LP's propensity to deposit under a transfer rate of $\lambda$. If $\d\mu /\d\lambda>0$ and 
\begin{align*}
    p\left(\frac{\mu (\lambda_0)}{1+\nu(\lambda_0)}\right)^2f > \frac{\d\mu}{\d\lambda}(\lambda_0)\cdot\frac{\d\delta_S}{\d q}\left(\frac{\mu (\lambda_0)}{1+\nu(\lambda_0)},1\right)
\end{align*}
for all $\lambda_0\in[0,1]$, then $\d\nu/\d\lambda>0$ for all $\lambda\in[0,1]$.
\end{lemma}
\begin{proof}
By assumption,
\begin{align*}
    p\left(\frac{\mu (\lambda_0)}{1+\nu(\lambda_0)}\right)^2f > \left(\lim_{\lambda_1\to\lambda_0^+}\frac{\mu (\lambda_1)-\mu (\lambda_0)}{\lambda_1-\lambda_0}\right)\cdot\frac{\d\delta_S}{\d q}\left(\frac{\mu (\lambda_0)}{1+\nu(\lambda_0)},1\right).
\end{align*}
By properties of limits and rearranging,
\begin{align*}
    \lim_{\lambda_1\to\lambda_0^+}\frac{1}{1+\nu(\lambda_0)}\cdot p(\lambda_1-\lambda_0)f\cdot\frac{\mu (\lambda_0)}{1+\nu(\lambda_0)} &> \lim_{\lambda_1\to\lambda_0^+}\frac{\mu (\lambda_1)-\mu (\lambda_0)}{\mu (\lambda_0)}\cdot\frac{\d\delta_S}{\d q}\left(\frac{\mu (\lambda_0)}{1+\nu(\lambda_0)},1\right) \\[0.25\baselineskip]
    \lim_{\lambda_1\to\lambda_0^+}\frac{\mu (\lambda_0)}{\mu (\lambda_1)}\cdot\frac{1}{1+\nu(\lambda_0)}\cdot p(\lambda_1-\lambda_0)f\cdot\frac{\mu (\lambda_0)}{1+\nu(\lambda_0)} &> \lim_{\lambda_1\to\lambda_0^+}\left(1-\frac{\mu (\lambda_0)}{\mu (\lambda_1)}\right)\cdot\frac{\d\delta_S}{\d q}\left(\frac{\mu (\lambda_0)}{1+\nu(\lambda_0)},1\right).
\end{align*}
Multiplying by $\mu(\lambda_0)/\mu(\lambda_1)$ and substituting the expression for $\d\delta_S/\d q$ yields that the above inequality is equivalent to
\begin{align*}
    &\lim_{\lambda_1\to\lambda_0^+}\left(\frac{\mu (\lambda_0)}{\mu (\lambda_1)}\cdot\left(1-\frac{1}{1+\nu(\lambda_0)}\right)\cdot\frac{\d\delta_S}{\d q}\left(\frac{\mu (\lambda_0)}{1+\nu(\lambda_0)},1\right)+\frac{\mu (\lambda_0)}{\mu (\lambda_1)}\cdot\frac{1}{1+\nu(\lambda_0)}\cdot p(\lambda_1-\lambda_0)f\cdot\frac{\mu (\lambda_0)}{1+\nu(\lambda_0)}\right) \\[0.25\baselineskip]
    &\hspace{1cm}>\lim_{\lambda_1\to\lambda_0^+}\left(1-\frac{\mu (\lambda_0)}{\mu (\lambda_1)}\cdot\frac{1}{1+\nu(\lambda_0)}\right)\cdot\frac{\d\delta_S}{\d q}\left(\frac{\mu (\lambda_0)}{1+\nu(\lambda_0)},1\right). \tag{$\star$}
\end{align*}
By the JIT LP's first-order condition, the LHS of $(\star)$ satisfies
\begin{align*}
    &\lim_{\lambda_1\to\lambda_0^+}\left(\frac{\mu (\lambda_0)}{\mu (\lambda_1)}\cdot\left(1-\frac{1}{1+\nu(\lambda_0)}\right)\cdot\frac{\d\delta_S}{\d q}\left(\frac{\mu (\lambda_0)}{1+\nu(\lambda_0)},1\right)+\frac{\mu (\lambda_0)}{\mu (\lambda_1)}\cdot\frac{1}{1+\nu(\lambda_0)}\cdot p(\lambda_1-\lambda_0)f\cdot\frac{\mu (\lambda_0)}{1+\nu(\lambda_0)}\right) \\[0.25\baselineskip]
    &\hspace{1cm}=\lim_{\lambda_1\to\lambda_0^+}\frac{\mu (\lambda_0)}{\mu (\lambda_1)}\cdot\frac{1}{1+\nu(\lambda_0)}\left(p(1+\lambda_1f)\cdot\frac{\mu (\lambda_0)}{1+\nu(\lambda_0)}-\delta_S\left(\frac{\mu (\lambda_0)}{1+\nu(\lambda_0)}\right)\right).
\end{align*}
Note that this is the LHS of the JIT LP's first-order condition if the transfer rate increased from $\lambda_0$ to $\lambda_1$ and the trader chooses $\mu (\lambda_1)$ while the JIT LP chooses a propensity to deposit $\Bar{\nu}$ to keep the trade-to-liquidity ratio constant, i.e.\!
\newpage \noindent \begin{align*}
    \Bar{\nu} = \frac{\mu (\lambda_1)}{\mu (\lambda_0)}\cdot(1+\nu(\lambda_0))-1>\nu(\lambda_0).
\end{align*}
Note that the RHS of $(\star)$ is the RHS of the JIT LP's first-order condition if the transfer rate increased from $\lambda_0$ to $\lambda_1$ and the trader chooses $\mu (\lambda_1)$ while the JIT LP chooses $\Bar{\nu}$, so the LHS exceeds the RHS in the JIT LP's first-order condition at $\Bar{\nu}$. Since $u_{J}(\nu;\mu (\lambda_1),\lambda_1)$ is quasiconcave in $\nu$ and is increasing at $\nu=0$, it follows that $u_{J}$ must be increasing at $\nu=\Bar{\nu}$. Thus $\nu(\lambda_1)>\Bar{\nu}$, so the statement follows.
\end{proof}

The JIT LP's best-response function is
\begin{align*}
    \nu^{BR}(\mu(\lambda),\lambda) = \frac{\lambda f(1+\mu(\lambda) )+\mu(\lambda)\cdot\sqrt{(1+\lambda f)(1+\mu(\lambda))}}{\mu(\lambda)-\lambda f},
\end{align*}
and the uninformed trader's first-order condition is
\begin{align*}
    \frac{(1-\pi )}{(1+\mu)^2} + \frac{\pi (2+\mu)\sqrt{(1+\lambda f)(1+\mu)}}{2(1+\mu)^2} = \frac{1+f}{\zeta}.
\end{align*}
Implicitly differentiating $\mu$ w.r.t.\! $\lambda$ subject to the trader's first-order condition yields
\begin{align*}
    \frac{\d\mu}{\d\lambda} = \frac{\pi (1+\mu )^2(2+\mu )}{\pi (1+\lambda f)(1+\mu )(4+\mu)+8(1-\pi )\sqrt{(1+\lambda f)(1+\mu)}} \leq \frac{(1+\mu)(2+\mu )}{(1+\lambda f)(4+\mu)}\cdot f.
\end{align*}
since the above expression for $\d\mu/\d\lambda$ is increasing in $\pi$. Note that
\begin{gather*}
    \left(1+\frac{\mu (\lambda)}{1+\nu(\lambda)}\right)^2 =\frac{1+\mu (\lambda)}{1+\lambda f}, \\[0.25\baselineskip]
    \frac{\d\delta_S}{\d q}\left(\frac{\mu (\lambda)}{1+\nu(\lambda)},1\right) = p\cdot\frac{\mu (\lambda)^2}{(1+\mu (\lambda)+\nu(\lambda))^2},
\end{gather*}
so it follows that
\begin{align*}
    \left(\frac{\d\delta_S}{\d q}\left(\frac{\mu (\lambda_0)}{1+\nu(\lambda_0)},1\right)\right)^{-1}\left(p\left(\frac{\mu (\lambda_0)}{1+\nu(\lambda_0)}\right)^2f\right) &= \frac{1+\mu (\lambda)}{1+\lambda f}\cdot f \\[0.25\baselineskip]
    &> \frac{(1+\mu (\lambda))(2+\mu (\lambda))}{(1+\lambda f)(4+\mu (\lambda))}\cdot f = \frac{\d\mu }{\d\lambda}.
\end{align*}
Thus the assumptions for Lemma A.15 are satisfied, and the result follows.

\subsection{Proof of Theorem \ref{thm:preventfreeze}}

We require the following lemma:

\begin{lemma}
For $\mu,f\geq0$ and $\pi\in[0,1]$, if
\begin{align*}
    (1-\pi)\left(1+\mu+\frac{1}{1+\mu}\right)f+\pi(1+\mu+\sqrt{1+\mu})f +\pi \leq \pi(1+\mu)
\end{align*}
then the passive LPs' per-unit utility is decreasing in $\lambda$ for all $\lambda\in[0,1]$.
\end{lemma}
\begin{proof}
Note that
\begin{align*}
    (1-\pi)\left(1+\mu+\frac{1}{1+\mu}\right)f+\pi(1+\mu+\sqrt{1+\mu})f +\pi &\leq \pi(1+\mu) \\[0.25\baselineskip]
    (1-\pi)\left(1+\mu+\frac{1}{1+\mu}\right)f+\pi\left[\sqrt{1+\mu}+\left(1+\mu+\sqrt{1+\mu}\right)+1\right] &\leq \pi(1+\mu+\sqrt{1+\mu}),
\end{align*}
so for all $\lambda\in[0,1]$,
\begin{align*}
    &(1-\pi)\left(\frac{1+\mu}{1+\lambda f}+\frac{1}{(1+\mu)(1+\lambda f)}\right)f+\pi\left[\frac{\sqrt{1+\mu}}{\sqrt{1+\lambda f}}+(1-\lambda)\left(\frac{1+\mu}{1+\lambda f}+\frac{\sqrt{1+\mu}}{\sqrt{1+\lambda f}}\right)f+1\right] \\[0.25\baselineskip]
    &\hspace{10.55cm} \leq \pi\left[1+\mu+\sqrt{(1+\lambda f)(1+\mu)}\right] \\[0.25\baselineskip]
    &(1-\pi)\left(\frac{1+\mu}{1+\lambda f}+\frac{1}{(1+\mu)(1+\lambda f)}\right)f+\pi\left[\frac{\sqrt{1+\mu}}{\sqrt{1+\lambda f}}+(1-\lambda)\left(\frac{1+\mu}{1+\lambda f}+\frac{\sqrt{1+\mu}}{\sqrt{1+\lambda f}}\right)f+1\right. \\
    &\hspace{7cm} \left.+\ \frac{\lambda}{2}\cdot\frac{1}{1+\lambda f}\cdot\frac{\sqrt{1+\mu}}{\sqrt{1+\lambda f}}\cdot f+\frac{1}{2}\cdot\frac{1}{\sqrt{(1+\lambda f)(1+\mu)}}\cdot f\right] \\[0.25\baselineskip]
    &\hspace{1cm} \leq \pi\left[1+\mu+\frac{\lambda}{2}\cdot\frac{1}{1+\lambda f}\cdot\frac{\sqrt{1+\mu}}{\sqrt{1+\lambda f}}\cdot f+\sqrt{(1+\lambda f)(1+\mu)}+\frac{1}{2}\cdot\frac{1}{\sqrt{(1+\lambda f)(1+\mu)}}\cdot f\right].
\end{align*}
since the LHS is decreasing in $\lambda$ and the RHS is increasing in $\lambda$. Recall from the proof of Proposition $\ref{thm:dampen}$ that $\mu'(\lambda)$ satisfies
\begin{align*}
    \mu'(\lambda) = \frac{(1+\mu)(2+\mu)}{(1+\lambda f)(4+\mu)}\cdot f < \frac{1+\mu}{1+\lambda f}\cdot f,
\end{align*}
so it follows that
\newpage \noindent \begin{align*}
    &(1-\pi)\left(\mu'(\lambda)+\frac{\mu'(\lambda)}{(1+\mu)^2}\right)+\pi\left[\frac{\sqrt{1+\mu}}{\sqrt{1+\lambda f}}+(1-\lambda)\cdot\mu'(\lambda)+\frac{\lambda}{2}\cdot\frac{1}{1+\mu}\cdot\frac{\sqrt{1+\mu}}{\sqrt{1+\lambda f}}\cdot\mu'(\lambda)\right. \\
    &\hspace{3cm} \left.+\ 1+\frac{1-\lambda}{2}\cdot\frac{\sqrt{1+\mu}}{\sqrt{1+\lambda f}}+\frac{1-\lambda}{2}\cdot\frac{\sqrt{1+\lambda f}}{\sqrt{1+\mu}}\cdot\mu'(\lambda)+\frac{1}{2}\cdot\frac{\sqrt{1+\lambda f}}{\sqrt{1+\mu}}\cdot\frac{1}{1+\mu}\cdot\mu'(\lambda)\right] \\[0.25\baselineskip]
    &\hspace{1cm} < \pi\left[1+\mu+\frac{\lambda}{2}\cdot\frac{1}{1+\lambda f}\cdot\frac{\sqrt{1+\mu}}{\sqrt{1+\lambda f}}\cdot f+\sqrt{(1+\lambda f)(1+\mu)}+\frac{1}{2}\cdot\frac{1}{\sqrt{(1+\lambda f)(1+\mu)}}\cdot f\right].
\end{align*}
Moving everything to the left-hand side and substituting $\mu\equiv\tilde{\mu}(\lambda,\pi)$ yields that
\begin{align*}
    \frac{\d}{\d\lambda}\left[(1-\pi)\left(\tilde{\mu}(\lambda,\pi)+\frac{\tilde{\mu}(\lambda,\pi)}{1+\tilde{\mu}(\lambda,\pi)}\right)+\pi\cdot\frac{1+(1-\lambda)\cdot\tilde{\nu}(\lambda,\pi)}{1+\tilde{\nu}(\lambda,\pi)}\left(\tilde{\mu}(\lambda,\pi)+\frac{(1+\tilde{\nu}(\lambda,\pi))\cdot\tilde{\mu}(\lambda,\pi)}{1+\tilde{\nu}(\lambda,\pi)+\tilde{\mu}(\lambda,\pi)}\right)\right]
\end{align*}
is negative. The statement follows.
\end{proof}
Solving the inequality
\begin{align*}
    (1-\pi)\left(1+\mu+\frac{1}{1+\mu}\right)f+\pi(1+\mu+\sqrt{1+\mu})f +\pi \leq \pi(1+\mu)
\end{align*}
for an upper bound on $f$ yields
\begin{align*}
    f \leq \frac{\pi\mu(1+\mu)}{2+\mu(2+\mu)+\pi\left[(1+\mu)^{3/2}-1\right]} = \bar{f}(\mu,\pi).
\end{align*}
Note that $\Bar{f}(\mu,\pi)$ is increasing in $\mu$ and $\lim_{\mu\to\infty}\bar{f}(\mu,\pi)=\pi$, so we have a bijection between $\mu\in(f,\infty)$ and $\bar{f}(\mu,\pi)\in(\bar{f}(f,\pi),\pi)$. Since $\tilde{\mu}(\lambda,\pi)$ is increasing in $\pi$, if $f\leq\Bar{f}(\tilde{\mu}(1,\pi),\pi)$, then $f\leq\bar{f}(\tilde{\mu}(\lambda,\pi),\pi)$ for any $\lambda\in[0,1]$, so it suffices to upper bound $\tilde{\mu}(1,\pi)$ to imply that $\d\mathcal{R}(\lambda,\pi)/\d\lambda<0$. Let $\zeta_U(\mu)$ be the shock size such that the uninformed trader's equilibrium trade size is $\mu$ under $\lambda=1$. Since \ref{thm:threshold} that $\zeta_U(\mu)=(1+f)/M_T(\mu)$. Since $M_T(\mu)$ is decreasing in $\mu$, it follows that $\zeta_U(\mu)$ is increasing in $\mu$, so we have a bijection between $\mu\in(f,\infty)$ and $\zeta_U(\mu)\in(\underline{\zeta}(f,\pi),\infty)$. Composing increasing bijections implies that we have a increasing bijection between $f\in(\bar{f}(f,\pi),\pi)$ and $\zeta_U\in(\underline{\zeta}(f,\pi),\infty)$. The result follows.

\subsection{Proof of Theorem \ref{thm:maxwelfare}}

Observe that welfare is zero in the case of informed buys and sells. In the case of an uninformed sell, the utilities of each agent are:
\newpage \noindent \begin{gather*}
    u_T = \pi \cdot\delta_S(q_R,\tilde{d}_{P}+\tilde{d}_{J})+(1-\pi )\cdot\delta_S(q_R,\tilde{d}_{P})-\zeta_U^{-1}p(1+f)q_R \\[0.25\baselineskip]
    u_{J} = \frac{\pi \tilde{d}_{J}}{\tilde{d}_{P}+\tilde{d}_{J}}(p(1+\lambda f)q_R-\delta_S(q_R,\tilde{d}_{P}+\tilde{d}_{J})) \\[0.25\baselineskip]
    u_A = \frac{\pi \tilde{d}_{P}}{\tilde{d}_{P}+\tilde{d}_{J}}(pq_R-(1+f)\cdot \delta_S(q_R,\tilde{d}_{P}+\tilde{d}_{J}))+(1-\pi )(pq_R-(1+f)\cdot\delta_S(q_R,\tilde{d}_{P})) \\[0.25\baselineskip]
    u_{P} = \pi \left(\frac{\tilde{d}_{P}+(1-\lambda)\tilde{d}_{J}}{\tilde{d}_{P}+\tilde{d}_{J}}pq_R+\frac{\tilde{d}_{P}}{\tilde{d}_{P}+\tilde{d}_{J}}\cdot \delta_S(q_R,\tilde{d}_{P}+\tilde{d}_{J})\right)f+(1-\pi )(pq_R+\delta_S(q_R,\tilde{d}_{P}))f,
\end{gather*}
where we omit the dependence of $q_R$ and $\tilde{d}_{J}$ on $\lambda$ for brevity, and the subscripts $T$, $J$, $A$, and $P$ represent the (uninformed) trader, JIT LP, arbitrageur, and passive LPs, respectively. In the case of an uninformed buy, the utilities of each agent are (suppressing arguments):
\begin{gather*}
    u_T = \zeta_U p \left( \pi \cdot\delta_R(q_S,\tilde{d}_{P}+\tilde{d}_{J})+(1-\pi )\cdot\delta_R(q_S,\tilde{d}_{P}) \right)-(1+f)q_S \\[0.25\baselineskip]
    u_{J} = \frac{\pi \tilde{d}_{J}}{\tilde{d}_{P}+\tilde{d}_{J}}((1+\lambda f)q_S-p\cdot\delta_R(q_S,\tilde{d}_{P}+\tilde{d}_{J})) \\[0.25\baselineskip]
    u_A = \frac{\pi \tilde{d}_{P}}{\tilde{d}_{P}+\tilde{d}_{J}}(q_S-p(1+f)\cdot\delta_R(q_S,\tilde{d}_{P}+\tilde{d}_{J}))+(1-\pi )(q_S-p(1+f)\cdot\delta_R(q_S,\tilde{d}_{P})) \\[0.25\baselineskip]
    u_{P} = \pi \left(\frac{\tilde{d}_{P}+(1-\lambda)\tilde{d}_{J}}{\tilde{d}_{P}+\tilde{d}_{J}}q_S+\frac{\tilde{d}_{P}}{\tilde{d}_{P}+\tilde{d}_{J}}\cdot p\cdot\delta_R(q_S,\tilde{d}_{P}+\tilde{d}_{J})\right)f+(1-\pi )(q_S+p\cdot\delta_R(q_S,\tilde{d}_{P}))f,
\end{gather*}
omitting the dependence of $q_S$ on $\lambda$ for brevity. In each case, welfare is
\begin{gather*}
    W(\mathsf{US}) = (1-\zeta_U^{-1})p(1+f)q_R \\
    W(\mathsf{UB}) = (\zeta_U-1)p\cdot\delta_R(q_S,\tilde{d}_{P}+\tilde{d}_{J}) 
\end{gather*}
Weighting each case by its probability yields an expected welfare of
\begin{align*}
    W =\psi_U(1-\zeta_U^{-1})p(1+f)\cdot q_R+(1-\psi_U)(\zeta_U-1)p\left[\pi \cdot\delta_R(q_S,\tilde{d}_{P}+\tilde{d}_{J})+(1-\pi )\cdot\delta_R(q_S,\tilde{d}_{P})\right].
\end{align*}
Since $q_R^\star(\lambda$, $q_S^\star(\lambda)$ and $\tilde{d}_{J}^\star(\lambda)$ are increasing in $\lambda$, and $\delta_R$ is 1-homogeneous and increasing in both arguments, it follows that $\delta_R(q_S^\star,\tilde{d}_{P}+\tilde{d}_{J}^\star)$ is also increasing in $\lambda$. 

\subsection{Cournot-Competing JIT LPs: Explicit Forms of Utilities}

The JIT LP's utility is now
\begin{gather*}
    u_J^{(j)}(\tilde{d}_J^{(j)};\sigma_{-J},(\omega,\mathsf{NA}_{-j})) = \begin{dcases}
        \frac{\tilde{d}_J^{(j)}}{\tilde{d}_P+\tilde{d}_J^{(j)}}\left(\frac{p(1+f)}{\zeta}\cdot q_R(\tilde{d}_P;\omega)-\delta_S(q_R(\tilde{d};\omega),\tilde{d}_P+\tilde{d}_J^{(j)})\right) & \omega=\mathsf{IS} \\[0.25\baselineskip]
        \frac{\tilde{d}_J^{(j)}}{\tilde{d}_P+\tilde{d}_J^{(j)}}((1+f)\cdot q_S(\tilde{d}_P;\omega)-\zeta p\cdot\delta_R(q_S(\tilde{d}_P;\omega),\tilde{d}_P+\tilde{d}_J^{(j)})) & \omega=\mathsf{IB} \\[0.25\baselineskip]
        \frac{\tilde{d}_J^{(j)}}{\tilde{d}_P+\tilde{d}_J^{(j)}}(p(1+f)\cdot q_R(\tilde{d}_P;\omega)-\delta_S(q_R(\tilde{d};\omega),\tilde{d}_P+\tilde{d}_J^{(j)})) & \omega=\mathsf{US} \\[0.25\baselineskip]
        \frac{\tilde{d}_J^{(j)}}{\tilde{d}_P+\tilde{d}_J^{(j)}}((1+f)\cdot q_S(\tilde{d}_P;\omega)-p\cdot\delta_R(q_S(\tilde{d}_P;\omega),\tilde{d}_P+\tilde{d}_J^{(j)})) & \omega=\mathsf{UB}
    \end{dcases} \\[0.25\baselineskip]
    u_J^{(j)}(\tilde{d}_J^{(j)};\sigma_{-J},(\omega,\mathsf{A}_{-j})) = \begin{dcases}
        \frac{\tilde{d}_J^{(j)}}{\tilde{d}_P+\tilde{d}_J^{(j)}+\tilde{d}_J^{(-j)}}\left(\frac{p(1+f)}{\zeta}\cdot q_R(\tilde{d}_P;\omega)-\delta_S(q_R(\tilde{d};\omega),\tilde{d}_P+\tilde{d}_J^{(j)}+\tilde{d}_J^{(-j)})\right) & \omega=\mathsf{IS} \\[0.25\baselineskip]
        \frac{\tilde{d}_J^{(j)}}{\tilde{d}_P+\tilde{d}_J^{(j)}+\tilde{d}_J^{(-j)}}((1+f)\cdot q_S(\tilde{d}_P;\omega)-\zeta p\cdot\delta_R(q_S(\tilde{d}_P;\omega),\tilde{d}_P+\tilde{d}_J^{(j)}+\tilde{d}_J^{(-j)})) & \omega=\mathsf{IB} \\[0.25\baselineskip]
        \frac{\tilde{d}_J^{(j)}}{\tilde{d}_P+\tilde{d}_J^{(j)}+\tilde{d}_J^{(-j)}}(p(1+f)\cdot q_R(\tilde{d}_P;\omega)-\delta_S(q_R(\tilde{d};\omega),\tilde{d}_P+\tilde{d}_J^{(j)}+\tilde{d}_J^{(-j)})) & \omega=\mathsf{US} \\[0.25\baselineskip]
        \frac{\tilde{d}_J^{(j)}}{\tilde{d}_P+\tilde{d}_J^{(j)}+\tilde{d}_J^{(-j)}}((1+f)\cdot q_S(\tilde{d}_P;\omega)-p\cdot\delta_R(q_S(\tilde{d}_P;\omega),\tilde{d}_P+\tilde{d}_J^{(j)}+\tilde{d}_J^{(-j)})) & \omega=\mathsf{UB}.
    \end{dcases}
\end{gather*}
For $\omega\in\{\mathsf{IS},\mathsf{IB}\}$, the informed trader's utility is (suppressing all arguments for the JIT LP's strategy expect for $\Omega_j$)
\begin{align*}
    u_T((q_R,q_S);\sigma_{-T},\omega) = \begin{dcases}
        (1-\pi)^2\cdot\delta_S(q_R,\tilde{d}_P)+\pi(1-\pi)\cdot\delta_S(q_R,\tilde{d}_P+\tilde{d}_J^{(1)}(\mathsf{NA}_{2})) \\
        \hspace{1cm}+\,\pi(1-\pi)\cdot\delta_S(q_R,\tilde{d}_P+\tilde{d}_J^{(2)}(\mathsf{NA}_{1})) \\
        \hspace{1cm}+\,\pi^2\cdot\delta_S(q_R,\tilde{d}_P+\tilde{d}_J^{(1)}(\mathsf{A}_2)+\tilde{d}_J^{(2)}(\mathsf{A}_1))-p'(\omega)\cdot(1+f)q_R & q_R>0 \\
        p'(\omega)[(1-\pi)^2\cdot\delta_R(q_S,\tilde{d}_P)+\pi(1-\pi)\cdot\delta_S(q_S,\tilde{d}_P+\tilde{d}_J^{(1)}(\mathsf{NA}_2)) \\
        \hspace{1cm}+\,\pi(1-\pi)\cdot\delta_R(q_S,\tilde{d}_P+\tilde{d}_J^{(2)}(\mathsf{NA}_1)) \\
        \hspace{1cm}+\,\pi^2\cdot\delta_R(q_S,\tilde{d}_P+\tilde{d}_J^{(1)}(\mathsf{A}_2)+\tilde{d}_J^{(2)}(\mathsf{A}_1))]-(1+f)q_S & q_S>0 \\
        0 & \text{o.w.}
    \end{dcases}.
\end{align*}
For $\omega\in\{\mathsf{US},\mathsf{UB}\}$, the uninformed trader's utility is (suppressing all arguments for the JIT LP's strategy except for $\Omega_j$)
\begin{align*}
    u_T((q_R,q_S);\sigma_{-T},\omega) = \begin{dcases}
        (1-\pi)^2\cdot\delta_S(q_R,\tilde{d}_P)+\pi(1-\pi)\cdot\delta_S(q_R,\tilde{d}_P+\tilde{d}_J^{(1)}(\mathsf{NA}_{2})) \\
        \hspace{1cm}+\,\pi(1-\pi)\cdot\delta_S(q_R,\tilde{d}_P+\tilde{d}_J^{(2)}(\mathsf{NA}_{1})) \\
        \hspace{1cm}+\,\pi^2\cdot\delta_S(q_R,\tilde{d}_P+\tilde{d}_J^{(1)}(\mathsf{A}_2)+\tilde{d}_J^{(2)}(\mathsf{A}_1))-P(\omega)\cdot(1+f)q_R & q_R>0 \\
        P(\omega)[(1-\pi)^2\cdot\delta_R(q_S,\tilde{d}_P)+\pi(1-\pi)\cdot\delta_S(q_S,\tilde{d}_P+\tilde{d}_J^{(1)}(\mathsf{NA}_2)) \\
        \hspace{1cm}+\,\pi(1-\pi)\cdot\delta_R(q_S,\tilde{d}_P+\tilde{d}_J^{(2)}(\mathsf{NA}_1)) \\
        \hspace{1cm}+\,\pi^2\cdot\delta_R(q_S,\tilde{d}_P+\tilde{d}_J^{(1)}(\mathsf{A}_2)+\tilde{d}_J^{(2)}(\mathsf{A}_1))]-(1+f)q_S & q_S>0 \\
        0 & \text{o.w.}
    \end{dcases}.
\end{align*}
The passive LPs' conditional total utility (suppressing arguments for the other agents' strategies) is then
\begin{gather*}
    u_P(\tilde{d}_P;\sigma_{-P},\omega) = \begin{dcases}
        \frac{\pi^2 \tilde{d}_P}{\tilde{d}_P+\tilde{d}_J^{(1)}(\mathsf{A}_2)+\tilde{d}_J^{(2)}(\mathsf{A}_1)}\left(\frac{p(1+f)}{\zeta}q_R-\delta_S(q_R,\tilde{d}_P+\tilde{d}_J^{(1)}(\mathsf{A}_2)+\tilde{d}_J^{(2)}(\mathsf{A}_1))\right) & \\
        \hspace{1cm}+\,\pi(1-\pi)\cdot\frac{\tilde{d}_P}{\tilde{d}_P+\tilde{d}_J^{(1)}(\mathsf{NA}_2)}\left(\frac{p(1+f)}{\zeta}q_R-\delta_S(q_R,\tilde{d}_P+\tilde{d}_J^{(1)}(\mathsf{NA}_2))\right) \\
        \hspace{1cm}+\,\pi(1-\pi)\cdot\frac{\tilde{d}_P}{\tilde{d}_P+\tilde{d}_J^{(2)}(\mathsf{NA}_1)}\left(\frac{p(1+f)}{\zeta}q_R-\delta_S(q_R,\tilde{d}_P+\tilde{d}_J^{(2)}(\mathsf{NA}_1))\right) \\
        \hspace{1cm}+\,(1-\pi)^2\left(\frac{p(1+f)}{\zeta}q_R-\delta_S(q_R,\tilde{d}_P)\right) & \omega=\mathsf{IS} \\[0.25\baselineskip]
        \frac{\pi^2 \tilde{d}_P}{\tilde{d}_P+\tilde{d}_J^{(1)}(\mathsf{A}_2)+\tilde{d}_J^{(2)}(\mathsf{A}_1)}((1+f)q_S-\zeta p\cdot\delta_R(q_S,\tilde{d}_P+\tilde{d}_J^{(1)}(\mathsf{A}_2)+\tilde{d}_J^{(2)}(\mathsf{A}_1))) & \\
        \hspace{1cm}+\,\pi(1-\pi)\cdot\frac{\tilde{d}_P}{\tilde{d}_P+\tilde{d}_J^{(1)}(\mathsf{NA}_2)}((1+f)q_S-\zeta p\cdot\delta_R(q_S,\tilde{d}_P+\tilde{d}_J^{(2)}(\mathsf{NA}_2)) & \\
        \hspace{1cm}+\,\pi(1-\pi)\cdot\frac{\tilde{d}_P}{\tilde{d}_P+\tilde{d}_J^{(2)}(\mathsf{NA}_1)}((1+f)q_S-\zeta p\cdot\delta_R(q_S,\tilde{d}_P+\tilde{d}_J^{(1)}(\mathsf{NA}_1)) \\
        \hspace{1cm}+\,(1-\pi)^2((1+f)q_S-\zeta p\cdot\delta_R(q_S,\tilde{d}_P)) & \omega=\mathsf{IB} \\[0.25\baselineskip]
        \left[\frac{\pi^2 \tilde{d}_P}{\tilde{d}_P+\tilde{d}_J}(pq_R+\delta_S(q_R,\tilde{d}_P+\tilde{d}_J^{(1)}(\mathsf{A}_2)+\tilde{d}_J^{(2)}(\mathsf{A}_1)))\right. \\
        \hspace{1cm}+\,\pi(1-\pi)\cdot\frac{\tilde{d}_P}{\tilde{d}_P+\tilde{d}_J^{(1)}(\mathsf{NA}_2)}(pq_R+\delta_S(q_R,\tilde{d}_P+\tilde{d}_J^{(1)}(\mathsf{NA}_2))) \\
        \hspace{1cm}+\,\pi(1-\pi)\cdot\frac{\tilde{d}_P}{\tilde{d}_P+\tilde{d}_J^{(2)}(\mathsf{NA}_1)}(pq_R+\delta_S(q_R,\tilde{d}_P+\tilde{d}_J^{(2)}(\mathsf{NA}_1))) \\
        \hspace{1cm}\left.+\,(1-\pi)(pq_R+\delta_S(q_R,\tilde{d}_P))\right]f & \omega=\mathsf{US} \\[0.25\baselineskip]
        \left[\frac{\pi \tilde{d}_P}{\tilde{d}_P+\tilde{d}_J}(q_S+p\cdot\delta_R(q_S,\tilde{d}_P+\tilde{d}_J))\right. \\
        \hspace{1cm}+\,\pi(1-\pi)\cdot\frac{\tilde{d}_P}{\tilde{d}_P+\tilde{d}_J^{(2)}(\mathsf{NA}_1)}(q_S+p\cdot\delta_R(q_S,\tilde{d}_P+\tilde{d}_J^{(2)}(\mathsf{NA}_1))) \\
        \hspace{1cm}+\,\pi(1-\pi)\cdot\frac{\tilde{d}_P}{\tilde{d}_P+\tilde{d}_J^{(2)}(\mathsf{NA}_1)}(q_S+p\cdot\delta_R(q_S,\tilde{d}_P+\tilde{d}_J^{(2)}(\mathsf{NA}_1))) \\
        \hspace{1cm}\left.+\,(1-\pi )(q_S+p\cdot \delta_R(q_S,\tilde{d}_P))\right]f & \omega=\mathsf{UB}.
    \end{dcases}
\end{gather*}

\subsection{Proof of Propositions \ref{thm:subgame-eq-comp} and \ref{thm:passive-br-comp}}

It suffices to show the following: let $\tilde{d}_{P}\in[0,\tilde{e}_{P}]$ and $\zeta_U>[\underline{\zeta},\overline{\zeta}]$. Then there exists a non-trivial Nash equilibrium in the subgame between the traders and JIT LP. Define the following:
\begin{gather*}
    \tilde{\mu}_I = \zeta^{1/2}(1+f)^{-1/2}-1, \\[0.25\baselineskip]
    \hat{\nu}(\tilde{d}_P) = \begin{dcases}
        \frac{\tilde{e}_J}{\tilde{d}_P} & \tilde{d}_P > 0 \\
        0 & \tilde{d}_P = 0
    \end{dcases} \\[0.25\baselineskip]
    \tilde{\mu}_C(\pi ) = \arg\min_{\mu\in\R_+}\,\left|\frac{(1-\pi)^2}{(1+\mu)^2}+\frac{\pi(1-\pi)(2+\mu)\sqrt{(1+f)(1+\mu)}}{(1+\mu)^2}+\frac{\pi^2(1+2\cdot\hat{\nu}(\tilde{d}_P))^2}{(1+2\cdot\hat{\nu}(\tilde{d}_P)+\mu)^2}-\frac{1+f}{\zeta_U}\right|\!, \\[0.25\baselineskip]
    \tilde{\nu}_C(\pi ) = \frac{f(1+\tilde{\mu}_C(\pi))+\tilde{\mu}_C(\pi)\sqrt{(1+f)(1+\tilde{\mu}_C(\pi))}}{\tilde{\mu}_C(\pi)-f}.
\end{gather*}
The equilibrium outcome is
\begin{gather*}
    (q_R,q_S)^\star(\tilde{d}_P;\omega) = \begin{dcases}
        (\tilde{\mu}_I\tilde{d}_P,0) & \omega=\mathsf{IS} \\
        (0,\tilde{\mu}_Ip\tilde{d}_P) & \omega=\mathsf{IB} \\
        (\tilde{\mu}(\pi)\cdot \tilde{d}_P,0) & \omega=\mathsf{US} \\
        (0,\tilde{\mu}(\pi)\cdot p\tilde{d}_P) & \omega=\mathsf{UB}
    \end{dcases}, \\[0.25\baselineskip]
    \tilde{d}_J^{(j)\star}(\tilde{d}_P,(q_R,q_S)^\star(\tilde{d}_P;\omega);(\omega,\omega_j)) = \begin{dcases}
        0 & (\omega,\omega_j)\in\{\mathsf{IS},\mathsf{IB}\}\times\{\mathsf{NA}_{-j},\mathsf{A}_{-j}\} \\
        \tilde{\nu}(\pi)\cdot\tilde{d}_P & (\omega,\omega_j)\in\{\mathsf{US},\mathsf{UB}\}\times\{\mathsf{NA}_{-j}\} \\
        \hat{\nu}(\tilde{d}_P)\cdot\tilde{d}_P & (\omega,\omega_j)\in\{\mathsf{US},\mathsf{UB}\}\times\{\mathsf{A}_{-j}\}
    \end{dcases}.
\end{gather*}
The multiples $\mu_I$, $\mu_C(\pi)$, and $\nu_C(\pi)$ are given by scaling $\tilde{\mu}$, $\tilde{\mu}_C(\pi)$, and $\tilde{\nu}_C(\pi)$ by $p^{1/2}$, respectively. It remains to find the passive LPs' best-responses when the remaining agents play the unique non-trivial Nash equilibrium of the trader--JIT LP subgame given $\tilde{d}_{P}$, the amount of passive liquidity provided, denoted $\sigma_{-P}^\star(\tilde{d}_{P})$. Substituting the explicit forms of $\delta_S$, $\delta_R$, and $\sigma_{-P}^\star(\tilde{d}_{P})$ into the expressions for $u_{P}(\tilde{d}_{P};\sigma_{-P}^\star(\tilde{d}_{P}))$ yields
\begin{align*}
    u_{P}(\tilde{d}_{P};\sigma_{-P}^\star) = \mathcal{U}(\pi,\tilde{d}_P) = (\alpha\mathcal{C}+(1-\alpha)\cdot\mathcal{R}(\pi,\tilde{d}_P))p\tilde{d}_{P}
\end{align*}
where $\mathcal{C}$ and $\mathcal{R}(\pi)$ are given by
\begin{gather*}
    \mathcal{C} = -\left[\psi\left(1-\frac{1+f}{\zeta}\right)^2+(1-\psi\left(\sqrt{\zeta}-\sqrt{1+f}\right)^2\right], \\[0.25\baselineskip]
    \mathcal{R}(\pi,\tilde{d}_P) = \left[(1-\pi)^2\left(\tilde{\mu}_C(\pi)+\frac{\tilde{\mu}_C(\pi)}{1+\tilde{\mu}(\pi)}\right)+2\pi(1-\pi)\cdot\frac{1}{1+\tilde{\nu}_C(\pi)}\left(\tilde{\mu}_C(\pi)+\frac{(1+\tilde{\nu}_C(\pi))\cdot\tilde{\mu}_C(\pi)}{1+\tilde{\nu}_C(\pi)+\tilde{\mu}_C(\pi)}\right)\right. \\
    \left.+\pi^2\cdot\frac{1}{1+2\cdot\hat{\nu}(\tilde{d}_P)}\left(\tilde{\mu}_C(\pi)+\frac{(1+2\cdot\hat{\nu}(\tilde{d}_P))\cdot\tilde{\mu}_C(\pi)}{1+2\cdot\hat{\nu}(\tilde{d}_P)+\tilde{\mu}_C(\pi)}\right)\right]f.
\end{gather*}

We show both of these results with an argument similar to the proofs of Propositions \ref{thm:subgame-eq} and \ref{thm:passive-br} but using the utilities delineated in the previous subsection. Note that if $\bar{\nu}=e_J/e_P=\tilde{e}_J/\tilde{e}_P$ is sufficiently large such that the JIT LP liquidity constraint does not bind when only one JIT LP arrives for all $\zeta_U\in[\underline{\zeta},\overline{\zeta}]$, then the liquidity constraint does not bind for all $\hat{\nu}(\tilde{d}_P)$ where $\tilde{d}_P\in[0,\tilde{e}_P]$ since $\hat{\nu}(\tilde{d}_P)\geq\hat{\nu}(\tilde{e}_P)=\bar{\nu}$. We also require the following lemmas.

\begin{lemma}
Let $\tilde{d}_P\in[0,\tilde{e}_P]$. The unique non-trivial Nash equilibrium of the trader--JIT LP subgame when $(\omega,\omega_j)=(\mathsf{US},\mathsf{A}_{-j})$ is given by
\begin{gather*}
    (q_R,q_S)^\star(\tilde{d}_P,\mathsf{US}) = (\tilde{\mu}_C(\pi)\cdot\tilde{d}_P,0) \\[0.25\baselineskip]
    \tilde{d}_J^{(j)^\star}(\tilde{d}_P,(q_R,q_S);(\mathsf{US},\mathsf{A}_{-j})) = \begin{dcases}
        0 & \tilde{d}_P = 0 \\
        \tilde{e}_J & \tilde{d}_P > 0
    \end{dcases}.
\end{gather*}
\begin{proof}
When $\tilde{d}_P=0$, the result follows immediately from Assumption 1. Now fix $\tilde{d}_P>0$, the uninformed trader's swap size $q_R>0$, and the competitor's deposit size $\tilde{d}_J^{(-j)}$. Let $\tilde{D}=\tilde{d}_P+\tilde{d}_J^{(-j)}$. From previous results, we know that 
\begin{align*}
    \tilde{d}_J^{(j)\star} = \frac{f\tilde{D}(\tilde{D}+q_R)+q_R\sqrt{(1+f)\tilde{D}(\tilde{D}+q_R)}}{q_R-f\tilde{D}}.
\end{align*}
If $q_R<f\tilde{D}$, then by previous arguments, the best-response tends to infinity. Note that
\begin{align*}
    (1-2f+f^2)\tilde{D}^2 &\leq (1+f)\tilde{D}^2+(1+f)\tilde{D}q_R \\
    (1-f)\tilde{D} &\leq \sqrt{(1+f)\tilde{D}(\tilde{D}+q_R)} \\
    (1-f)\tilde{D}q_R &\leq q_R\sqrt{(1+f)\tilde{D}(\tilde{D}+q_R)}+2f\tilde{D}^2 \\
    \tilde{D}q_R-f\tilde{D}^2 &\leq f\tilde{D}^2+f\tilde{D}q_R+q_R\sqrt{(1+f)\tilde{D}(\tilde{D}+q_R)} \\
    \tilde{D} &\leq \frac{f\tilde{D}(\tilde{D}+q_R)+q_R\sqrt{(1+f)\tilde{D}(\tilde{D}+q_R)}}{q_R-f\tilde{D}},
\end{align*}
so $\tilde{d}_J^{(j)\star}\geq \tilde{D}$. Then for an interior solution to a Cournot--Nash equilibrium, we must have
\begin{align*}
    \tilde{d}_J^{(1)\star} &\geq \tilde{d}_P+\tilde{d}_J^{(2)\star} \\
    \tilde{d}_J^{(2)\star} &\geq \tilde{d}_P+\tilde{d}_J^{(1)\star}.
\end{align*}
Adding the inequalities implies that $0\geq\tilde{d}_P$, contradicting the assumption that $\tilde{d}_P>0$. 

We finish by checking for corner cases. If the competitor does not deposit, then JIT LP $j$'s optimal deposit amount is the monopolist's optimum, which is positive, so there is no Nash equilibrium where one or both JIT LPs does not deposit. 

Suppose that the competitor's liquidity constraint binds i.e.\!\, $\tilde{d}_J^{(-j)}=\tilde{e}_J$. Then we must have $\tilde{d}_J^{(j)\star}\geq\tilde{D}=\tilde{q}_P+\tilde{e}_J>\tilde{e}_J$. Since a JIT LP's utility function is quasiconcave in the deposit size, it follows that JIT LP $j$'s utility increasing in $\tilde{d}_J^{(j)}$ until $\tilde{d}_J^{(j)\star}$. It follows that the optimal deposit size given the liquidity constraint for JIT LP $j$ is $\tilde{e}_J$. Reversing the roles of JIT LPs $j$ and $-j$ yields a unique symmetric Cournot--Nash equilibrium where both JIT LPs depsoit their entire endowment.

We now focus on the trader's problem. We claim that if the JIT LP's utility function has a unique maximum in $\tilde{d}_J$ on $\R_+$ and $q_R^\star$ solves the uninformed trader's problem given $\tilde{d}_P$ and $\hat{\nu}(d_P)$ when $\omega=\mathsf{US}$, then $q_R^\star/c$ solves the trader's problem given $\tilde{d}_P/c$ and $c\cdot\hat{\nu}(\tilde{d}_P/c)$ when $\omega=\mathsf{US}$ for any $c>0$.

The trader's problem given $\tilde{d}_{P}$ is
\begin{align*}
    \max_{q_R\in\R_+} \ & (1-\pi)^2\cdot\delta_S(q_R,\tilde{d}_P) +2\pi(1-\pi)\cdot\delta_S(q_R,\tilde{d}_{P}+\tilde{d}_{J}^{(j)}(\tilde{d}_{P},(q_R,0);(\mathsf{US},\mathsf{NA}_{-j}))) \\
    & \hspace{1cm}+\pi^2\cdot\delta_S(q_R,(1+2\cdot\hat{\nu}(\tilde{d}_P))\cdot\tilde{d}_P)-\zeta_U^{-1}p(1+f)q_R
\end{align*}
and the trader's problem given $\tilde{d}_{P}/c$ and $c\cdot\hat{\nu}(\tilde{d}_P)$ is
\begin{align*}
    \max_{q_R\in\R_+} \ & (1-\pi)^2\cdot\delta_S(q_R,\tilde{d}_P/c) +2\pi(1-\pi)\cdot\delta_S(q_R,\tilde{d}_{P}/c+\tilde{d}_{J}^{(j)}(\tilde{d}_{P}/c,(q_R,0);(\mathsf{US},\mathsf{NA}_{-j}))) \\
    & \hspace{1cm}+\pi^2c\cdot\delta_S(q_R,(1+2c\cdot\hat{\nu}(\tilde{d}_P/c))\cdot\tilde{d}_P/c)-\zeta_U^{-1}p(1+f)q_R
\end{align*}
Let $\tilde{q}_R=cq_R$. By Lemma A.5, the unique maximum assumption, and the 1-homogeneity of $\delta_S$, this problem is equivalent to
\begin{align*}
    \max_{\tilde{q}_R\in\R_+} \ & c^{-1}(1-\pi)^2\cdot\delta_S(\tilde{q}_R,\tilde{d}_P)+ c^{-1}2\pi(1-\pi)\cdot\delta_S(\tilde{q}_R,\tilde{d}_{P}+\tilde{d}_{J}(\tilde{d}_{P},(\tilde{q}_R,0);(\mathsf{US},\mathsf{NA}_{-j}))) \\
    \hspace{1cm}&+c^{-1}\pi^2\cdot\delta_S(q_R,(1+2\cdot\hat{\nu}(\tilde{d}_P))\cdot\tilde{d}_P)-c^{-1}\zeta_U^{-1}p(1+f)\tilde{q}_R
\end{align*}
The objective is $c^{-1}$ times the original objective, so the claim follows.

The trader's utility function when $\omega=\mathsf{US}$ is
\begin{align*}
    u_T((q_R,0);\sigma_{-T},\mathsf{US}) &= (1-\pi )^2\frac{p\tilde{d}_{P}q_R}{\tilde{d}_{P}+q_R}+2\pi(1-\pi) \cdot\frac{p(\tilde{d}_{P}+\tilde{d}_{J}^{(j)\star}(\tilde{d}_{P},(q_R,0);(\mathsf{US},\mathsf{NA}_{-j}))q_R}{\tilde{d}_{P}+\tilde{d}_{J}^{(j)\star}(\tilde{d}_{P},(q_R,0);(\mathsf{US},\mathsf{NA}_{-j}))+q_R} \\
    &\hspace{1cm}+\pi^2\cdot\frac{p(\tilde{d}_{P}+2\cdot\tilde{d}_{J}^{(j)\star}(\tilde{d}_{P},(q_R,0);(\mathsf{US},\mathsf{A}_{-j}))q_R}{\tilde{d}_{P}+2\cdot\tilde{d}_{J}^{(j)\star}(\tilde{d}_{P},(q_R,0);(\mathsf{US},\mathsf{A}_{-j}))+q_R}-\zeta_U^{-1}p(1+f)q_R.
\end{align*}
This has partial derivative
\begin{align*}
    \frac{\d u_T}{\d q_R} &= (1-\pi )^2p\cdot\frac{\tilde{d}_{P}^2}{(\tilde{d}_{P}+q_R)^2}+2\pi(1-\pi)  p\cdot\frac{(2\tilde{d}_{P}+q_R)\sqrt{(1+f)\tilde{d}_{P}(\tilde{d}_{P}+q_R)}}{2(\tilde{d}_{P}+q_R)^2} \\
    &\hspace{1cm}+\pi^2\cdot\frac{(\tilde{d}_P+2\tilde{e}_J)^2}{(\tilde{d}_P+2\tilde{e}_J+q_R)^2} -\frac{p(1+f)}{\zeta_U} \\[0.25\baselineskip]
    &= \left[(1-\pi )^2\cdot\frac{\tilde{d}_{P}^2}{(\tilde{d}_{P}+q_R)^2}+2\pi(1-\pi)\cdot\frac{(2\tilde{d}_{P}+q_R)\sqrt{(1+f)\tilde{d}_{P}(\tilde{d}_{P}+q_R)}}{2(\tilde{d}_{P}+q_R)^2}\right. \\
    &\hspace{1cm}\left.+\pi^2\cdot\frac{(\tilde{d}_P+2\tilde{e}_J)^2}{(\tilde{d}_P+2\tilde{e}_J+q_R)^2}-\frac{1+f}{\zeta_U}\right]p.
\end{align*}
By Lemma A.17, we can normalize $\tilde{d}_P$ to 1, so the trader's first order condition is
\begin{align*}
    M_{TC}(\mu) \equiv \frac{(1-\pi)^2}{(1+\mu)^2}+\frac{\pi(1-\pi)(2+\mu)\sqrt{(1+f)(1+\mu)}}{(1+\mu)^2}+\frac{\pi^2(1+2\hat{\nu})^2}{(1+2\hat{\nu}+\mu)^2} = \frac{1+f}{\zeta_U}
\end{align*}
where $\hat{\nu}\equiv\hat{\nu}(\tilde{d}_P)$.
\end{proof}
\end{lemma}

\begin{lemma}
Let $\tilde{d}_P\in[0,\tilde{e}_P]$. The unique non-trivial Nash equilibrium of the trader--JIT LP subgame when $(\omega,\omega_j)=(\mathsf{UB},\mathsf{A}_{-j})$ is given by
\begin{gather*}
    (q_R,q_S)^\star(\tilde{d}_P,\mathsf{UB}) = (0,p\cdot\tilde{\mu}_C(\pi)\cdot\tilde{d}_P) \\[0.25\baselineskip]
    \tilde{d}_J^{(j)^\star}(\tilde{d}_P,(q_R,q_S);(\mathsf{US},\mathsf{A}_{-j})) = \begin{dcases}
        0 & \tilde{d}_P = 0 \\
        \tilde{e}_J & \tilde{d}_P > 0
    \end{dcases}.
\end{gather*}
\end{lemma}
\begin{proof}
The proof of Lemma A.18 is symmetric to the proof of Lemma A.17.
\end{proof}

We now characterize the best-response of the passive LPs. Suppose that $k<N$, $\mathcal{U}(\pi,d_P(k))\geq0$, and $\mathcal{U}(\pi,d_P(k+1))<0$. Let $\mathbf{d}_P$ be a vector of passive LP deposit amounts such that $1^\top\mathbf{d}_P=d_P(k)=ke_P/N$. Then it must be that $k$ passive LPs chose to provide liquidity while the remaining did not. If a passive LP that chose to provide liquidity wishes to change their strategy, then they would go from nonnegative utility to zero utility (as a result of not participating; recall Assumption 1), so they have no profitable deviation. If a passive LP that chose to not provide liquidity wishes to change their strategy, then they would go from zero utility to negative utility, so they have no profitable deviation. Thus $\mathbf{d}_P$ corresponds to an equilibrium strategy for the passive LPs that yields a total of $d_P(k)$ units of passive liquidity provided.

Now suppose that $\mathbf{d}_P$ is an equilibrium strategy for the passive LPs that yields a total of $d_P(k)$ units of passive liquidity provided. Then it must be that $k$ passive LPs chose to provide liquidity while the remaining did not. Each passive LP who chose to provide liquidity has no profitable deviation if their current utility is nonnegative (since not depositing yields zero utility; Assumption 1). Each passive LPs who chose to not provide liquidity has no profitable deviation if their utility from depositing (given the contribution of $k$ other passive LPs) is negative. We must then have $\mathcal{U}(\pi,d_P(k))\geq0$ and $\mathcal{U}(\pi,d_P(k+1))<0$.

The case when $k=N$ is proved similarly.

\subsection{Proof of Theorem \ref{thm:threshold-comp}}

The passive LPs' per-unit fee revenue scaled by $f^{-1}$ in the complete absence of a JIT LP ($\pi=0$) is the total trading volume attracted by the DEX when an uninformed trader arrives when $\pi=0$, which is given by
\begin{gather*}
    V_0 \equiv \tilde{\mu}(0)+\frac{\tilde{\mu}(0)}{1+\tilde{\mu}(0)} = \left(\frac{\zeta_U}{1+f}\right)^{1/2}-\left(\frac{1+f}{\zeta_U}\right)^{1/2} = \frac{\mathcal{R}(0)}{f}.
\end{gather*}
Fix $\tilde{d}_/P$ and let $\hat{\nu}\equiv\hat{\nu}(\tilde{d}_P)=\tilde{e}_J/\tilde{d}_P$. Then passive LPs' per-unit fee revenue scaled by $f^{-1}$ given a JIT LP arrival probability of $\pi$ is the total trading volume attracted by the AMM when an uninformed trader arrives given a JIT LP arrival probability of $\pi$, which can be expressed as a function of the equilibrium uninformed trade size $\mu\equiv\tilde{\mu}(\pi)$:
\begin{align*}
    V_C(\mu) &\equiv (1-\pi)^2\left(\mu+\frac{\mu}{1+\mu}\right)+2\pi(1-\pi)\cdot\frac{1}{1+\nu(\mu)}\left(\mu+\frac{(1+\nu(\mu))\mu}{1+\nu(\mu)+\mu}\right) \\
    &\hspace{1cm}+\pi^2\cdot\frac{1}{1+2\hat{\nu}}\left(\mu+\frac{(1+2\hat{\nu})\mu}{1+\hat{\nu}+\mu}\right) \\[0.25\baselineskip]
    &= (1-\pi)^2\left(\mu+\frac{\mu}{1+\mu}\right)+2\pi(1-\pi)\left[\left(\frac{1+\mu}{1+f}\right)^{1/2}-\left(\frac{1+f}{1+\mu}\right)^{1/2}\right] \\
    &\hspace{1cm}+\pi^2\left(\frac{\mu}{1+2\hat{\nu}}+\frac{\mu}{1+2\hat{\nu}+\mu}\right) = \frac{\mathcal{R}(\pi)}{f}.
\end{align*}
Here $\nu(\mu)$ is the JIT LP's normalized deposit size written as a function of the uninformed trader's normalized swap size:
\begin{align*}
    \nu(\mu) = \frac{f(1+\mu)+\mu\sqrt{(1+f)(1+\mu)}}{\mu-f}.
\end{align*}
We require the following lemma.
\begin{lemma}
Fix $f\in\R_+$, $\pi\in[0,1]$, and $\hat{\nu}\in(1,\infty)$. Then $M_{TC}(\mu)\cdot(2+V_C(\mu)^2+V_C(\mu)\sqrt{4+V_C(\mu)^2})$ is increasing in $\mu$ on some interval $\mu\in I$ where $I\to(f,\infty)$ as $\bar{\nu}\to\infty$.
\end{lemma}
\begin{proof}
Recall that $\hat{\nu}\equiv\hat{\nu}(\tilde{d}_P)\geq\bar{\nu}$. It suffices to show that
\begin{align*}
    \frac{\d\log(2+V_C(\mu)^2+V_C(\mu)\sqrt{4+V_C(\mu)^2})}{\d\mu} &\geq -\frac{\d\log M_{TC}(\mu)}{\d\mu} \\[0.25\baselineskip]
    \frac{2\cdot V_C'(\mu)}{\sqrt{4+V_C(\mu)^2}} &\geq \frac{-M_{TC}'(\mu)}{M_{TC}(\mu)}.
\end{align*}
Let $(\Omega,\F,\P)$ where $\Omega=\{\omega_0,\omega_1\,\omega_2\}$, $\F=2^\Omega$, $\P(\omega_0)=(1-\pi)^2$, $\P(\omega_1)=2\pi(1-\pi)$, and $\P(\omega_2)=\pi^2$ be a probability space. Define random variables $\Phi,\tilde{\Phi}:\Omega\to\R$ such that
\begin{gather*}
    \Phi(\omega) = \begin{dcases}
        \frac{1}{(1+\mu)^2} & \omega=\omega_0 \\[0.25\baselineskip]
        \frac{(2+\mu)\sqrt{(1+f)(1+\mu)}}{2(1+\mu)^2} & \omega=\omega_1 \\[0.25\baselineskip]
        \frac{(1+2\hat{\nu})^2}{(1+2\hat{\nu}+\mu)^2} & \omega=\omega_2
    \end{dcases} \\[0.25\baselineskip]
    \tilde{\Phi}(\omega) = \begin{dcases}
        \frac{2}{(1+\mu)^3} & \omega=\omega_0 \\[0.25\baselineskip]
        \frac{(4+\mu)\sqrt{(1+f)(1+\mu)}}{4(1+\mu)^2} & \omega=\omega_1 \\[0.25\baselineskip]
        \frac{2(1+2\hat{\nu})^2}{(1+2\hat{\nu}+\mu)^3} & \omega=\omega_2 
    \end{dcases}.
\end{gather*}
Then we have
\begin{align*}
    \frac{\tilde{\Phi}(\omega)}{\Phi(\omega)} = \begin{dcases}
        \frac{2}{1+\mu} & \omega=\omega_0 \\[0.25\baselineskip]
        \frac{2}{(1+\mu)(2+\mu)} & \omega=\omega_1 \\[0.25\baselineskip]
        \frac{2}{1+2\hat{\nu}+\mu} & \omega=\omega_2
    \end{dcases}.
\end{align*}
Note that as $\hat{\nu}\to\infty$, $\frac{\tilde{\Phi}(\omega_0)}{\Phi(\omega_0)} \geq \frac{\tilde{\Phi}(\omega_1)}{\Phi(\omega_1)}\geq\frac{\tilde{\Phi}(\omega_2)}{\Phi(\omega_2)}$ while $\Phi(\omega_0) \leq \Phi(\omega_1)\leq\Phi(\omega_2)$ on an interval $I$ that converges to $(f,\infty)$ as $\hat{\nu}\to\infty$, so it follows that 
\begin{align*}
    \Cov\left[\frac{\tilde{\Phi}}{\Phi},\Phi\right]&\leq0 \\[0.25\baselineskip]
    \frac{\E[\tilde{\Phi}]}{\E[\Phi]} &\leq \E\left[\frac{\tilde{\Phi}}{\Phi}\right].
\end{align*}
Define random variables $\Psi,\tilde{\Psi}:\Omega\to\R$ (on the same probability space) such that
\begin{gather*}
    \Psi(\omega) = \begin{dcases}
        \sqrt{4+\left(\mu+\frac{\mu}{1+\mu}\right)^2} = 1+\mu+\frac{1}{1+\mu} & \omega=\omega_0 \\[0.25\baselineskip]
        \sqrt{4+\left(\sqrt{\frac{1+\mu}{1+f}}-\sqrt{\frac{1+f}{1+\mu}}\right)^2} = \sqrt{\frac{1+\mu}{1+f}}+\sqrt{\frac{1+f}{1+\mu}} & \omega=\omega_1 \\[0.25\baselineskip]
        \sqrt{4+\left(\frac{\mu}{1+2\hat{\nu}}+\frac{\mu}{1+2\hat{\nu}+\mu}\right)^2} = 2+\frac{\mu^2}{(1+2\hat{\nu})(1+2\hat{\nu}+\mu)} & \omega=\omega_2
    \end{dcases} \\[0.25\baselineskip]
    \tilde{\Psi}(\omega) = \begin{dcases}
        2\left(1+\frac{1}{(1+\mu)^2}\right) & \omega=\omega_0 \\[0.25\baselineskip]
        \frac{1}{1+\mu}\left(\sqrt{\frac{1+\mu}{1+f}}+\sqrt{\frac{1+f}{1+\mu}}\right) & \omega=\omega_1 \\[0.25\baselineskip]
        2\left(\frac{1}{1+2\hat{\nu}}+\frac{1+2\hat{\nu}}{(1+2\hat{\nu}+\mu)^2}\right) & \omega=\omega_2
    \end{dcases}.
\end{gather*}
Then we have
\begin{align*}
    \frac{\tilde{\Psi}(\omega)}{\Psi(\omega)} = \begin{dcases}
        \frac{2}{1+\mu} & \omega=\omega_0 \\[0.25\baselineskip]
        \frac{1}{1+\mu} & \omega=\omega_1 \\[0.25\baselineskip]
        \frac{2}{1+2\hat{\nu}+\mu} & \omega=\omega_2
    \end{dcases}.
\end{align*}
Note that as $\hat{\nu}\to\infty$, $\frac{\tilde{\Psi}(\omega_0)}{\Psi(\omega_0)}\geq\frac{\Tilde{\Psi}(\omega_1)}{\tilde{\Psi}(\omega_1)}\geq\frac{\Tilde{\Psi}(\omega_2)}{\tilde{\Psi}(\omega_2)}$ and $\Psi(\omega_0)\geq\Psi(\omega_1)\geq\Psi(\omega_2)$ on an interval $I$ that converges to $(f,\infty)$ as $\hat{\nu}\to\infty$, so it follows that
\begin{align*}
    \Cov\left[\frac{\Tilde{\Psi}}{\Psi},\Psi\right] &\geq 0 \\[0.25\baselineskip]
    \frac{\E[\tilde{\Psi}]}{\E[\Psi]} &\geq \E\left[\frac{\tilde{\Psi}}{\Psi}\right].
\end{align*}
Observe that $\E[\tilde{\Phi}/\Phi] \leq \E[\tilde{\Psi}/\Psi]$. Chaining everything together yields
\begin{align*}
    \frac{2\cdot V_C'(\mu)}{\sqrt{4+V_C(\mu)^2}} \geq \frac{\E[\tilde{\Psi}]}{\E[\Psi]} &\geq \E\left[\frac{\tilde{\Psi}}{\Psi}\right] \geq \E\left[\frac{\tilde{\Phi}}{\Phi}\right] \geq \frac{\E[\tilde{\Phi}]}{\E[\Phi]} = \frac{-F'_c(\mu)}{M_{TC}(\mu)}
\end{align*}
as desired, noting that $\E[\Psi]\geq\sqrt{4+V_C(\mu)^2}$ due to Jensen's inequality.
\end{proof}

Fix $\mu\in(f,\infty)$. Let $\zeta_{UC}(\mu)$ be the private value shock size such that the uninformed trader's equilibrium swap size is $\mu$. The first-order condition implies that
\begin{align*}
    \zeta_{UC}(\mu) = \frac{1+f}{M_{TC}(\mu)},
\end{align*}
so $\zeta_{UC}(\mu)$ is well-defined. Since $V_0$ is increasing in $\zeta_U$ and $\lim_{\zeta_U\to\infty}V_0=\infty$, there exists a unique value of $\zeta_U\in(1+f,\infty)$ such that $V_C(\mu)=V_0$ under $\zeta_U$; let us denote it $\bar{\zeta}_{UC}(\mu)$. Then $V_C(\mu)\geq V_0$ if and only if $\zeta_{UC}(\mu)\leq\bar{\zeta}_{UC}(\mu)$. The expression for $V_0$ yields
\begin{gather*}
    \bar{\zeta}_U(\mu) = (1+f)\left(1+\frac{V_C(\mu)^2+V_C(\mu)\cdot\sqrt{4+V_C(\mu)^2}}{2}\right).
\end{gather*}
Note that
\begin{align*}
    2 &\leq M_{TC}(\mu)\cdot\left(2+V_C^2(\mu)+V_C(\mu)\cdot\sqrt{4+V_C^2(\mu)}\right) \\[0.25\baselineskip]
    \frac{1}{M_{TC}(\mu)} &\leq \frac{2+V_C^2(\mu)+V_C(\mu)\cdot\sqrt{4+V_C^2(\mu)}}{2} \\[0.25\baselineskip]
    \zeta_{UC}(\mu) &\leq \bar{\zeta}_{UC}(\mu).
\end{align*}
We now have two cases:
\begin{itemize}
    \item Let $\underline{\mu}$ be the equilibrium trade size when $\zeta_U=\underline{\zeta}$. If $M_{TC}(\underline{\mu})\cdot\left(2+V_C^2(\underline{\mu})+V_C(\underline{\mu})\cdot\sqrt{4+V_C^2(\underline{\mu})}\right)\geq 2$, then by Lemma A.19, we have $\zeta_{UC}(\mu)\leq\bar{\zeta}_{UC}(\mu)$, corresponding to the first case of the theorem.
    \item Note that $M_{TC}(\mu)=O(\mu^{-1/2})$ and $V_C(\mu)=\Omega(\mu^{1/2})$, so it follows that
    \begin{align*}
        \lim_{\mu\to\infty}M_{TC}(f)\cdot\left(2+V_C^2(f)+V_C(f)\cdot\sqrt{4+V_C^2(f)}\right)=\infty.
    \end{align*}
    If $M_{TC}(f)\cdot\left(2+V_C^2(f)+V_C(f)\cdot\sqrt{4+V_C^2(f)}\right)< 2$, then by Lemma A.19 and the above asymptotic analysis, there exists a unique $\mu^\star\in(f,\infty)$ such that $\zeta_U(\mu^\star)=\bar{\zeta}_U(\mu^\star)$. 
    \begin{itemize}
        \item If $\mu<\mu^\star$, then $\zeta_{UC}(\mu)$ $>\bar{\zeta}_{UC}(\mu)$, so $V_C(\mu)<V_0$ under $\zeta_U=\zeta_{UC}(\mu)$: for small trade sizes and thus small shock sizes, we have crowding out.
        \item If $\mu\geq\mu^\star$, then $\zeta_{UC}(\mu)$ $\leq\bar{\zeta}_{UC}(\mu)$, so $V_C(\mu)\geq V_0$ under $\zeta_U=\zeta_{UC}(\mu)$: for large trade sizes and thus large shock sizes, we have complementing.
    \end{itemize}  
    This corresponds to the second and third cases of the theorem, depending on the values of $\underline{\zeta}$ and $\overline{\zeta}$.
\end{itemize}

It remains to compare the thresholds in the monopolist and competitive settings. Note that $V(\mu)\geq V_C(\mu)$ for all $\mu\in I$ where $I\to(f,\infty)$ as $\hat{\nu}\to\infty$, so it follows that $\Bar{\zeta}_U(\mu)\geq\bar{\zeta}_{UC}(\mu)$ for $\hat{\nu}$ sufficiently large. Since $M_T(\mu)\leq M_{TC}(\mu)$, we have $\zeta_U(\mu)\leq\zeta_{UC}(\mu)$. It follows that if $\zeta_U(\mu)\geq\bar{\zeta}_U(\mu)$, then $\zeta_{UC}\geq$ $\bar{\zeta}_{UC}(\mu)$. Thus if the monopolist JIT LP complements the passive LPs under $\zeta_U$, then the competing JIT LPs complement the passive LPs under $\zeta_U$ as well. Since there exists at most one threshold for $\zeta_U$ in both the monopolist and competitive cases where the regime switches from crowding out to complementing, it follows that $\zeta^\star_C(f,\pi,d_P)\leq\zeta^\star(f,\pi)$ for all $d_P\in[0,e_P]$.

\end{document}